\documentclass[12pt]{article}
\usepackage[utf8]{inputenc}

\usepackage[suppress]{color-edits}
\addauthor{nw}{purple}
\addauthor{tg}{red}
\usepackage{xcolor}
\usepackage{graphicx}
\usepackage[margin=1in]{geometry}
\usepackage[onehalfspacing]{setspace}
\usepackage{hyperref}

\usepackage{tikz}
\usepackage{caption}
\usepackage{subcaption}
\usetikzlibrary{decorations.markings}

\usepackage{amsfonts}
\usepackage{amsmath}
\usepackage{amssymb}
\usepackage{amsthm}
\usepackage{bbm}
\usepackage{natbib}
\usepackage{xcolor}

\newtheorem{theorem}{Theorem}

\newtheorem{definition}{Definition}
\newtheorem{corollary}{Corollary}
\newtheorem{lemma}{Lemma}

\newtheorem{proposition}{Proposition}

\newcommand{\conv}{\textnormal{conv}}

\newcommand{\dd}[1]{\mathrm{d}#1}
\newcommand{\ICIR}{\textnormal{IC-IR}}

\usepackage{accents}
\newcommand{\ubar}[1]{\underaccent{\bar}{#1}}
\DeclareMathOperator*{\argmax}{arg\,max}

\title{From Doubt to Devotion:\\ Trials and Learning-Based Pricing\thanks{ We would like to thank Dirk Bergemann, Simon Board, Alessandro Bonatti, Thomas Brzustowski, Daniel Clark, 
Piotr Dworczak, Brett Green, Marina Halac, Kevin He, Johannes H\"orner, Ryota Iijima, Jonathan Libgober, Bart Lipman, Elliot Lipnowski, Erik Madsen, Alessandro Pavan, Daniel Rappoport,  Philip J. Reny, Larry Samuelson, Anna Sanktjohanser, Ludvig Sinander, Philipp Strack, Harry Pei, Kai Hao Yang, Jidong Zhou, and participants at the EC 2024, 20th IO Theory Conference, SITE '23, and all other seminars and conferences for helpful comments and discussions.} 
}
\author{Tan Gan\thanks{Department of Management, London School of Economics, t.gan2@lse.ac.uk.} \and Nicholas Wu \thanks{Department of Economics, Yale University, nick.wu@yale.edu. }}
\date{\today}

\begin{document}

    \maketitle
\begin{abstract}
An informed seller designs a dynamic mechanism to sell an experience good. The seller has private information about the product match, which affects the buyer's private consumption experience. The belief gap between both parties, coupled with the buyer's learning, yields mechanisms providing the skeptical buyer with limited access to the product and an option to upgrade if the buyer is swayed by a good experience. Depending on the seller's screening technology, this takes the form of free/discounted trials or dynamic tiered pricing, which are prevalent in digital markets. Unlike static environments, having consumer data can reduce sellers' revenue in equilibrium.
\vspace{0.5cm}

\end{abstract}
\newpage

\section{Introduction}

In the digital marketplace, a prominent feature distinguishing software products and digital services from traditional retail goods is the implementation of dynamic pricing mechanisms that are fine-tuned by data. For example, consider an enterprise software (SaaS) company, offering a small business (a potential client) access to marketing software. The small business sees boosts in their sales if the marketing software fits their needs. On one hand, the SaaS company is better informed about whether its service can benefit the client based on proprietary knowledge about their software and data from similar clients. On the other hand, while the client is initially uninformed about the service, they can become convinced of the service's value by privately observing the service's impact on their sales.  How does a seller leverage its own proprietary data together with the buyer's ability to learn, and what are the welfare implications in equilibrium?

We address this question by studying a dynamic mechanism design problem of an informed principal. We believe the informed principal approach is particularly relevant for digital markets, where sellers can commit to pricing mechanisms but cannot credibly communicate information obtained from unpublished data.
In our baseline model, a seller is privately informed about a binary match value between the service and the buyer. The informed seller commits to a dynamic mechanism that determines the service access provided and prices charged over time. If the match value is high, the buyer privately receives positive utility shocks from the service via a Poisson process with an arrival rate depending on the service access; if the match value is low, the buyer experiences no utility shocks. The size of the utility shocks is initially unknown to both parties but is learned perfectly by the buyer once they receives a shock. 
In the software example, we interpret the payoffs as the small business experiencing jumps in their sales as a result of successful marketing.

We believe that the key features of our model (dynamic pricing, the informed seller, and the buyer learning) are important determinants for how software-as-a-service (SaaS) products are sold. As one such SaaS company, Stripe, write in their description of ``What makes SaaS pricing models different from other pricing models?'' in their guide to pricing, \footnote{See \href{https://stripe.com/resources/more/saas-pricing-models-101}{https://stripe.com/resources/more/saas-pricing-models-101}.}
\begin{quote}
    With a SaaS model, companies have access to a wealth of customer data that can inform pricing decisions... This data allows for a nuanced pricing strategy that aligns with actual usage patterns. This is rarely possible in traditional sales models, where the relationship with the customer often ends after the initial sale.
\end{quote}
In particular, the Stripe guide highlights how both data-driven pricing strategy and dynamic pricing flexibility have dramatically changed how software is sold.

Central to our model is the interplay between dynamic price discrimination against buyer learning and the belief gap regarding the buyer's experience between the informed seller and the uninformed buyer. In contrast to standard dynamic mechanism design, the buyer's belief process is no longer a martingale from the informed designer's point of view. This introduces new tradeoffs into the mechanism design problem. 
To see this more clearly, note that in our framework, an ex-ante revenue-maximizing mechanism is to always provide full access to the service and charge a price at the beginning. This deters the buyer from leveraging their private consumption experience to gain information rent. Such a mechanism can be supported in equilibrium if either of our two key model components - dynamic pricing and the belief gap regarding learning - is removed. However, in our model, because the seller is privately informed, the seller with high match quality anticipates a greater likelihood of the buyer receiving a utility shock relative to the buyer's own expectation. When the belief gap is large enough to offset the information rent lost when the buyer learns their own value (the size of the utility shock), it becomes profitable for the high-type seller to deviate by offering a \textit{free trial} to ``prove'' themselves and then selling the service post-trial at a higher price. Faced with this deviation, the buyer uses the free trial regardless of their belief about match quality, and purchases the post-trial service if swayed by a good experience during the trial. 

Since the ex-ante revenue-maximizing mechanism might not be supported in equilibrium, we establish a more precise characterization of equilibrium payoffs. To this end, we first characterize the payoffs attainable by incentive-compatible dynamic mechanisms, a necessary condition for payoffs to be achieved in equilibrium. Theorem \ref{thm:trial} shows that payoffs on the boundary of the feasible payoff set can be uniquely achieved by trial mechanisms. In such mechanisms, sellers of both types sell full service access up to an intermediate time, after which they set a premium price for the remaining service.  The buyer always purchases the initial trial despite being skeptical about the match quality, and purchases the post-trial service only if they experience a sufficiently large utility shock during the trial. 

Technically, the proof of Theorem \ref{thm:trial} solves for dynamic mechanisms where the designer maximizes a weighted average of the two seller types' payoffs, or equivalently operates with a belief about the seller's type that diverges from the buyer's prior. The canonical technique for solving dynamic design relaxes the problem and only considers impulse responses to the initial agent reports; however, this does not work in our environment because the designer's belief differs from the buyer's, so the designer perceives the buyer as inaccurately predicting the evolution of their own belief. From the perspective of the designer or informed principal, the agent's belief is progressively ``corrected'' through learning. This poses a significant challenge to tractability in general settings. Nonetheless, by considering a relaxed problem that exploits the Poisson conclusive learning structure, we fully characterize optimal mechanisms. The two key constraints to the relaxed problem are: (i) the local incentive constraints on reporting buyer value, assuming the timing of the report is truthful, and (ii) a global intertemporal constraint, which requires that upon receiving a payoff shock, the buyer finds it suboptimal to stay silent about its arrival forever.

Returning to equilibrium payoffs, we consider a weak refinement of perfect Bayesian equilibrium, where we require that the buyer must learn after seeing evidence from their experience and that high-type seller receives at least as much as the low-type seller. We show that all equilibrium payoffs under this refinement are achievable in equilibria where sellers of both types propose the same trial mechanisms. In these pooling equilibria, the sellers' strategic choice of the mechanism and its dynamic allocation rule do not disclose any information, and the buyer only updates their belief based on private signals from the consumption experience. Further, the signaling incentives of sellers have sharp implications when stronger equilibrium refinements are considered.
We show that all equilibrium trial mechanisms surviving the D1 refinement have the same trial length and price threshold as the free trial.\footnote{The only difference between these trial mechanisms is the ex-ante price of the trial. The free trial here refers to the deviation mechanism that we discussed before.} 
Compared to all equilibrium trial mechanisms, the D1-surviving trial mechanisms have the shortest trial length and the highest price threshold, resulting in minimum ex-ante revenue/efficiency. Intuitively, the high-type seller has the incentive to signal by proposing dynamic mechanisms that generate higher profit from selling post-trial services, which the low type can not profit from. Consequently, the surviving trial mechanisms maximize the revenue from post-trial services and minimize the trial length.

The distortion of signaling incentives can be so large that sellers of both types benefit from not possessing private information regarding match quality. This, in conjunction with the previous impossibility result of supporting the ex-ante revenue-maximizing mechanism in equilibrium,\footnote{The previous result has weaker predictions but it does not require the stronger D1 equilibrium refinement.} highlights that having consumer preference data may actually decrease sellers' revenue. These findings contrast with the literature on the informed principal problem in the static environment \citep{KOESSLER2016456}, where the ex-ante profit maximizing mechanism is supported in equilibrium and survives equilibrium refinement.


Trial mechanisms, which are prevalent in digital markets, arise endogenously in our baseline model because we restrict sellers to controlling only service access. This setting is relevant when the seller's ability to customize the service is constrained: for example, when the seller can only offer either full service or no service at all.
In a more general setting, we allow the seller to also alter service quality, so the buyer's learning rate and expected utility flows are not necessarily co-linear. In this case, Theorem \ref{thm:tieredpricing} shows that the relevant equilibrium mechanisms take the form of another commonly seen selling practice: dynamic tiered pricing. 
Streaming platforms offering ad-supported free access exemplify this. Frequent ad interjections decrease the consumption value of the content, reducing the information rent to premium service users while providing skeptical users an avenue to evaluate the platform's content.  Once the buyer becomes convinced of the value of the service, they upgrade to the ad-free premium tier. In contrast to intuition in standard mechanism design, when the seller is privately informed, enriching their screening technology may reduce both revenue and social efficiency in equilibrium.

The paper is structured as follows. The next section surveys the relevant literature. We present the formal model in Section \ref{sec:model}. We characterize the first-best in Section \ref{sec:benchmark}, and show that the first-best is attainable in the absence of the informed seller property or the dynamic nature of the problem. Section \ref{sec:mechanisms} characterizes dynamic mechanisms, and shows that trial mechanisms are revenue-maximizing mechanisms. Section \ref{sec:equilibria} characterizes equilibrium payoffs and discusses refinements. Section \ref{sec:general_screening} generalizes the screening technology and shows that dynamic tiered pricing emerges in this case, and Section \ref{sec:conclusion} concludes.

\subsection{Literature}
This paper contributes to the literature on the informed principal stemming from \citet{RothSchildStiglitz76} and \citet{myerson83}. Due to its growing relevance in the era of big data, the informed principal framework has seen a recent resurgence in interest
\citep{KOESSLER2016456,KoesslerSkreta19,KoessSkreta23,Dosis22,Nishimura22,Clark2023,Clark2023TheIP,BLR23}. Our paper contributes to this literature by considering an environment in which the informed principal can commit to a dynamic mechanism and the buyer learns about their value. The closest paper to ours is \citet{KOESSLER2016456}, which considers a general framework of static informed principal problem, where a seller with a single indivisible good faces a buyer whose willingness to pay depends on taste, privately-known by the buyer, and on product characteristics, privately known by the seller. They prove that the ex-ante revenue-maximizing mechanism can always be implemented in equilibrium; thus, information about buyer preferences always helps the seller in a static environment. However, in our dynamic market for experience goods where buyer learning is possible, more data about buyer preferences can hurt the seller. 

Our paper also contributes to the literature on dynamic mechanism design.
The canonical method for solving the revenue-maximizing dynamic mechanism with transferable utility, as in \citet{EsoSzentes07} and \citet{PST15}, focuses on distinguishing between private information at the time of contracting and subsequent independent private information. The canonical method
considers a relaxed problem where agents can only misreport the first period and are truthful about subsequent signals.\footnote{One exception is \citet{MP20}, who point out that focusing on the impulse responses of future types to only the initial ones is generically insufficient when
agents are risk-averse so that utility is nontransferable. In our paper, we focus on a risk-neutral agent.} However, the canonical approach does not work in our informed principal dynamic mechanism design problem; due to the seller's private information, the buyer and seller disagree about the evolution of the buyer's belief. As a result, the timing of trade matters, and so the optimal design must keep track of incentive constraints at each point in time, which can be intractable in many settings. Our insight is that the two key forces in this environment are the buyer's incentives to misreport what they learn and to hide their learning; by considering the relaxed problem retaining only the constraints relevant for these deviations, we show that the resulting solution solves the full dynamic informed principal design problem.

This paper also relates to work on experience goods in a dynamic environment. There are two lines of literature that are most related, and we will use two benchmarks in Section 3 to further illustrate the difference and our distinct contribution. First, earlier work \citep{MilgromRoberts86,Bagwell87,BagwellRiordan91,JuddRiordan94} explored how informed sellers use prices as a signal of quality. Those papers consider two-period models where the informed seller has no commitment power and sequentially sets prices in both periods. By model restriction, all feasible pricing behaviors in these models take the form of trials. They focus on how the signaling incentive shapes the first-period price in the separating equilibrium, and whether the price path is increasing or decreasing.\footnote{The trial mechanism in our framework does not specifically refer to an increasing  (average) rate path. In some equilibria, it might feature a decreasing rate if the size of the shock is sufficiently dispersed.} In contrast, our informed seller has the full flexibility to design a dynamic mechanism. In particular, the seller can sell everything ex-ante to avoid buyers gaining information rent from their private experience. We contribute to this literature by proving that trial mechanisms endogenously appear in equilibrium when the mechanism can only determine whether to provide the buyer with access to the service, and provide new insights on the economic forces determining trial length. Further, we show that when the screening technology is richer, a more sophisticated mechanism, dynamic tiered pricing, appears instead. 

Another related line of work is the literature on experimentation in markets such as \citet{Cremer84,BergemannValimaki00,BergemannValimaki06,YDK16,HagiuWright20,CST23}, where the seller is uninformed and needs to post prices period by period or set the flow rate in continuous time model; recent papers in this literature focus on competition between multiple firms. Our model is distinct in three ways. First, a key component of our paper, the informed seller, is not discussed in this literature. Second, we allow the seller to flexibly commit to any dynamic mechanism. 
Third, our discussion allows within-period price discrimination, which is actively featured in dynamic tiered pricing mechanisms.

Lastly, our paper also relates to dynamic adverse-selection models. For example, \cite{GreenTaylor16}, \cite{Madsen22}, and  \cite{CurelloSinander24} all have a principal who elicits disclosure from an agent who experiences a private information arrival process. The main difference in our model is that the principal is privately informed and the mechanism also elicits a private value from the buyer in addition to the time of arrival.

\section{Baseline Model}\label{sec:model}

\paragraph{General Setting}
A seller dynamically sells a service to a buyer. The service is only available for a fixed time interval, so time is given by $t \in \mathcal{T} = [0, T]$. The value of the service depends on the match quality between the service and the buyer: $\theta \in \Theta = \{H, L\}$.

The seller can control the buyer's access to the service $I \in [0,1]$ at any time, where we interpret $I = 0$ as no access to the service, and $I = 1$ as full access.\footnote{In Section 6 we extend our analysis to the case where the seller has access to a richer screening technology and can control both the service access and service quality.} Access determines the rate at which a buyer privately receives instantaneous rewards from using a good service; namely,  if $\theta = H$, a buyer receives instantaneous rewards according to a Poisson process with flow rate $\lambda I$, where $\lambda$ is an exogenous parameter. If $\theta = L$, no instantaneous rewards ever arrive. In the software example, Poisson rewards might correspond to instantaneous boosts in sales when the seller's marketing strategy is suited to their product and industry. 

\paragraph{Outcomes and Payoffs}
We assume both parties are risk-neutral and there is no discounting.\footnote{No discounting merely simplifies the algebraic expressions. Our results extend to infinite time with discounting; see the online appendix.} Instantaneous rewards give utility $v$ to the buyer; the size of $v$ is persistent over time. Thus, if by time $T$, the buyer has cumulatively received $N$ instantaneous rewards, and paid a total price of $p$, their utility is $Nv-p$, while the seller's payoff is simply $p$.

\paragraph{Beliefs and Information}
The common prior on $\theta=H$ is $\mu_0 \in (0, 1)$, and the common prior on $v$ has a  Myerson-regular distribution $F$\footnote{By Myerson-regular, we mean $v - (1 - F(v))/f(v)$ is increasing in $v$.} which admits a density $f$ with nonnegative support $V = [\ubar{v}, \bar{v}]$. We normalize $\mathbb{E}[v]=1$.
The seller does not observe $v$ but privately observes $\theta$: from consumer data the seller gets partial information about whether the buyer would like the service. If the match quality is low, $\theta = L$, we say the seller is of low type; otherwise, the seller is of high type. We suppose that the potential length of the service is sufficiently long: $\lambda T \geq 2(\mu_0-\ubar{v})/((1-\mu_0)\ubar{v})$. 

The buyer does not observe $\theta$ and $v$ at the beginning but privately observes the outcome of the Poisson process from using the service provided by the seller. That is, the buyer privately observes the arrival of instantaneous rewards and their value. Due to the nature of the Poisson process, the buyer perfectly learns both $\theta$ and $v$ upon arrival of the first instantaneous reward.

\paragraph{Timing}
The timing of the game is as follows:
\begin{enumerate}
    \item Nature draws $\theta$ and $v$. The seller observes $\theta$, and the buyer observes neither.
    \item The seller publicly proposes a dynamic mechanism (which we specify below).
    \item The buyer accepts or rejects. Rejecting ends the game and yields outside option $0$.
    \item The mechanism starts, and both parties start reporting messages. The buyer privately observes the arrival and size of rewards.
    \item The game ends when the buyer exits or time expires at $T$.
\end{enumerate}

\paragraph{Mechanism}
The seller proposes a deterministic direct mechanism $m \in \mathcal{M}$.\footnote{This is without loss of generality when focusing on pure-strategy equilibria where sellers can only propose deterministic mechanisms; this is because focusing on pure strategy equilibria restores the inscrutability principle of deterministic direct mechanisms. Furthermore, in the formal proof of Proposition \ref{prop:payoff_set} we show that sellers cannot profitably deviate by offering an indirect mechanism.}
A direct mechanism asks the seller to report $\theta \in \{H,L\}$ only at the beginning. It also asks the buyer to report the payoff-relevant information from the Poisson learning process; at time $t$, the buyer reports $m_t\in M_t$, where the message space $M_t(\{ m_s\}_{s<t})$ is history dependent.
Because of the conclusive Poisson learning structure, the mechanism asks whether the buyer has received an instantaneous reward and if so, its size. Thus, we let $M_0 = \{ U \}$, which means that the buyer initially cannot report an instantaneous reward at the beginning. Next, $M_t(\{ m_s\}_{s<t}) = \{U\} \cup  V$ if $m_s = U$ for all $s<t$; if the buyer has not reported the arrival of the first instantaneous reward, they can continue reporting $U$ or report the arrival and its value by sending $v\in V$. Lastly, $M_t(\{ m_s\}_{s<t}) = \{v\}$ if $m_s = v$ for some $s<t$; that is, having elicited the payoff-relevant information, the mechanism stops eliciting information from the buyer. The buyer can always quit at any time $t$ and get continuation payoff 0.\footnote{All results remain the same if the buyer cannot quit after accepting the mechanism.  In the direct mechanisms we consider, the only information the buyer can transmit to the seller is the time and size of the first reward; however, even if the seller could elicit more information from the buyer, these messages are sufficient to design an optimal mechanism. 
    The restriction on the message space ensures that the sample path of the buyer's reports is always measurable with respect to time, regardless of what message the buyer chooses to report at any point in time. }

A time-$t$ history $h_t = \{ \hat{\theta}, \{ m_s \}_{s < t} \}$ consists of an initial type report from the seller $\hat{\theta}$ and a path of messages $\{m_s\}_{s < t}$. Let $H_t$ be the set of all time-$t$ histories. A mechanism is a tuple $(I,p)$. The first term $I$ consists of a collection of measurable functions $I_t: H_t \to [0,1]$ that map history to an access. The second term $p$ is a collection of measurable functions $p_t: H_t \to \mathbb{R}$ that map history to a cumulative price.  We further require that $(I,p) \in \mathcal{M}$ has sample paths of $(I_t,p_t)_{t=0}^T$ that are measurable with respect to $t$ under any $h_T \in H_T$.


\paragraph{Strategies and Equilibrium} 
The seller's strategy is $\sigma_S = (\sigma_S^P,\sigma_S^R)$, where the proposing strategy $\sigma_S^P: \Theta \to \mathcal{M} $ maps their type into a mechanism proposal; and the reporting strategy $\sigma_S^R: \Theta  \times \mathcal{M} \to \Theta $ maps their type and the mechanism into an initial report to the mechanism.

The buyer's strategy is a pair $\sigma_B = (\sigma_B^A,\sigma_B^R)$, consisting of an acceptance strategy $\sigma_B^A$ and a reporting strategy $\sigma_B^R$. The acceptance strategy  $\sigma_B^A: \mathcal{M} \to \{\text{accept}, \text{reject}\}$ maps the proposed mechanism to acceptance/rejection. To define the reporting strategy, let $N_t$ denote the counting process corresponding to the number of reward arrivals for the buyer, and $n_t$ denote a time-$t$ realization (and note that $N_t$ depends on the access chosen by the mechanism). The buyer's reporting strategy then maps each triple of the realized reward arrival process $n_t$, accepted mechanism $M \in \mathcal{M}$, and realized mechanism history $\{ m_s, I_s, p_s \}_{s < t}$, into a valid message $m_t \in M_t(\{m_s\}_{s<t})$ or an exit decision. Formally, $\sigma_B: (M, \{m_s, I_s, p_s\}_{s<t}, n_t, v) \mapsto m_t \in M_t(\{m_s\}_{s<t}) \cup \{\textnormal{exit}\}$. 
Additionally, we restrict the reporting strategy to only depend on $v$ when $n_t > 0$.

A buyer belief system $\mu(H | \cdot)$ specifies the buyer's belief that $\theta=H$. We consider pure perfect Bayesian equilibria with one additional off-path restriction. A strategy profile $(\sigma_S, \sigma_B)$ together with a buyer belief system\footnote{We do not specify the seller's belief because, at the time when the seller proposes a mechanism or reports a message, the buyer does not have private information. Also, we directly assume the buyer's belief about $v$ is the prior $F$ when $N=0$, and the buyer perfectly learns $v$ if $N>0$ (in particular, this is the no-signaling-what-you-don't-know assumption, which rules out buyer belief updating about $v$ in response to seller actions).} $\mu$ is an equilibrium if
\begin{enumerate}
    \item The seller's strategy $\sigma_S$ is sequentially optimal given $\sigma_B$.
    \item The buyer's strategy $\sigma_B$ is sequentially optimal given the belief system $\mu$ and $\sigma_S$.
    \item The buyer's belief system $\mu$ is consistent with the seller's strategy $\sigma_S$ and Bayes' rule whenever possible.
    \item The buyer belief $\mu(H|(M, \{m_s, I_s, p_s\}_{s < t}, n_t)) = 1$ if $n_t > 0$.
\end{enumerate}

The first three conditions are standard for a weak perfect Bayesian equilibrium. The fourth condition imposes an additional restriction, which states that if a reward has ever arrived (i.e. the counting process realized a positive value) then the buyer must believe that $\theta = H$ with probability 1. This is not immediate from Bayes' rule because the buyer may have assigned probability 0 to $\theta = H$. In that case, reward arrival is subjectively a zero-probability event, and Bayes' rule does not specify what the buyer should believe after seeing a reward. We view condition 4 as very mild; such a restriction is naturally satisfied in sequential equilibria of any discrete approximation of the model.


\paragraph{Applications}
Our framework applies to a range of interactions where sellers possess private information about product or service quality, buyers learn through experience, and dynamic pricing is available. 

Such features naturally arise in digital markets for software products. Enterprise software-as-a-service (SaaS) companies typically have initial private information about the degree to which their products can help another company, usually based on proprietary knowledge about system capabilities. However, although buyers may initially be uncertain about fit, they eventually assess the value of the service based on their experience. Further, many SaaS companies provide their services by issuing licenses or API keys; through these instruments, they have the capability to dynamically control access and prices. 

For example, HubSpot and Salesforce offer marketing software and sales automation as a service to other businesses. Initially, HubSpot or Salesforce possess better information than a prospective client about how much the client might benefit from its lead generation or sales tools, possibly based on private information about similar clients. However, a client eventually learns their own idiosyncratic value for the marketing software service through experience, by observing their own sales.

Streaming platforms and online education services also align with the primitives of our model, as sellers possess private information about content quality while buyers must learn through consumption. Streaming services know how well their recommendation algorithms and content library match user preferences, but individual consumers can only discover this through experience; for example, Prime Video may forecast how a newly released show or movie would land with different segments of their users, and subsequently utilize such predictions to make customized offers.  Similarly, online education platforms such as Coursera, Udemy, and MasterClass have better information about course quality, instructor effectiveness, and educational outcomes than prospective students, who update their beliefs by engaging with the material. In both cases, firms design pricing structures and access policies to account for buyer uncertainty and learning dynamics.

Throughout the paper, we will interpret $\theta$ as a horizontal match quality parameter. This parameter $\theta$ could also be interpreted as a vertical component of product quality in traditional lemon markets. This suggests that trial-like mechanisms should also arise in other adverse selection settings. In fact, trial mechanisms predate the digital economy. In more conventional lemon markets such as used cars, there are rent-to-own contracts and other trial-like mechanisms; for example, online car seller Carvana allows users to return cars within seven days, which could be interpreted as a trial mechanism. Our analysis provides a theoretical foundation for understanding why these mechanisms arise and what forms they take in more complex environments.


\section{First-Best Benchmark}\label{sec:benchmark}
The two key components of our model are (1) dynamic pricing with buyer learning and (2) a seller who knows $\theta$. In this section, we demonstrate that the interaction of both components is vital; in the absence of either, the seller can achieve the first-best outcome by proposing a simple mechanism in equilibrium.

Note that the total surplus generated by the service is $\lambda v T$ when $\theta = H$, and 0 otherwise. Hence, the seller's maximum expected revenue (expectation over both types) is bounded above by the ex-ante total surplus created, which is $\mu_0 \lambda T \mathbb{E}[v] = \mu_0 \lambda T$, since we normalized $\mathbb{E}[v] = 1$. Thus, the first-best revenue is exactly the total surplus $\lambda \mu_0 T$, which can be achieved by selling the entire service at the beginning at the price of $\lambda \mu_0 T$. The following result shows that by removing either of the two key components of our model, selling the entire service at the beginning at the price of $\lambda \mu_0 T$ can be implemented in equilibrium.

\begin{proposition}
    \label{prop:benchmark}
    Suppose at least one of the following is true:
    \begin{enumerate} 
        \item The seller can only sell the service at time $0$: formally, the seller is restricted to proposing mechanisms such that the cumulative price only has a $\theta$-dependent jump at time $0$, $p(h_t) = p(\theta)$ for all histories $h_t$.
        \item The seller does not observe $\theta$. 
    \end{enumerate}
    Then the ex-ante revenue-maximizing mechanism can be implemented in equilibrium.
\end{proposition}

The first point of Proposition \ref{prop:benchmark} implies having more data does not hurt the seller if trade is static. This result is proved by \cite{KOESSLER2016456} in a general environment. One can also easily check that the first best equilibrium survives standard equilibrium refinements of the signaling games. An important assumption supporting the efficiency result is that the seller has type-independent cost, which contrasts with the canonical lemon-market models. We think this assumption is appropriate in the digital market as the type is about match quality and most digital products have zero marginal cost.

The second point of Proposition \ref{prop:benchmark} says that the seller would optimally choose to sell the entire service ex-ante if the seller is uninformed. This highlights how the seller's ability to commit to a dynamic mechanism distinguishes our model from the literature on market experimentation stemming from \cite{Cremer84} and \cite{BergemannValimaki06}, where the uninformed seller sells the service period by period, leaving information rent to the buyer.

The first-best benchmark highlights that in the baseline model, consumer data on $\theta$ is not required for maximizing social welfare, and it is a weakly dominant strategy for the seller to never collect consumer data. Although in reality sellers may benefit from using consumer data in other ways,\footnote{For example, consumer data could help with better product targeting or improving recommendation systems. We abstract away most of these well-studied forces in our main analysis; for more detailed remarks, see the discussion section.} we deliberately study a model with a stark benchmark to isolate a novel channel for how the joint interaction of dynamic pricing and an informed seller can strictly decrease seller revenue. 
While we do not argue that collecting consumer data is always bad for firms, we suggest that the effects of data collection (especially data that the buyer may not know) are nuanced.




\section{Dynamic Mechanisms}\label{sec:mechanisms}

In contrast to the previous section, we will show that in the dynamic informed principal game, an ex-ante revenue-maximizing mechanism may not be supported in any equilibrium. We demonstrate this by constructing a deviating mechanism that a high-type seller can propose and show that this deviation is profitable under some parameter regimes. Motivated by this observation, we devote the rest of the section to characterizing the payoff boundary of IC-IR dynamic mechanisms, which is the main result that will aid our equilibrium payoff characterization in the next section.

Throughout this section, we use some additional notation. Due to the conclusive good news feature of the buyer's learning process, the mechanism only needs to elicit the first reward arrival from the buyer. Thus, we split the histories into histories where no signals have arrived, and histories where a reward of value $v$ arrived at time $t$. We express the mechanism's access $I$ as $\left(I_U^\theta, \{ I_{v,t}^\theta \}_t\right)$, where $I_U^\theta: [0,T] \to [0,1]$, denotes access after history $h_t = \{\theta, \{U\}_{s \le t} \}$, and $I_{v,t}^\theta : [t, T] \to [0,1]$ denotes access after the history where the first reward is communicated at time $t$ with value $v$. We notate $\left(p_U^\theta, \{ p_{v,t}^\theta \}_t\right)$ similarly for prices.

\subsection{High-Type Seller Deviation}\label{sec:free_trial}

Suppose that on-path, both types of seller propose to sell the entire service ex-ante at a price $\lambda \mu_0 T$. We first intuitively describe a deviating mechanism the high-type seller may consider offering, and then discuss when this implies that the first-best cannot be attained.
Consider a high-type seller making the following informal offer:
\begin{quote}
    I will pay you a flow $\epsilon$ to use the service up to time $t_M$. After $t_M$, I will give you the option to buy the remaining $T-t_M$ at price $\lambda v_M (T - t_M)$. 
\end{quote}
We will refer to the limit mechanism when $\epsilon \to 0$ as the \textit{Myersonian free trial} mechanism.
Facing this off-path deviation, the buyer is willing to participate in the trial regardless of their belief about the match quality\footnote{When $\epsilon = 0$, the buyer is indifferent between ``purchasing'' the free trial or not if they believe $\theta=L$ for sure, which could lead to rejection. The $\epsilon$ induces strict incentives to participate.}. Furthermore, once the buyer receives a reward, they learn $\theta$ and $v$ perfectly and are willing to purchase the post-trial service at the price $\lambda v_M(T-t_M)$ if and only if the value $v>v_M$. 
This implies that the high-type seller can ensure a payoff arbitrarily close to $\pi_F$, where
\begin{gather*}
    \pi_F := (1-e^{-\lambda t_M})(1 - F(v_M))\lambda v_M(T-t_M).
\end{gather*}
To maximize $\pi_F$, let $v_M$ denote the monopoly price $v_M := \arg\max_v \left \{ v (1 - F(v)) \right \},$ where $F$ is the distribution of buyer values. Define $t_M$ as 
\begin{gather*}
            t_M := \arg \max_t \left \{  (1-e^{-\lambda t})(T-t) \right\}.
    \end{gather*}
It is straightforward to check that $t_M$ is the maximizer of a strictly concave function and hence is unique. Formally, the deviating dynamic mechanism is
\begin{align*}
 I^L_U(t) = I^H_U(t) = & \begin{cases} 1 & t \le t_M, \\ 0 & t > t_M, \end{cases} &
     p^L_U(t) = p^H_U(t) =&  -\epsilon \min(t, t_M), \\
  I_{v,t}^L(t) =I_{v,t}^H(t) = & \begin{cases}
        1 & v \ge v_M, \\ 
        I_U(t) & v < v_M,
    \end{cases}  & p_{v,t}^L(t) =p_{v,t}^H(t) = & \begin{cases}
        \lambda v_M (T-t_M)  -\epsilon \min(t, t_M) & v \ge v_M, \\ 
        -\epsilon \min(t, t_M) & v < v_M.
    \end{cases} 
\end{align*}
To understand the notation, note that the first statement indicates that the offered mechanism provides full access to an uninformed buyer before time $t_M$, and no access to the uninformed buyer after time $t_M$. A buyer who receives a reward of size $v$ at time $t$ is given full access if $v$ is greater than $v_M$, but is treated as an uninformed buyer otherwise. The seller pays the uninformed buyer a flow $\epsilon$ up to $t_M$, but charges a payment of $\lambda v_M (T - t_M)$ from an informed buyer with value above $v_M$.

It turns out that whether the first-best can be attained in an equilibrium exactly depends on whether the Myersonian free trial is more profitable than selling the entire service ex-ante.

\tgcomment{I think apart from the intuitive criterion, we shall also use point F to argue for the instability. (This requires us to introduce the refinement on belief in the base model, which I think we should). The reason for this is that both Johannes and Dirk seem to not be a huge fan of D1 criterion or in general any refinement. Of course, point F only }

\begin{proposition}
        \label{prop:free_trial_deviation}
        The first-best cannot be implemented in any equilibrium if and only if
        \begin{equation}\label{eqn:fb_condition}
            (1-e^{-\lambda t_M})(1 - F(v_M))\lambda v_M(T-t_M)  >  \lambda \mu_0 T.
        \end{equation}
    \end{proposition} 
Intuitively, the proposition implies that the first-best can be an equilibrium if and only if the Myersonian free trial is not a sufficiently attractive deviation for the high-type seller. 
The full proof of the proposition requires a characterization of the payoffs attainable by all incentive compatible, individually rational (IC-IR) mechanisms, which we will provide as our main result in the next section. However, we first discuss some of the economic implications here before proceeding to the mechanism design problem.

Embedded in Proposition \ref{prop:free_trial_deviation} is the implication that the largest profit a high-type seller can make in any first-best IC-IR mechanism is $\lambda \mu_0 T$. Intuitively, evidence disclosure takes time, so in order for the high-type seller to separate themselves from the low-type seller, they have to spend time and sometimes exclude buyers, inevitably losing surplus. 
The following corollary to Proposition \ref{prop:free_trial_deviation} shows that when a lot of time is available, the high-type seller really would like to use the trial to separate themselves:
\begin{corollary}
\label{coroll_tradeoff}
    For sufficiently large $T$, the first-best cannot be implemented.
\end{corollary}



\subsection{IC-IR Mechanisms}
Having established that the first-best payoff may not be supported in equilibrium, we proceed to provide a sharper characterization of the equilibrium payoffs. We start by characterizing the payoffs that could result from an IC-IR mechanism, which is a necessary condition for equilibrium payoffs. An IC-IR mechanism ensures that (1) it is incentive-compatible for the seller to report $\theta$ and for the buyer to report the private learning process truthfully, and (2) that sellers of both types and a buyer with an initial belief $\mu_0$ receive a non-negative ex-ante payoff. That is, define
\[ \Pi_{\ICIR} =  \left\{(\pi_L, \pi_H) \left| \, \parbox{17em}{ $\exists$ an IC-IR mechanism s.t. the ex-ante payoff for a seller of type $\theta$ is $\pi_\theta$. } \right\}.\right. \]

In order to understand the incentives of the high-type seller, we will focus on characterizing the upper boundary of $\Pi_{\ICIR}$. To do so, we consider the points of $\Pi_\ICIR$ that maximize a weighted sum of seller payoffs, normalizing the weight for the high-type seller to $1$ and allowing the weight on the low-type seller to vary. Specifically, define the boundary as
\[ \partial\Pi_+ = \left\{(\pi_L, \pi_H) \left| \exists w_L \in \left(-\infty, \frac{1-\mu_0}{\mu_0}\right], \ (\pi_L, \pi_H)  \in \argmax_{(\pi_L, \pi_H)  \in \Pi_\ICIR} (w_L \pi_L + \pi_H) \right\}\right. . \]
To clarify the restriction on the range of $w_L$, note that when $w_L = (1-\mu_0)/\mu_0$, the maximization objective is precisely seller profit weighted according to the prior; as shown in the next section, the payoff outcome from this maximization features $\pi_L = \pi_H$. As such, for $w_L > (1-\mu_0)/\mu_0$, the optimal mechanism induces outcomes where $\pi_L > \pi_H$, which we will later rule out in equilibrium.\footnote{See discussion in Section \ref{sec:reasonable_eq}.}. Thus, we focus on $w_L \le (1-\mu_0)/\mu_0$.


Let $\tau$ denote the stochastic arrival time of the first reward, with $\tau=\infty$ if no reward ever arrives; note that $\tau$ endogenously depends on the mechanism access $I$. Since the low type never generates a positive arrival rate, the payoff of the low-type seller is $p^L_U(T)$. Hence, we can rewrite the boundary point condition in $\partial \Pi_+$ as the following maximization:
\begin{gather}
    \max_{M \in \mathcal{M}} ~ \left \{ w_L p^L_U(T) + \mathbb{E}_M \left[p^H_{v,\tau}(T)\mathbbm{1}_{\tau \le T} + p^H_U(T)\mathbbm{1}_{\tau > T} \mid \theta = H \right] \right \} \label{prblm:icir_general} \\
    \text{subject to IC and IR}. \notag
\end{gather}
To characterize the solutions to this maximization problem, we first define the notion of trial mechanisms. Informally, a trial mechanism consists of two phases; an initial \textit{trial} where the seller provides access to everyone, and a second phase where the seller sells the remaining service to buyers who have received a high utility shock. 
\begin{definition}
\label{definition_trail}
    A dynamic mechanism $(I,p)$ is a \textit{trial mechanism} if and only if there exists $(v_0, t_0,p_0)$ such that 
\begin{align*}
I^L_U(t) = I^H_U(t) = & \begin{cases} 1 & t \le t_0, \\ 0 & t > t_0, \end{cases} &
 I^L_{v,t}(t) = I_{v,t}^H(t) = & \begin{cases}
        1 & v \ge v_0, \\ 
        I_U(t) & v < v_0.
    \end{cases}  \\
p^L_U(t) = p^H_U(t) = & ~ p_0 &
 p^L_{v,s}(t) = p_{v,s}^H(t) = & \begin{cases} p_0 +  \lambda v_0 (T-t_0) 1_{t>t_0}  & v \ge v_0 \\ p_0 & v < v_0 \end{cases}
\end{align*}
\end{definition}
By inspection, the Myersonian free trial satisfies the definition of a trial mechanism. In addition, in the special case where $t_0 = T$, the trial mechanism reduces to selling the entire service ex-ante: $I^\theta_U = I^\theta_{v,t} = 1$. Since a low-type seller cannot generate a reward, $I^L_{v,t}$ is off-path; further, intermediate prices can be charged either as lump-sums or flow payments.\footnote{Note that in our definition of trial mechanisms, the cumulative price is constant except at time zero and possibly the report time; so we have taken the stance that payments are made as lump-sums at those times. An outcome-equivalent mechanism could charge some appropriate flow payments instead of lump-sums.} As such, there will be multiple optimal solutions to the design problem. We will characterize the solutions to \eqref{prblm:icir_general} up to \textit{outcome-uniqueness}, by which we mean that all interim payoffs and on-path allocations are unique. 

We are now ready to present our first main result, which states that payoffs on the boundary $\partial \Pi_+$ are outcome-uniquely attained by trial mechanisms.


\begin{theorem}\label{thm:trial}
Suppose payoff pair $(\pi_L, \pi_H) \in \partial\Pi_+$ maximizes the objective \eqref{prblm:icir_general} with weight $w_L$. Then the trial mechanism with parameters $(v_0,t_0, p_0=\pi_L)$ outcome-uniquely attains payoffs $(\pi_L, \pi_H)$, where the value threshold $v_0$ satisfies the implicit equation
    \begin{equation}\label{eqn:thm_v0}
        v_0 = \max \left\{ \ubar{v}, (1 - \mu_0(w_L + 1)_+)\frac{1 - F(v_0)}{f(v_0)} \right \}, \footnote{The notation $(\cdot)_+$ denotes the nonnegative component, i.e. $\max(\cdot, 0)$.}
    \end{equation}
    and the trial length $t_0$ satisfies
    \begin{equation}\label{eqn:thm_t0}
        \begin{cases}
            0 = \lambda e^{-\lambda t_0}(T - t_0) - (1 - e^{-\lambda t_0}) + \frac{\mu_0(1 + w_L)_+}{\pi_0} & \text{if } 1 - e^{-\lambda T} > \frac{\mu_0(1 + w_L)_+}{\pi_0} ,  \\ 
            t_0 = T & \text{if } 1 - e^{-\lambda T} \leq  \frac{\mu_0(1 + w_L)_+}{\pi_0} ,
        \end{cases}
    \end{equation}
      where $\pi_0 = \int_{v_0}^{\bar{v}} \left(v - (1 - \mu_0(w_L+1)_+) \frac{1-F(v)}{f(v)} \right) f(v) \dd v$ is the expected virtual surplus.
\end{theorem}

Qualitatively, Theorem \ref{thm:trial} implies that any IC-IR mechanism which maximizes weighted revenue must feature a trial. Indeed, in the applications discussed earlier, the sellers offer such trials. For marketing software, Salesforce and Hubspot\footnote{As of March 2025, both companies had such offers available at the following links: \href{https://offers.hubspot.com/free-trial}{https://offers.hubspot.com/free-trial} and \href{https://www.salesforce.com/form/signup/freetrial-salesforce-starter/}{https://www.salesforce.com/form/signup/freetrial-salesforce-starter/}.} both offer trials for their products; we also see such trials available for streaming platforms like Apple TV and Paramount Plus, and for educational services such as MasterClass, Coursera, and Udemy.

We will defer the proof of Theorem \ref{thm:trial} to the next section, and first discuss intuition and implications of the result. To build intuition for the fundamental tradeoffs in the mechanism design problem, consider the mechanism designer who aims to maximize the high-type seller's revenue only; equivalently, the designer solves the problem when $w_L=0$. Firstly, consider the structure of the trial mechanism itself. The trial serves to provide evidence about $\theta$ to the buyer, which in turn allows the high-type seller to price higher. To maximally benefit the high-type, the designer wants to offer full access during the trial; doing so accelerates learning and allows the seller to command a higher initial price. After sufficiently many buyers have been convinced, the high-type seller would like to exploit this learning; thus, the seller then discriminates by setting a higher price for the remaining service, excluding low-value buyers and skeptical buyers.

What economic forces influence the discrimination threshold $v_0$? When $w_L = 0$, according to \eqref{eqn:thm_v0}, the threshold $v_H$ maximizes the expected virtual surplus:
\begin{gather*}
    v_H = \argmax_{x}   \int_{x}^{\bar{v}} \left(v - (1 - \mu_0) \frac{1-F(v)}{f(v)} \right) f(v) \dd v.
\end{gather*}
In a standard static design problem, the threshold is chosen to maximize virtual surplus, but the virtual surplus does not feature the $(1-\mu_0)$ term. To understand why the coefficient in front of the information rent term is $(1-\mu_0)$, note that the buyer knows that they might receive information rent, which happens if they see a signal and their value is high. However, the designer can use the price of the trial to extract this expected surplus because the buyer anticipates receiving information rent after the trial has ended, and thus is willing to pay a higher price for the trial.
Given that the buyer believes the seller is the high type with probability $\mu_0$, the buyer is only willing to pay a $\mu_0$ fraction of the information rent they would receive if they learned their value was high during the trial. Hence, the designer cannot extract a $1 - \mu_0$ fraction of the information rent. 

To understand the tradeoffs in determining the length of the trial, first denote the maximized virtual surplus  obtained by $v_H$ as $\pi_H$. The optimal trial length $t_H$ solves
\begin{gather*}
    \max_{t\in[0,T]} ~  \lambda \mu_0 \cdot t +  \lambda \pi_H \cdot (T-t) \cdot (1-e^{-\lambda t}) .
\end{gather*}
The first term in the maximization corresponds to revenue from providing service during the trial phase, which allows the seller to charge a price of $\lambda \mu_0$ per unit of trial time; the $\mu_0$ appears due to the buyer's ex-ante skepticism. After a length-$t$ trial, the high-type seller extracts the virtual surplus from the rest of time, $\lambda \pi_H (T - t)$; however, the high-type seller can only extract this surplus from the $(1 - e^{-\lambda t})$ fraction of buyers who received a signal during the trial. Note that the total expected revenue from learning-based price discrimination is not monotone in $t$; extending the trial length increases the fraction of convinced buyers, but reduces the remaining service available for price discrimination. 

When $\mu_0 > \lambda (1-e^{-\lambda T}) \pi_H$, even following a trial of length $T$, the expected number of convinced buyers fails to make learning-based price discrimination profitable, so $t_H=T$. When the inequality is reversed, $t_H$ is interior. In particular, as long as $\lambda T$ is sufficiently large, the optimal $t_H$ is always interior.

\begin{figure}
     \centering
     \begin{subfigure}[b]{0.45\textwidth}
         \centering
         \begin{tikzpicture}[scale=1]
        \draw[->, thick] (-1,0) -- (5,0) node[anchor=west]{$\pi_L$};
        \draw[->, thick] (0,-1,0) -- (0,5) node[anchor=south]{$\pi_H$};
        \draw[dashed] (-1,-1) -- (5,5);
        \draw[dashed] (-1,3.5) -- (5,3.5);

        \draw[thick] (3.2,3.2) .. controls (2.6,5) and (1.5,5) .. (0.3,3.8) -- (0,3.5);
        
        \filldraw (3.2,3.2) circle (2pt) node[anchor=north west]{$B$};
        \filldraw (1.9,4.65) circle (2pt) node[anchor=south]{$H$};
        \filldraw (0,3.5) circle (2pt) node[anchor=north east]{$F$};

    \end{tikzpicture}
         \caption{$\pi_F  >  \lambda \mu_0 T$}
         \label{fig:mech_payoffs_no_fb}
     \end{subfigure}
     \hfill
     \begin{subfigure}[b]{0.45\textwidth}
         \centering
         \begin{tikzpicture}
        \draw[->, thick] (-1,0) -- (5,0) node[anchor=west]{$\pi_L$};
        \draw[->, thick] (0,-1,0) -- (0,5) node[anchor=south]{$\pi_H$};
        \draw[dashed] (-1,-1) -- (5,5);
        \draw[dashed] (-1,2.3) -- (5,2.3);

        \draw[thick] (4,4) .. controls (3,5) and (1.5,3.8) .. (1,3.3) -- (0,2.3);
        
        \filldraw (4,4) circle (2pt) node[anchor=north west]{$B$};
        \filldraw (3.15,4.37) circle (2pt) node[anchor=south]{$H$};
        \filldraw (0,2.3) circle (2pt) node[anchor=north east]{$F$};
        
    \end{tikzpicture}
         \caption{$\pi_F  \le  \lambda \mu_0 T$}
         \label{fig:mech_payoffs_fb}
     \end{subfigure}
    \caption{The boundary of the IC-IR mechanism payoff set $\Pi_\ICIR$. Point $F$ denotes the payoffs from the free-trial mechanism, point $H$ denotes the best mechanism for the $\theta=H$ type seller, and point $B$ denotes the payoff outcome from selling ex-ante. Figure (a) shows an example where the first-best cannot be attained, and Figure (b) shows an example where the first-best can be attained. }
    \label{fig:mech_payoffs}
\end{figure}

Having discussed the tradeoffs in designing the best mechanism for the high-type seller, we now discuss other boundary payoffs, which optimally maximize different objectives for the designer. 
One such objective might be for the designer to maximize net industry seller profits; equivalently, the designer maximizes the objective when $w_L = (1-\mu_0)/\mu_0$.
\begin{corollary}\label{corr:exante}
    The payoff outcome from selling ex-ante $(\lambda \mu_0 T, \lambda \mu_0 T)$ is the unique maximal element of $\Pi_{\ICIR}$ in the direction $( 1-\mu_0, \mu_0)$: \[ \max_{(\pi_L,\pi_H) \in \Pi_\ICIR} [(1-\mu_0) \pi_L + \mu_0 \pi_H ] = \lambda \mu_0 T.\]
\end{corollary}
As we already established, selling the entire service ex ante attains the first-best; however, our result further implies that any IC-IR mechanism attaining first-best revenue must be outcome-equivalent to selling ex ante.  
Another important mechanism arises if the designer wishes to maximize the difference between the payoffs of the high-type and low-type seller; equivalently, the designer optimizes with weight $w_L = -1$.
\begin{corollary}\label{corr:free}
    The Myersonian free trial payoff $(0, \pi_F)$ is a maximal element of $\Pi_{\ICIR}$ in the direction $( -1, 1)$: \[ \max_{(\pi_L,\pi_H) \in \Pi_\ICIR} [ \pi_H - \pi_L ] = \pi_F. \]
\end{corollary}
Figure \ref{fig:mech_payoffs} illustrates points $F$ and $B$. In the figure, $H$ denotes the best mechanism for the high type. As we analyzed earlier, $H$ does not always have an interior trial length; it may coincide with $B$, but $H$ always Pareto-dominates $F$. This follows intuitively because $F$ corresponds to the Myersonian free trial, and both types can be better off by charging an ex-ante payment for the duration of the trial, $\lambda \mu_0 t_M$.

The following comparative static establishes how the trial length and post-trial price change with the market belief and with the designer's objective weights.
\begin{proposition}\label{prop:comparative_private}
    For $w_L \in [-1, (1-\mu_0)/\mu_0]$, the trial length $t_0$ is weakly increasing in $w_L$ and $\mu_0$, and the post-trial price $v_0$ is weakly decreasing in $w_L$ and $\mu_0$.
\end{proposition}
Intuitively, as the market belief becomes more optimistic ($\mu_0$ increases), both types of seller can extract higher revenue from the trial, and so the mechanism designer increases the trial length. Additionally, as $\mu_0$ increases, the initial price can also extract a higher fraction of the expected post-trial information rents; hence, the designer excludes fewer buyer types after the trial, so $v_0$ decreases. Similarly, as the designer assigns higher weight to the low-type seller (larger $w_L$), the designer becomes willing to trade high-type post-trial revenue for a longer trial, and hence $t_0$ increases. Further, as the designer values the high-type seller less relative to the low-type seller, the designer is more willing to provide information rents to the buyer post-trial, since a fraction of these rents can be extracted with the initial price; hence, when $w_L$ increases, $v_0$ decreases.

Revisiting the example of an enterprise software company, the prior belief $\mu_0$ corresponds to the market belief that the marketing software matches any given client's needs. Proposition \ref{prop:comparative_private} then implies that if the market becomes more pessimistic about the value of the software, then a market designer with a fixed objective should choose mechanisms with \textit{shorter} trials, but at a lower price. Also, observe that the Myersonian free trial parameters $t_M, v_M$ are independent of $\mu_0$; thus, a designer that wishes to maximize the difference between the high- and low-type seller's payoffs is agnostic to the buyer belief over types.

How often does the high-type seller withhold service from the buyer, and how does this inefficiency depend on the designer's objective? Proposition \ref{prop:comparative_private} shows that as we reduce $w_L$ and move from point $B$ to $F$ in Figure \ref{fig:mech_payoffs}, the expected service provided by a high-type seller decreases in both an extensive and an intensive margin; in other words, the more the designer weights the high-type seller's payoff relative to the low-type seller's, the less the expected total amount of service sold by the high-type seller. In the extensive margin, the trial length decreases, leading to fewer buyers learning $\theta=H$; in the intensive margin, the threshold $v_0$ increases, which excludes more buyers with low values. As such, the Myersonian free trial corresponding to point $F$ has the shortest trial and excludes the most buyers after the trial, so the Myersonian free trial provides the least amount of service in expectation across all trial mechanisms. We will return to this observation in our discussion of equilibrium refinements.

Another interesting comparative analysis is how the dispersion of consumers' idiosyncratic preferences affects the length of the trial. Consider an illustrative parametrized example where the buyer's private value $v \sim U[1 - \delta , 1 + \delta ]$ is uniformly distributed.
\begin{proposition}\label{prop:comparative_v}
    Suppose $v \sim U[1 - \delta , 1 + \delta]$. The trial length $t_0$ is weakly increasing in $\delta$.
\end{proposition}
Intuitively, as $\delta$ decreases, the private value of consumers becomes more concentrated and it becomes more profitable to implement learning-based price discrimination. Consequently, it is optimal for the informed seller to shorten the trial and start price discrimination earlier.

\subsection{Proof of Theorem \ref{thm:trial}}
This subsection sketches the proof of Theorem \ref{thm:trial}. Readers less concerned with the technical details may jump to the next section.
The proof proceeds in three overall steps. First, we consider a problem relaxation, in which we drop most of the IC and IR constraints. Next, we show that for any $w_L$, the only solutions to the relaxed problems are trial mechanisms with $t_0,v_0$ determined by equations \eqref{eqn:thm_v0} and \eqref{eqn:thm_t0}. Finally, we show that the construction is feasible under the original IC-IR constraints. 

\subsubsection{Relaxed Problem}

We begin by providing an intuitive interpretation of the relaxed problem. We consider the problem of designing a mechanism that satisfies a subset of IC-IR constraints, namely only the following four constraints:
\begin{enumerate}
    \item A buyer who receives a reward $v$ at time $t$ does not want to misreport being uninformed forever. We will refer to this as (IC-U) for the incentive-constraint on not misreporting as uninformed. 
    \item A buyer who receives a reward $v$ at time $t$ does not want to misreport $v'$ at time $t$; we will refer to this as (IC-V) for the incentive constraint on not misreporting value.
    \item Participation is individually rational for the buyer at time zero. We will refer to this constraint as (IR-0).
    \item The low seller type is willing to participate and report truthfully. We will refer to this constraint as (IC-S).
\end{enumerate}
That is, we drop constraints that prevent an uninformed buyer from reporting as informed, constraints where a buyer who sees a reward $v$ at time $t$ \textit{jointly} deviates to reporting $v'$ at time $t'$, and the constraint on the high-type seller not reporting as low-type.


To formulate the relaxed problem mathematically, we first introduce some additional notation.
Let 
\[ I(v,t) := \int_t^T  I_{v,t}^H(s) \ \dd s, \quad I(U,t) := \int_t^T  I^H_U(s) \ \dd s, \] be the cumulative service access provided after the buyer reports $(v,t)$ and for the uninformed buyer after time $t$, respectively. We start by considering the problem relaxation:
\begin{align}
    \max_{M \in \mathcal{M}} ~ & \left \{ w_L p^L_U(T) + \mathbb{E} \left[p^H_{v,\tau}(T)\mathbbm{1}[\tau \le T] + p^H_U(T)\mathbbm{1}[\tau > T] \mid \theta = H \right] \right \} \label{prblm:relaxed} \\
    \text{subject to} \quad & 
    \lambda v I(v,t) - p^H_{v,t}(T) \geq \lambda v I(U,t)  - p^H_U(T) \qquad \forall (v, t) \tag{IC-U}\\
    & \lambda v I(v,t) - p^H_{v,t}(T)  \geq \lambda v I(v',t) - p^H_{v',t}(T) \qquad \forall (v, v', t) \tag{IC-V}\\
    &  \mu_0\mathbb{E}\left[N_T v - p^H_{v,\tau}(T) \mid \theta = H \right] - (1-\mu_0) p^L_U(T) \ge 0 \tag{IR-0} \\
    & p^L_U(T) \ge p^H_U(T)  \tag{IC-S}
\notag
\end{align}
where we slightly abuse notation in writing the ex-ante IR for the buyer in the third constraint; that is, if $\tau = \infty$ (no reward arrives) we denote $p^\theta_{\cdot,\infty}(T) = p^\theta_U(T)$. Also recall that $N_t$ denotes the counting process for the number of arrivals during the service.

Solving the relaxed problem \eqref{prblm:relaxed} takes several steps. We first eliminate the IR constraints (Lemma \ref{lem:uninformed_payments_equal}). Second, we show that the access provided after any report $I(v,t)$ is always at least $I(U,t)$ (Lemma \ref{Lemma_low_bind}). This allows us to explicitly solve what the mechanism should do after a report, by characterizing $I(v,t)$ in terms of $I(U,t)$. With some careful rearrangement (Lemma \ref{lemma:pluginPointWise}), the remaining problem reduces to a control optimization over a single function, $I(U,t)$. Finally, we solve the control problem and argue that the optimal solutions are trial mechanisms satisfying \eqref{eqn:thm_v0} and \eqref{eqn:thm_t0} with the appropriate price (Lemma \ref{lem:necessity_high}).

We begin solving the relaxed problem by eliminating $p^L_U$ and (IR-0):
\begin{lemma}\label{lem:uninformed_payments_equal}
    The cumulative payment by an uninformed buyer is equal for both types: \[p^L_U(T) = p^H_U(T).\]
    If $w_L > -1$, (IR-0) must bind at optimum. If $w_L \le -1$ it is without loss to assume it binds.
\end{lemma}
The proof approach supposes that for some feasible solution, the inequality $p^L_U(T) > p^H_U(T)$ is strict; by casework on $w_L$, either a decrease in $p^L_U(T)$ or an increase in $p^H_U(T)$ (or both) improve the objective while maintaining feasibility. Since the casework is mechanical, it is left to the appendix. To simplify the rest of the exposition, suppose $w_L \ge -1$. \footnote{This is purely for cleanliness and clarity. It will turn out that the optimal mechanism when $w_L = -1$ ensures that the low-type gets payoff zero, and hence will also be optimal for $w_L < -1$ (since the payoff for the low-type cannot be negative by IR).} 

Lemma \ref{lem:uninformed_payments_equal} allows us to set (IR-0) as binding, which eliminates $p^L_U(T)$ from the maximization problem. Since the lemma implies that $p^L_U(T) = p^H_U(T)$, it will be useful to characterize the payments $p^H_{v,t}(T)$ relative to the uninformed buyer's payment $p^H_U(T) = p^L_U(T)$. Denote $\Delta_{v,t}^H(T) = p_{v,t}^H(T) - p_U^H(T)$, the additional net payment incurred by a report $(v,t)$. By setting (IR-0) binding, we get 
\begin{align} 0 &= \mu_0\mathbb{E}\left[N_T v - p^H_{v,\tau}(T) \mid \theta = H \right] - (1-\mu_0) p^L_U(T), \notag \\
 \implies p^H_U(T) &= \mu_0 \mathbb{E}[N_T v - \Delta^H_{v,\tau}(T) \mid \theta = H].\label{eqn:ir0-bind}
\end{align}
Further, recall that the expected number of arrivals of an inhomogenous Poisson process is given by the cumulative arrival rate. Let the random variable $S$ correspond to the report time of the buyer, and note that the service changes from $I_U^H$ to $I_{v,t}^H$ at the $S$. Then 
\begin{align}  \mathbb{E}[\mathbb{E}[N_T \mid S] \mid \theta = H] 
&= \mathbb{E}\left[\left. \int_0^S \lambda I^H_U(t) \dd{t} + \int_{S}^T \lambda I_{v,t}^H(t) \dd{t} \right|  \theta = H\right], \notag \\ 
&= \mathbb{E}\left[\left.\int_0^T \lambda I^H_U(t) \dd{t} + \int_{S}^T \lambda(I_{v,t}^H(t) - I^H_U(t)) \dd{t} \right| \theta = H\right], \notag \\ 
&= \mathbb{E}[\lambda I(U,0) + \lambda (I(v,S) - I(U,S)) \mid \theta = H]. \notag 
\end{align}
Since the maximization objective is evaluated for truthful buyer reporting, the law of iterated expectations then implies that
\begin{equation}
    \mathbb{E}[N_T \mid \theta = H]= \mathbb{E}[\lambda I(U,0) + \lambda (I(v,\tau) - I(U,\tau)) \mid \theta = H]. 
 \label{eqn:nt_derivation}
\end{equation}

Then \eqref{eqn:ir0-bind} and \eqref{eqn:nt_derivation} combined with Lemma 1 gives us an expression for $p^L_U(T) = p^H_U(T)$, which we can plug into the objective \eqref{prblm:relaxed}:
\begin{align}
    \eqref{prblm:relaxed}
    =& (w_L + 1)p^H_U(T) + \mathbb{E}\left[\Delta^H_{v,\tau}(T)\mathbbm{1}[\tau \le T] \mid \theta = H \right] \notag \\
    =& (w_L+1)\mu_0\left( \mathbb{E}\left[ \lambda v (I(U,0) + I(v,\tau) - I(U,\tau)) - \Delta^H_{v,\tau}(T) \mid \theta = H \right]  \right) \notag \\
    & \quad + \mathbb{E} \left[\Delta^H_{v,\tau}(T)\mathbbm{1}[\tau \le T]\mid \theta = H \right] \notag
\end{align}
All together, the optimization problem then becomes
\begin{align}
    \max_{M \in \mathcal{M}} ~ & \begin{Bmatrix} (w_L+1)\mu_0\left( \mathbb{E}\left[ \lambda v (I(U,0) + I(v,\tau) - I(U,\tau)) - \Delta^H_{v,\tau}(T) \mid \theta = H \right]  \right)  \\ + \mathbb{E} \left[\Delta^H_{v,\tau}(T)\mathbbm{1}[\tau \le T]\mid \theta = H \right] \end{Bmatrix} \label{prblm:relaxed_no_ir} \\
   \text{subject to} \quad & 
    \lambda v I(v,t) - \Delta^H_{v,t}(T) \geq \lambda v I(U,t)  \qquad \forall (v, t) \tag{IC-U}\\
    & \lambda v I(v,t) - \Delta^H_{v,t}(T)  \geq \lambda v I(v',t) - \Delta^H_{v',t}(T)  \qquad \forall (v, v', t) \tag{IC-V}
\end{align}

An important feature of this problem is that all incentive constraints are time-separable: the terms $I(v,t)$ and $I(U,t)$ do not interact with $I(v,t')$ and $I(U,t')$. This is the main tractability advantage we get from our conclusive-learning framework. In the non-relaxed problem with all IC constraints, we have to keep track of incentive constraints at each time $t$; without time-separability, the term $I(v,t)$ will appear in all time $t'$ IC constraints for $t'<t$.

Note that there are still intertemporal tradeoffs in the design problem, since $I(U,t)$ affects the distribution of the first arrival $\tau$, but the intertemporal considerations only factor into the problem through the optimal choice of $I(U,t)$. As such, we will first solve for $I(v,t)$ (i.e., what the mechanism should do post-report) and use that solution to solve for $I(U,t)$ (what the mechanism does pre-report). 

\subsubsection{Solving Post-Report}

To take advantage of the time-separability, we first fix $I(U,t)$ and focus on optimizing $I(v,t)$ and $\Delta_{v,t}^H$. When $I(U,t)$ is fixed, the distribution of $\tau$ in the objective function is also fixed, and so is the $I(U,0)$ term. Consequently, to solve for $I(v,t)$, we only need to consider pointwise optimization at each time $t \le T$:
\begin{align}
    \max ~ &  (w_L+1)\mu_0\left( \mathbb{E}\left[ \lambda v I(v,t)  - \Delta^H_{v,t}(T) \mid \theta = H \right]  \right)   + \mathbb{E} \left[\Delta^H_{v,t}(T) \mid \theta = H \right] \label{prblm:pointwise} \\
   \text{subject to} \quad & 
    \lambda v I(v,t) - \Delta^H_{v,t}(T) \geq \lambda v I(U,t)  \qquad \forall  v \tag{IC-U}\\
    & \lambda v I(v,t) - \Delta^H_{v,t}(T)  \geq \lambda v I(v',t) - \Delta^H_{v',t}(T) \qquad \forall  (v, v') \tag{IC-V}
\end{align}
Note that the maximization objective \eqref{prblm:pointwise} drops terms that are fixed by $I(U,t)$ (namely, the $I(U,0)$ and $I(U,t)$ terms) and fixes time $t$, so the expectations are taken over $v$. To start solving for $I(v,t)$, we first show that $I(U,t)$ must lower bound $I(v,t)$:
\begin{lemma}
\label{Lemma_low_bind}
    In the optimal solution of the above relaxed problem, $I(v,t) \geq I(U,t)$ for any $v$.
\end{lemma}
The intuition of Lemma \ref{Lemma_low_bind} is that, because the buyer can always ensure a rent of $\lambda v I(U,t)$ by never reporting, whenever $I(v,t) < I(U,t)$, the access provided by $I(v,t)$ has no impact on the information rent that must be paid to higher value buyer types $v'\geq v$. The formal details are left to the appendix.

Given Lemma \ref{Lemma_low_bind}, a standard envelope theorem argument simplifies (IC-V). Define the interim utility of the buyer:
\begin{equation}\label{eqn:interim_u}
    u(v,t) := \lambda v I(v,t) - \Delta^H_{v,t}(T). 
\end{equation} 
By the envelope theorem, since the (IC-V) constraint is equivalent to $u(v,t) = \max_{v'} \{ \lambda v I(v',t) - \Delta^H_{v',t}(T) \}$, we can integrate to get 
\[ u(v,t) = u(\ubar{v},t) + \int_{\ubar{v}}^v \lambda I(w,t) \dd{w}. \]
Also, since Lemma \ref{Lemma_low_bind} implies that $I(w,t) \ge I(U,t)$, it follows that 
\[ u(v,t) \ge u(\ubar{v},t) + \lambda(v - \ubar{v}) I(U,t) = u(\ubar{v},t) - \lambda \ubar{v} I(U,t) + \lambda v I(U,t) . \]
Hence, if (IC-U) is satisfied for $\ubar{v}$, it is satisfied for all $v$. By standard arguments, we rewrite the problem retaining the envelope representation of (IC-V), (IC-U) for the worst type, and isolate the $\Delta^H_{v,t}(T)$ term, so the problem becomes 
\begin{align*}
    \max ~ &  (w_L+1)\mu_0\left( \mathbb{E}\left[ \lambda v I(v,t)  \mid \theta = H \right]  \right)   + (1 -  \mu_0(w_L + 1))\mathbb{E} \left[\Delta^H_{v,t}(T) \mid \theta = H \right] \notag \\
   \text{subject to} \quad & 
    u(v,t) = u(\ubar{v},t)  + \int_{\ubar{v}}^v \lambda I(w,t) \dd w , \\
    & I(v,t) \geq I(U,t) \quad \forall  v, \quad u(\ubar{v},t) \ge \lambda \ubar{v} I(U,t). 
\end{align*}
The remaining optimization on $I(v,t)$ is standard; \eqref{eqn:interim_u} pins down $\Delta_{v,t}^H(T)$ for any given $u$, and so we can replace the last $\Delta^H_{v,t}(T)$ term in the objective and integrate by parts, yielding
\begin{gather*}
     \max_{I(v,t) \geq I(U,t)} ~\mathbb{E}\left[    \lambda v I(v,t) - (1 - \mu_0(w_L+1))\frac{1-F(v)}{f(v)}  \lambda I(v,t)    | \theta = H \right].
\end{gather*}
By the regularity assumption, $v-(1-F(v))/f(v)$ is increasing in $v$, so $v-(1-\mu_0(w_L+1))(1-F(v))/f(v)$ is also increasing in $v$. Consequently, the above expression is maximized by increasing $I(v,t)$ as much as possible for $v$ such that
\begin{equation}\label{eqn:v_ineq}
    v \ge (1-\mu_0(w_L + 1))\frac{1 - F(v)}{f(v)}
\end{equation}
and decreasing $I(v,t)$ to $I(U,t)$ otherwise.
The smallest $v$ such that \eqref{eqn:v_ineq} holds is precisely $v_0$ in \eqref{eqn:thm_v0}.
Then, it follows that the unique point-wise maximization solution is
\begin{equation} I(v,t) = \begin{cases}
        T-t & v \ge v_0, \\ 
        I(U,t) & v < v_0.
    \end{cases} \label{eqn:first_ivt} \end{equation}
The remaining step is to solve for the optimal $I(U,t)$.

\subsubsection{Solving Pre-Report}
Having solved for the optimal $I(v,t)$ given $I(U,t)$, and noting that such an $I(v,t)$ ensures (IC-U) and (IC-V) are satisfied, we can revisit the optimization objective \eqref{prblm:relaxed_no_ir}:
\begin{equation}
     \max_{M \in \mathcal{M}}  \begin{Bmatrix} (w_L+1)\mu_0\left( \mathbb{E}\left[ \lambda v (I(U,0) + I(v,\tau) - I(U,\tau)) - \Delta^H_{v,\tau}(T) \mid \theta = H \right]  \right)  \\ + \mathbb{E} \left[\Delta^H_{v,\tau}(T)\mathbbm{1}[\tau \le T]\mid \theta = H \right] \end{Bmatrix} \notag
\end{equation}
Since the $I(v,t)$ schedule determined in the previous section satisfies (IC-U) and (IC-V), by plugging in the form $I(v,t)$ and rearranging the objective function, we have the following simplification:
\begin{lemma}
\label{lemma:pluginPointWise}
The maximization objective \eqref{prblm:relaxed_no_ir}, given the characterization of $I(v,t)$ in \eqref{eqn:first_ivt}, is equivalent to the following:
    \begin{align}
    \max_{M \in \mathcal{M}} ~ & \left \{ (w_L+1)\mu_0  \lambda \mathbb{E}[v]I(U,0) + \lambda \pi_0 \mathbb{E}[(T - \tau - I(U,\tau)) \mathbbm{1}[\tau \le T] | \theta = H ] \right \}\notag 
\notag
\end{align}
where \[ \pi_0 := \mu_0(w_L+1) \int_{v_0}^{\bar{v}}v  f(v)\ \dd v  + v_0(1 - F(v_0))(1 - \mu_0(w_L+1)). \]
\end{lemma}
The proof is purely algebraic and we relegate it to the appendix. In this final optimization problem, the design of $I(U,t)$ incorporates all the intertemporal tradeoffs. To solve it, we change variables to convert the maximization problem in Lemma \ref{lemma:pluginPointWise} into a control problem, using the fact that the density of $\tau$ is given by $\lambda (-\dot{I}(U,t)) \exp \left(\lambda (I(U,t) - I(U,0)) \right)$. Define the \textit{state} variable as $X(t) = \int_0^t I^H_U(t)$, and the measurable \textit{control} variable as $U(t) = I^H_U(t)$. Then the problem is equivalent to:
\begin{align}
    \max_{X, U} ~ &\left \{ \pi_0 \int_0^T \left( \lambda \left(T - t - (X(T) - X(t))\right) \exp \left(- \lambda X(t)  \right) + \frac{\mu_0(1 + w_L)}{\pi_0}\right) \lambda U(t)\dd t  \right \} \label{prblm:ctrl}\\
    \textnormal{subject to } ~ & \dot{X}(t) = U(t) \in [0,1], ~ X(0) = 0, ~ X(T) \textnormal{ free}. \notag
\end{align}
 We solve this control problem using Lemma \ref{lem:ctrl} in the appendix. The unique optimal policy is bang-bang:
\[I^H_U(t) = \begin{cases} 1 & t \le t_0, \\ 0 & t > t_0, \end{cases}\]
where $t_0$ satisfies the condition in \eqref{eqn:thm_t0}.
Finally, having identified a candidate optimal mechanism, we now argue that this is essentially unique (up to degeneracy in specifying the mechanism off-path when a buyer reports an arrival with a low-type seller):
\begin{lemma}\label{lem:necessity_high}
    Any optimal mechanism of the relaxed problem must be outcome-equivalent to a trial mechanism, with $v_0$ given by \eqref{eqn:thm_v0}, trial length $t_0$ given by \eqref{eqn:thm_t0}, and price $p_0$ given by 
    \begin{equation}\label{eqn:thm_p0}
        \begin{cases}
        p_0 = \lambda \mu_0 t_0 + \lambda (\pi_0 - v_0(1 - F(v_0)))(T - t_0)(1 - e^{-\lambda t_0})  & w_L > -1, \\
        p_0 \in \left[0, \lambda \mu_0 t_0 + \lambda (\pi_0 - v_0(1 - F(v_0)))(T - t_0)(1 - e^{-\lambda t_0})\right] & w_L = -1.
    \end{cases}
    \end{equation}  
\end{lemma}
The expressions for $p_0$ are determined by obtaining $p^H_U(T) = p^L_U(T)$ from Lemma \ref{lem:uninformed_payments_equal} and plugging in the forms for $I(U,t)$ and $I(v,t)$ solved for previously. Since the analysis pinned down $I(U,t), I(v,t)$ and $\Delta^H_{v,t}(T)$, it is straightforward to check that given these two objects, all interim payoffs are uniquely pinned down. 
To complete the proof of Theorem \ref{thm:trial}, one just needs to verify all constraints dropped in the relaxed problem hold. The details are left to the appendix. 


\section{Equilibrium Analysis}\label{sec:equilibria}

After establishing Theorem \ref{thm:trial} to characterize the set of payoff outcomes of IC-IR mechanisms in the previous section, we move on to discuss equilibria. We first characterize the equilibrium payoffs, and then discuss which equilibria  survive the stronger D1 equilibrium refinement.

\subsection{Equilibrium Payoffs}\label{sec:reasonable_eq}
In this subsection, we characterize the equilibrium payoffs where the high-type seller makes weakly more than the low-type seller, and justify our focus on such equilibrium outcomes.

\begin{proposition}\label{prop:payoff_set}
    A payoff pair $(\pi_L,\pi_H)$ is an equilibrium payoff if and only if it is the outcome of a feasible IC-IR mechanism $(\pi_L,\pi_H)\in \Pi_\ICIR$ and the high-type seller gets at least the free-trial revenue: $\pi_H \geq \pi_F$.
\end{proposition}

\begin{figure}
     \centering
     \begin{subfigure}[b]{0.45\textwidth}
         \centering
         \begin{tikzpicture}[scale=1]
        \draw[->, thick] (-1,0) -- (5,0) node[anchor=west]{$\pi_L$};
        \draw[->, thick] (0,-1,0) -- (0,5) node[anchor=south]{$\pi_H$};
        \draw[dashed] (-1,-1) -- (5,5);
        \draw[dashed] (-1,3.5) -- (5,3.5);

        \filldraw[blue, opacity=0.4] (3.1,3.5) .. controls (2.6,5) and (1.5,5) .. (0.3,3.8) -- (0,3.5) -- (3.1,3.5);
        
        \filldraw (3.2,3.2) circle (2pt) node[anchor=north west]{$B$};
        \filldraw (1.9,4.65) circle (2pt) node[anchor=south]{$H$};
        \filldraw (0,3.5) circle (2pt) node[anchor=north east]{$F$};

    \end{tikzpicture}
         \caption{$\pi_F  \geq  \lambda \mu_0 T$}
         \label{fig:eq_payoffs_no_fb}
     \end{subfigure}
     \hfill
     \begin{subfigure}[b]{0.45\textwidth}
         \centering
         \begin{tikzpicture}
        \draw[->, thick] (-1,0) -- (5,0) node[anchor=west]{$\pi_L$};
        \draw[->, thick] (0,-1,0) -- (0,5) node[anchor=south]{$\pi_H$};
        \draw[dashed] (-1,-1) -- (5,5);
        \draw[dashed] (-1,2.3) -- (5,2.3);

        \filldraw[blue,opacity=0.4] (4,4) .. controls (3,5) and (1.5,3.8) .. (1,3.3) -- (0,2.3) -- (2.3,2.3) -- (4,4);
        \filldraw[blue,opacity=0.15] (2.3,2.3) -- (3, 2.3) .. controls (3.8,3) and (4.3,3.9) .. (4,4) -- (2.3,2.3);
        
        \filldraw (4,4) circle (2pt) node[anchor=north west]{$B$};
        \filldraw (3.15,4.37) circle (2pt) node[anchor=south]{$H$};
        \filldraw (0,2.3) circle (2pt) node[anchor=north east]{$F$};
        
    \end{tikzpicture}
         \caption{$\pi_F  <  \lambda \mu_0 T$}
         \label{fig:eq_payoffs_fb}
     \end{subfigure}
    \caption{The equilibrium payoff sets. Point $F$ denotes the payoffs from the free-trial mechanism, point $H$ denotes the best mechanism for the $\theta=H$ type seller, and point $B$ denotes the payoff outcome from selling ex-ante. }
    \label{fig:eq_payoffs}
\end{figure}

Proposition \ref{prop:payoff_set} implies that the additional restriction imposed by the equilibrium condition is precisely that the high-type seller must get at least as much as their free-trial payoff. Figure \ref{fig:eq_payoffs} plots equilibrium payoff outcomes. As Figure \ref{fig:eq_payoffs_no_fb} illustrates, this sometimes excludes the possibility of selling ex-ante; as Figure \ref{fig:eq_payoffs_fb}, the equilibrium payoff set could include the outcome from selling ex-ante, as well as mechanisms where the high-type seller gets less than the low-type seller. We first argue that a weak refinement rules out these outcomes.

Define a payoff outcome $(\pi_L,\pi_H)$ to be a \textit{reasonable equilibrium payoff} if and only if $\pi_H \geq \pi_L$.
We argue that any equilibrium payoff $(\pi_L,\pi_H)$ with $\pi_H<\pi_L$ does not survive a very weak refinement of off-path beliefs. To see this, consider an equilibrium payoff $(\pi_L,\pi_H)$ such that $\pi_L> \pi_H$. From the ex-ante IR constraint of the buyer, we know the average payoff
\begin{gather*}
    p := \mu_0 \pi_H + (1-\mu_0) \pi_L < \mu_0 \lambda T.
\end{gather*}
Moreover, $\pi_L>p>\pi_H$. Now suppose that the high-type seller deviates and sells the entire service to the buyer at a price of $p$. The high-type seller makes the following informal speech:
\begin{quote}
    If I were a low-type seller, this deviation would not benefit me because the maximum I could receive is $p$, which is strictly lower than my equilibrium payoff. Therefore, when confronted with this deviation, you should not be more concerned about me being a low-type seller than you would be in equilibrium. 
    \footnote{A related speech appears in Section 6 of \citet{myerson83} as justification for focusing on core mechanisms; however, our refinement is weaker than the core mechanism requirement of \cite{myerson83}.}
\end{quote}

As long as the buyer holds a belief $\mu\geq \mu_0$ when facing this deviation, it is strictly optimal to purchase the service. Thus, the high-type seller has a profitable deviation, and any payoff where $\pi_H < \pi_L$ cannot be supported in equilibrium.

We separate this refinement from the stronger equilibrium refinements and embed it in the analysis of the equilibrium payoff set because the assumption that the buyer is convinced by the above informal speech is fundamentally weaker than the  standard signaling game refinement concepts.\footnote{For example, it is weaker than the Intuitive Criterion as the Intuitive Criterion requires the buyer to have an extreme belief $\mu=1$ facing this deviation, whereas here the buyer can have any belief $\mu\geq \mu_0$.} 
Note that Theorem \ref{thm:trial} precisely characterized the boundary of payoff outcomes of mecahnisms where $\pi_H \ge \pi_L$; together, Proposition \ref{prop:payoff_set} and Theorem \ref{thm:trial} characterize all reasonable equilibrium payoffs.\footnote{In contrast to Theorem \ref{thm:trial}, which shows that any boundary payoff can be uniquely implemented by a trial mechanism, the implementation of non-boundary payoffs is not unique. To achieve an interior (non-boundary) payoff, one can pick a trial mechanism $(v_0, t_0, p_0)$ on the boundary and reduce price $p_0$ so that the payoff pair moves downwards along a 45-degree line.}
\begin{corollary}
All reasonable equilibrium payoffs can be implemented by trial mechanisms.
\end{corollary}
Note that trial mechanisms feature complete pooling, in which the allocation rule of the mechanism does not depend on the seller's report. This is stronger than the inscrutability principle of Myerson, which only states that the buyer receives no information from the mechanism proposal. In trial mechanisms, the buyer receives no information from the \textit{entire} dynamic allocation of the mechanism and only updates their belief according to observed signals. 
Complete pooling implies that both types of sellers offer the same trial regardless of their private information about the buyer; this implies that on-path, we would not expect to see different trials offered to different buyers.

\subsection{Further Refinements}
In the equilibrium payoff set characterized previously, there are many payoff outcomes generated by many trial mechanisms, as is typical of signaling games. In this subsection, we now discuss whether standard equilibrium refinements used in signaling games offer sharper predictions. We follow the signaling literature and focus on the D1 criterion. 
Roughly, the D1 refinement requires that, when observing an unexpected deviating mechanism, the buyer believes that the seller is of the type that ``benefits most'' from the deviation, where ``benefit most'' means the type benefits for a larger set of possible buyer responses. 
The following result shows how the D1 criterion refines the equilibrium payoff set; we will provide intuition for the result while relegating the formalism to the appendix.
\begin{proposition}
\label{prop:D1}
The only trial mechanisms proposed in equilibria that survive the D1 Criterion have trial length $t_M$ and post-trial rate $v_M$. They differ only in the price charged for the trial. All payoff outcomes are Pareto dominated by point $H$ for any prior $\mu_0$, and are Pareto dominated by point $B$ if $\mu_0$ is small.
\end{proposition}

\begin{figure}
    \centering
    \begin{tikzpicture}[scale=0.7]
        \draw[->, thick] (-1,0) -- (5,0) node[anchor=west]{$\pi_L$};
        \draw[->, thick] (0,-1,0) -- (0,5) node[anchor=south]{$\pi_H$};
        \draw[dashed] (-1,-1) -- (5,5);

        \draw[dashed] (0,0) -- (4,4) .. controls (3,5) and (1.5,3.8) .. (1,3.3) -- (0,2.3) -- cycle;
        
        \filldraw (0,2.3) circle (2pt) node[anchor=north east]{$F$};
        \filldraw (1,3.3) circle (2pt) node[anchor=south east]{$D$};
        \filldraw (3.15,4.37) circle (2pt) node[anchor=south]{$H$};
        \filldraw (4,4) circle (2pt) node[anchor=south]{$B$};
        
        \draw[very thick] (0,2.3) -- (1,3.3);
    \end{tikzpicture}
    \caption{The payoff set surviving D1 is the segment from $D$ to $F$.}
    \label{fig:D1_payoffs}
\end{figure}

Among all possible trial mechanisms that achieve different boundary payoffs in $\partial \Pi_+$,
 the only equilibrium trial mechanisms that survive equilibrium selection have the shortest trial length and the highest post-trial rate. They feature the maximum degree of learning-based discrimination in both the extensive margin and the intensive margin, and lead to the minimum social surplus. 
 As illustrated in Figure \ref{fig:D1_payoffs}, all points between $D$ and $F$ are Pareto dominated for the seller by point $H$. Furthermore, they can also be Pareto dominated for the seller by the ex-ante profit maximizing mechanism at point $B$ if the market belief $\mu_0$ is sufficiently small. Therefore, sellers of both types may be strictly worse off when the seller can predict the buyer's experience but cannot credibly communicate the information.

Compared to Proposition \ref{prop:free_trial_deviation}, the negative result in Proposition \ref{prop:D1} requires a stronger equilibrium refinement and yields a stronger negative prediction. In  Proposition \ref{prop:free_trial_deviation}, ex-ante revenue maximization is not possible under some parameters, and only the low-type seller would be worse off than the first-best benchmark. In Proposition \ref{prop:D1}, ex-ante revenue maximization is never possible, and even the high-type seller might be worse off than the first-best benchmark. This channel of signaling incentives does not exist in the static environment either: selling everything ex-ante at price $\lambda \mu_0 T$ survives the D1 criterion in the benchmark environment as defined in Proposition \ref{prop:benchmark}.

The basic intuition of Proposition \ref{prop:D1} is as follows. The high-type seller profits from both the initial payment for the trial and from post-trial sales. The optimal design (point H) optimally balances this tradeoff. However, the low-type also benefits from the payment for the trial. Thus, the high-type seller always has an incentive to signal high match quality by deviating to a mechanism that reduces the ex-ante payment but generates more revenue after the trial, which can only benefit the high-type seller. Consequently, only the mechanisms such that the seller charges the monopoly price $v_M$ after the trial can survive equilibrium refinement; these are precisely the mechanisms with maximum learning-based price discrimination. 

\paragraph{Discussion}
To connect our results with the informed principal literature, we relate our trial mechanisms to the mechanisms discussed in \cite{myerson83}. The best ``safe'' mechanism is precisely the Myersonian free trial if the Myerson optimal price for distribution $F$ is larger than the expected buyer value $\mathbb{E}[v]$. However, because of the learning that appears in the dynamic informed principal model, the Myersonian free trial is not safe if the Myerson optimal price is less than the expected buyer value. In our framework, the agent must find it optimal to accept the proposal \textit{and} report the learning process truthfully. Under the free trial mechanism, if the buyer knows $\theta=H$ but does not receive any rewards during the trial, the buyer will find it optimal to misreport and purchase the post-trial service. This implies the Myersonian free trial is not safe; in this case, the best safe mechanism for the high-type seller is a free trial with an even larger price after the trial; in particular, it must set a post-trial price at least $\mathbb{E}[v] = 1$. 
Since the Myersonian free trial is dominated, all safe mechanisms are dominated, which implies that the dynamic informed principal problem we study has no strong solution. 
Further, the D1-surviving equilibria are dominated, which implies that they eliminate all of the neutral optimum and core mechanisms.

The equilibrium selection result in our model appears quite different from standard signaling games such as Spence’s education signaling model, where all pooling equilibria are eliminated by the intuitive criterion. This contrast stems from the nature of the sender’s signaling instrument. In Spence’s classical setup, there is no upper bound on education costs, allowing high types to always deviate to a more costly education level to credibly signal their strength. In our model, by contrast, high types can signal their strength by proposing trial mechanisms that generate higher revenue through post-trial price discrimination. However, since the revenue from price discrimination is bounded above, pooling equilibria in which sellers offer trial mechanisms with maximal post-trial revenue survive the refinement.

\section{Endogenizing Quality}\label{sec:general_screening}
In the baseline model, the seller's choice of access $I$ simultaneously determines the buyer's learning rate $\lambda I$ and the expected payoff flow $\lambda Iv$, conditioned on $\theta=H$. This perfect co-linearity is particularly relevant when the seller's capability to customize the service is constrained, such as when the seller can only offer either full service or no service at all.

Nonetheless, in many applications, sellers often possess more sophisticated screening technologies. For instance, software vendors may offer versions with limited functionality along with a premium version. Such functional limitations may diminish the service's value, yet they can still highlight the potential advantages of the premium service to users. Similarly, streaming platforms may permit free viewing interspersed with numerous advertisements, along with an uninterrupted premium experience. Although frequent ad interruptions are a source of displeasure for viewers, users' ability to use the service facilitates their assessment of the platform's content offerings.

In this section, we extend our results to allow the seller to control quality as well as access. Formally, we now assume that the seller can provide $(I,q) \in {\mathcal{D}}\subseteq [0,1]^2$ at any time. The rewards arrives at Poisson  rate $\lambda I$ if and only if the $\theta=H$, and its arrival gives  utility $v {q}$ to the buyer. That is, the new instrument $q$ allows the seller to potentially degrade the service, reducing the value of rewards. We refer to $q$ as the service quality and $I$ as the learning rate. We take $\mathcal{D}$ to be any finite set that contains $(1,1)$ and some point where $I = 0$; that is, at minimum there exist options to provide the best service or no service.
In this context, the baseline we previously analyzed is a special case where $\mathcal{D} = \{ (I,1) \mid I \in [0,1] \}$; that is, the seller could only control access and quality was fixed to 1.

A mechanism is a triple $(I,q,p)$. The first term $I$ is a collection of measurable functions $I_t: H_t \to [0,1]$ that map history to learning rate. The second term $q$ consists of a collection of measurable functions $q_t: H_t \to [0,1]$ that map history to service quality. Any chosen allocation (learning rate-quality pair) must be feasible, so for any $h_t$, $(I(h_t),q(h_t))\in \mathcal{D}$. Prices $p$ are defined analogously as before. 

The results in the baseline model generalize under this specification. For the purpose of exposition, we present only the formal generalization of the IC-IR dynamic mechanisms analogously to Section \ref{sec:mechanisms}. That is, for weight $w_L \le (1-\mu_0)/\mu_0$, consider the design problem:
\begin{gather}
    \max_{I,q, p} ~ \left \{ w_L p^L_U(T) + \mathbb{E} \left[p^H_{v,\tau}(T)\mathbbm{1}[\tau \le T] + p^H_U(T)\mathbbm{1}[\tau > T] \mid \theta = H \right] \right \}, \label{prblm:quality} \\
    \text{subject to IC and IR.} \notag
\end{gather}
We will discuss the welfare implications of enriching the seller's screening technology after characterizing the IC-IR mechanisms. To present the characterization, we introduce some additional notation. Observe that when the service $(I,{q})\in {\mathcal{D}}$ is provided, the learning rate is $\lambda I$ and the expected flow payoff conditioned on $\theta=H$ is $\lambda v {q}I$. The general closed set $\mathcal{D}$ might seem difficult to handle at first glance, but it turns out that the relevant allocations are the extreme points with high learning rate and low expected flow payoff. To precisely characterize these extreme points, define \[ \tilde{\mathcal{D}} := \{(I,u) \mid \exists (I,q) \in \mathcal{D} \textnormal{ such that } qI = u\}, \] which is the pair of feasible learning rate / flow utility pairs. Denote the lower convex envelope of the set $\tilde{\mathcal{D}}$ as
\begin{gather*}
    \partial^- \conv(\tilde{\mathcal{D}}) := \{  (I,u) \in \conv(\tilde{\mathcal{D}})| u \leq u',~ \forall (I,u')\in \conv(\tilde{\mathcal{D}})     \}.
\end{gather*}
Because $\tilde{\mathcal{D}}$ is not necessarily convex, points in $ \partial^- \conv(\tilde{\mathcal{D}})$ are not necessarily in $\tilde{\mathcal{D}}$; however, extreme points of $ \partial^- \conv(\tilde{\mathcal{D}})$ are in $\tilde{\mathcal{D}}$. Let $\mathcal{D}_*$ be the extreme points of $\partial^- \conv(\tilde{\mathcal{D}})$. Let $\Gamma$ denote the set of functions that induce outcomes in $\mathcal{D}_*$ with decreasing learning rate:
\begin{align*}
    \Gamma &:= \{ \gamma(\cdot ):[0,T] \to \mathcal{D} ~ | ~ (\gamma_I(t), \gamma_q(t)\gamma_I(t)) \in \mathcal{D}_* ~ \forall t, \textnormal{ and }\gamma_I(s) \geq \gamma_I(t) ~\forall s \leq t      \} .
\end{align*}
The set $\Gamma$ contains the feasible allocation rules $\gamma(t)$ which induce some service with extremally high learning rate and low flow payoff, with decreasing learning rate over time. 
\paragraph{Example} To fix ideas, consider an example $\mathcal{D}=\{  (1,1), (0,0),(I_1,q_1),(I_2,q_2) \}$ with $0<q_1<q_2<1$. That is, besides the full service and no service, the seller also has two intermediate options which provide intermediate qualities and learning rates. For this example, Figure \ref{fig:D_example} depicts $\mathcal{D}$ and $\tilde{\mathcal{D}}$. The first panel, Figure \ref{fig:D_picture} plots $\mathcal{D}$, which lies in the space of quality/learning rate pairs. The second panel, Figure \ref{fig:Dtilde} plots $\tilde{\mathcal{D}}$, which lies in flow utility/learning rate space. Note that in the example in the figure, only $(I_1, q_1)$ corresponds to a learning rate/flow utility pair $(I_1, u_1)$ which lies in the lower convex envelope of $\tilde{\mathcal{D}}$; hence, any function in $\Gamma$ can never provide the allocation $(I_2, q_2)$. 
\begin{figure}
     \centering
     \begin{subfigure}[b]{0.44\textwidth}
         \centering
         \begin{tikzpicture}
            \draw[->, thick] (-0.5,0) -- (4.2,0) node[anchor=west]{$I$};
            \draw[->, thick] (0,-0.5) -- (0,4) node[anchor=south]{$q$};
            \filldraw[blue] (0,0) circle (2pt) node[anchor=north west]{(0,0)};
            \filldraw[blue] (4,4) circle (2pt) node[anchor=south west]{(1,1)};
            \filldraw[blue] (2.5,3.2) circle (2pt) node[anchor=south west]{$(I_2,q_2)$};
            \filldraw[blue] (1.5,0.5) circle (2pt) node[anchor=south west]{$(I_1,q_1)$};
         \end{tikzpicture}
         \caption{The set $\mathcal{D}$ }
         \label{fig:D_picture}
     \end{subfigure}
     \hfill
     \begin{subfigure}[b]{0.44\textwidth}
         \centering
         \begin{tikzpicture}
            \draw[->, thick] (-0.5,0) -- (4.2,0) node[anchor=west]{$I$};
            \draw[->, thick] (0,-0.5) -- (0,4) node[anchor=south]{$Iq$};
            \filldraw[blue] (0,0) circle (2pt) node[anchor=north west]{(0,0)};
            \filldraw[blue] (4,4) circle (2pt) node[anchor=south west]{(1,1)};
            \filldraw[blue] (2.5,2) circle (2pt) node[anchor=south east]{$(I_2,u_2)$};
            \filldraw[blue] (1.5,0.1875) circle (2pt) node[anchor=west]{$(I_1,u_1)$};
            \filldraw[blue, opacity = 0.2] (0,0) -- (4,4) -- (1.5,0.1875) -- cycle;
            \draw[red, very thick] (0,0) -- (1.5,0.1875) -- (4,4);
         \end{tikzpicture}
         \caption{The set $\tilde{\mathcal{D}}$.}
         \label{fig:Dtilde}
     \end{subfigure}
     
        \caption{Example with quality options $\mathcal{D} = \{ (0,0), (1,1), (I_1, q_1), (I_2, q_2) \}$. The red line in the second plot depicts the lower convex envelope $\partial^-\conv(\tilde{\mathcal{D}})$. Note that $\mathcal{D}_*$ only consists of the points $(0,0), (1,1)$, and $(I_1, u_1)$.}
        \label{fig:D_example}
\end{figure}

We can now define dynamic tiered pricing mechanisms.

\begin{definition}
    A dynamic mechanism $(I,q,p)$ is a \textit{dynamic tiered pricing mechanism} if and only if there exists $v_0$ and $\gamma_0(t) \in \Gamma$ such that
\begin{align*}
      & (I^L_U,q^L_U)(t) =(I^H_U,q^H_U)(t) = \gamma_0(t), \quad   \quad    (I_{v,t}^L,q_{v,t}^L)(t) =(I_{v,t}^H,q_{v,t}^H)(t) = 
     \begin{cases}
        (1,1) & v \ge v_0, \\ 
        (I_U^H,q_U^H)(t) & v < v_0,
    \end{cases}    
 \\
    & p_U^L   (t) = p_U^H   (t)  = p_0, \qquad \qquad      p_{v,s}^L(t) = p_{v,s}^H(t)  = 
    \begin{cases}
        p_0 + v_0 \lambda \int_s^t [1-I_U^H(x)q_U^H(x)] \dd x & v \ge v_0, \\
        p_0  & v < v_0.
    \end{cases}
\end{align*}
\end{definition}
Intuitively, in dynamic tiered pricing mechanisms, the seller provides the uninformed buyer with service that degrades over time, and offers an  option to upgrade to the premium service at an additional price. 
Notably, immediate communication between buyer and seller is necessary, as the service is upgraded immediately after the buyer receives a high-value utility shock. In reality, many software vendors provide a very short trial of the premium version and then offer a free budget version of the software, which not only has limited functionality but also might not be updated as frequently as the premium version. Many streaming services also provide a limited trial of the premium service (ad-free etc.) and then offer a budget version after it ends.

Trial mechanisms are a special case of dynamic tiered pricing mechanisms when $\mathcal{D} = \{ (I,1)| I\in[0,1] \}$. To see this, note that in the baseline model $\mathcal{D}_*= \{ (0,0), (1,1) \}$ and any $\gamma(t) \in \Gamma$ takes the form $\gamma(t) = (1_{t<t_0},1_{t <t_0})$. To the uninformed buyer, the mechanism provides the premium service $(1,1)$ before time $t_0$ and provides nothing after $t_0$. 

\begin{theorem}
\label{thm:tieredpricing}
    Denote $\ubar{I} = \min \{ I>0 | (I,u) \in \mathcal{D}_*\}$. If $\lambda T>\frac{2 s +1 - \ubar{v}}{ \ubar{I} s}  \frac{\mu_0 -\ubar{v}}{(1-\mu_0)\ubar{v}}$.\footnote{Here $s=\int_{\mu_0}^{\bar{v}}  (v-\mu_0) f(v) \dd v$ is the informed consumer's rent relative to uninformed consumer. Note that when $\mu_0\leq \ubar{v}$ this constraint is trivially satisfied.}, 
    There exists a dynamic tiered pricing mechanism with parameters $v(w_L)$ and $\gamma(w_L,t)\in \Gamma$ that solves the program \eqref{prblm:quality}.
\end{theorem} 

Analogously to Theorem \ref{thm:trial}, Theorem \ref{thm:tieredpricing} characterizes the boundary of the payoff set of IC-IR mechanisms, which can be achieved by sellers of both types proposing the same dynamic tiered pricing mechanism. To highlight how the seller's screening instrument affects the informed principal's tradeoffs and to build intuition for the result, we will demonstrate the insights with the example where we augment the baseline model by granting the seller two additional quality levels. We leave the formal proofs to the appendix.

\paragraph{Example Revisited}
Recall the example $\mathcal{D}=\{  (1,1), (0,0),(I_1,q_1),(I_2,q_2) \}$ with $0<q_1<q_2<1$. First, which service(s) would the informed seller provide in a dynamic tiered pricing mechanism? By calculating the corresponding $\mathcal{D}_*$, one can see that the service with the lowest quality $(I_1,q_1)$ will always be provided as a screening device, the intermediate service $(I_2,q_2)$  will be used if and only if $I_2(1-q_2) > I_1(1-q_1)$, which requires the learning rate $I_2$ to exceed $I_1$ so much that the informed seller provides the service as an intermediate screening device when inducing more learning is still valuable. Figure \ref{fig:D_example} shows an example where $I_2(1 - q_2) \le I_1(1-q_1)$; equivalently, the point $(I_2, q_2I_2)$ lies within the convex hull of $(0,0), (1,1)$, and $(I_1, q_1I_1)$.

Assuming $I_2(1-q_2) \le I_1(1-q_1)$ as in Figure \ref{fig:D_example}, the dynamic tiered pricing mechanism offers the uninformed buyer $\gamma_0$ which can be characterized by two switching times $t_0$ and $t_1$:
\begin{align*}
 (I^L_U,q^L_U)(t) =(I^H_U,q^H_U)(t) = \gamma_0(t) =
 \begin{cases}
     (1,1)  \quad &  \forall t \in [0,t_0), \\
     (I_1,q_1)  \quad &  \forall t \in [t_0,t_1), \\
      (0,0) \quad &  \forall t \in [t_1,T ].
 \end{cases}
\end{align*}
We refer to $\gamma_0$, what the mechanism offers an uninformed buyer, as the budget service. The budget service initially provides premium service for a duration of \( t_0 \), followed by the intermediate $(I_1, q_1)$ until \( t_1 \). If the buyer has not reported any reward by \( t_1 \), the budget service terminates. However, if the buyer reports the arrival of a reward of size \( v \geq v_0 \), they are immediately upgraded to premium service. 

To build intuition for why the mechanism chooses allocations corresponding to the lower convex envelope of $\tilde{D}$, consider a high-type seller choosing what $(I,q)$ to offer. Increasing $I$ speeds up the rate that the uninformed buyer learns, which always benefits the seller. The impact of increasing the budget service's quality is more nuanced. On the one hand, increasing the expected flow payoff increases uninformed buyers' consumption value, which leads to a higher ex-ante payment. On the other hand, it hinders the seller's ability to extract surplus from high-value informed buyers. 

The magnitudes of these forces change over time. At $t=0$, the benefit of increasing $I$ is highest, since all buyers are uninformed. Similarly, since all buyers are uninformed, the seller wants to offer higher quality, as it increases the ex-ante value of the service and there are no information rents to be lost to buyers who have experienced a signal. As $t$ increases, the mass of uninformed buyers decreases while the mass of informed buyers increases. With fewer uninformed buyers, the benefit of providing higher learning rate $I$ decreases. With more informed buyers, providing higher $q$ also means conceding more information rent. Eventually, the seller wants to provide as low of a flow payoff as possible to the uninformed buyers.

\begin{figure}[]
         \centering
         \begin{subfigure}[b]{0.3\textwidth}
         \centering
         \begin{tikzpicture}[scale=0.6]
            \draw[->, thick] (-1,0) -- (5,0) node[anchor=west]{\small  $I$};
            \draw[->, thick] (0,-1,0) -- (0,5) node[anchor=south]{\small  $Iq$};
            \filldraw[blue, opacity = 0.4] (0,0) -- (4,4) --   (1.5,0.1875)-- (0,0) ;
            \filldraw[blue] (4,4) circle (2pt);
        \begin{scope}[blue, very thick, decoration={
                markings,  
                mark=at position 0.25 with {\arrow{>}},
                mark=at position 0.75 with {\arrow{>}},
                }] 
            \draw[postaction={decorate}] (4,4) --   (1.5,0.1875) -- (0,0);
        \end{scope}

        \draw[dashed] (-1, 0) -- (0, -1);
        \draw[dashed] (-1, 1) -- (1, -1);
        \draw[dashed] (-1, 2) -- (2, -1);
        \draw[dashed] (-1, 3) -- (3, -1);
        \draw[dashed] (-1, 4) -- (4, -1);
        \draw[dashed] (-1, 5) -- (5, -1);
        \draw[dashed] (0, 5) -- (5, 0);
        \draw[dashed] (1, 5) -- (5, 1);
        \draw[dashed] (2, 5) -- (5, 2);
        \draw[dashed] (3, 5) -- (5, 3);
        \draw[dashed] (4, 5) -- (5, 4);
        \filldraw[red] (4,4) circle (4pt) node[anchor=south west]{$I_t,q_t$};
        
         \end{tikzpicture}
         \caption{Initial time}
         \label{fig:general_intuition_1}
     \end{subfigure}
     \hfill
     \begin{subfigure}[b]{0.3\textwidth}
         \centering
         \begin{tikzpicture}[scale=0.6]
               \draw[->, thick] (-1,0) -- (5,0) node[anchor=west]{\small  $I$};
            \draw[->, thick] (0,-1,0) -- (0,5) node[anchor=south]{\small  $Iq$};
            \filldraw[blue, opacity = 0.4] (0,0) -- (4,4) --   (1.5,0.1875) -- (0,0) ;
            \filldraw[blue] (4,4) circle (2pt);
        \begin{scope}[blue, very thick, decoration={
                markings,  
                mark=at position 0.25 with {\arrow{>}},
                mark=at position 0.75 with {\arrow{>}},
                }] 
            \draw[postaction={decorate}] (4,4) --  (1.5,0.1875) -- (0,0);
        \end{scope}
        \draw[dashed] (1.5, -1) -- (5, 4);
        \draw[dashed] (0.5, -1) -- (4.7, 5);
        \draw[dashed] (-0.5, -1) -- (3.7, 5);
        \draw[dashed] (-1, -0.2857) -- (2.7, 5);
        \draw[dashed] (-1, 1.142857) -- (1.7, 5);
        \draw[dashed] (-1, 2.571428) -- (0.7, 5);
        \draw[dashed] (2.5, -1) -- (5, 2.57143);
        \draw[dashed] (3.5, -1) -- (5, 1.142857);
        \filldraw[red]  (1.5,0.1875) circle (4pt) node[anchor=north ]{$I_t,q_t$};
        
         \end{tikzpicture}
         \caption{Intermediate time}
         \label{fig:general_intuition_2}
     \end{subfigure}
     \hfill
     \begin{subfigure}[b]{0.3\textwidth}
         \centering
         \begin{tikzpicture}[scale=0.6]
                 \draw[->, thick] (-1,0) -- (5,0) node[anchor=west]{\small  $I$};
            \draw[->, thick] (0,-1,0) -- (0,5) node[anchor=south]{\small  $Iq$};
            \filldraw[blue, opacity = 0.4] (0,0) -- (4,4) --   (1.5,0.1875) -- (0,0) ;
            \filldraw[blue] (4,4) circle (2pt);
        \begin{scope}[blue, very thick, decoration={
                markings,  
                mark=at position 0.25 with {\arrow{>}},
                mark=at position 0.75 with {\arrow{>}},
                }] 
            \draw[postaction={decorate}] (4,4) --   (1.5,0.1875) -- (0,0);
        \end{scope}

        \draw[dashed] (-1, 3.5) -- (5, 4);
        \draw[dashed] (-1, 2.5) -- (5, 3);
        \draw[dashed] (-1, 1.5) -- (5, 2);
        \draw[dashed] (-1, 0.5) -- (5, 1);
        \draw[dashed] (-1, -0.5) -- (5, 0);

        \filldraw[red] (0,0) circle (4pt) node[anchor=north west]{$I_t,q_t$};
         \end{tikzpicture}
         \caption{Later time}
         \label{fig:general_intuition_3}
     \end{subfigure}
     \hfill
         
        \caption{Indifference curves over $(I,Iq)$  as time progresses. The red point illustrations the evolution of the low-tier service in time. }
        \label{fig:general_intuition}
\end{figure}

Figure \ref{fig:general_intuition} illustrates the dynamics of this tradeoff in the context of the example. Initially, the incentives of the seller are pictured in Figure \ref{fig:general_intuition_1}. At the start, all buyers are uninformed, so increasing the expected flow payoff of the budget service does not impede surplus extraction while it increases buyer's consumption payoff which can be extracted via ex-ante payment. The seller prefers a higher learning rate $I$ as well as a higher expected flow payoff, so the premium service is provided. 

Figure \ref{fig:general_intuition_2} depicts the incentives after some time has passed and some buyers have become convinced about the match quality. Providing quality in the budget service means that the seller must concede  information rents to users of the premium service. When comparing the services $(1,1)$ or $(I_1, q_1)$, the benefits from higher learning rate no longer outweigh the information rent losses, so the seller degrades the service.

As time progresses further, as in Figure \ref{fig:general_intuition_3}, more and more buyers receive reward and upgrade to the premium service. Providing a budget service only increases the payoff of the dwindling mass of uninformed buyers. Worse, it also requires conceding information rents to users of the premium service. At this point, providing even the intermediate service for the remaining unconvinced buyers is no longer worth the information rent losses to the convinced buyers, and the budget service ends.

\paragraph{Welfare Implications}
We now demonstrate novel welfare implications of expanding the seller's screening technology $\mathcal{D}$. When the seller is uninformed or has perfect commitment power over data usage, expanding the seller's screening ability typically increases social welfare. However, this is not the case if the seller is privately informed.

To illustrate this possibility, consider a seller who only has the baseline technology $\mathcal{D} = \{ (I,1) \mid I \in[0,1] \}$ compared to a seller who has the baseline technology but also the additional point $(I_1, q_1) = (1,0)$.

\begin{figure}[h]
    \centering
    \includegraphics[width=0.55\linewidth]{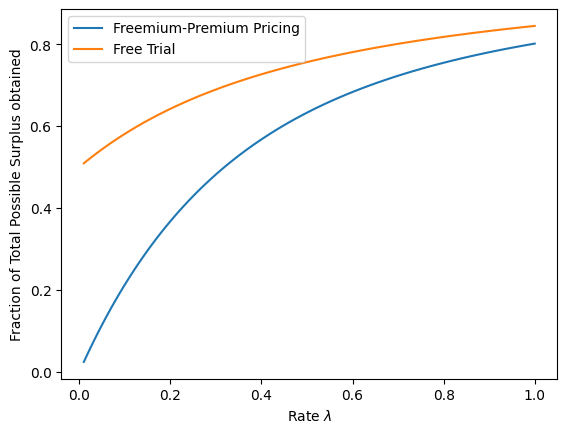}
    \caption{Comparison of welfare induced by the Myersonian free trial and by freemium-premium pricing. Numerical simulations were run with $T = 5$ and values uniformly distributed on $[0.9,1.1]$. The welfare as a fraction of total possible surplus $\lambda \mu_0 T$ is plotted. }
    \label{fig:simulations}
\end{figure}
When faced with the most pessimistic buyer belief ($\theta = L$ for sure), the baseline seller offers a Myersonian free trial, as established before. However, the seller with the additional allocation can now offer $(1,0)$, a low-tier service with maximal learning rate but zero flow value; one can interpret this as ``freemium-premium'' pricing, where the seller offers a low-quality free version and the option to upgrade to $(1,1)$.

Intuitively, the high-type seller generates value for all types above $v_M$ after those types receive their first payoff. Compared to the Myersonian free trial, freemium-premium pricing never excludes any buyers, even if they do not become convinced for a long time. However, the free trial provides more consumption value to buyers who experience a signal during the trial, while freemium-premium pricing reduces the quality received by the buyer during the same trial period. As demonstrated in the numerical simulations in Figure \ref{fig:simulations}, the Myersonian free trial produces higher welfare than freemium-premium pricing; so giving the seller the additional screening instrument $(1,0)$ may decrease total welfare.

\section{Extensions and Discussion} \label{sec:conclusion}
\subsection{Extensions}
We now consider some extensions of our model. To simplify the exposition, we focus on the baseline model where the seller only chooses access $I\in[0,1]$. All results generalize with general $\mathcal{D}$. However, we do impose a restriction that is with loss of generality. In particular, instead of assuming the seller can propose a mechanism that elicits information from both the seller and the buyer, we assume the proposed mechanism only elicits information from the buyer. 

The purpose of our prior analysis in the baseline model without imposing this restriction is two-fold. First, it connects our model to the informed principal approach of \cite{myerson83}. Second, it helps to contrast our setting with the results of \cite{KOESSLER2016456}, who show that the ex-ante revenue-maximizing mechanism can be supported as an equilibrium in the static setting. The inefficiency can arise even if general mechanisms are allowed.
In the baseline model, it turns out that imposing the restriction to mechanisms that only elicit information from the buyer is without loss, because the boundary is achieved by a complete pooling mechanism, where the allocation does not depend on the seller's type.

From an applied point of view, we view the more restrictive model assumption, that mechanisms can only elicit information from the buyer, as more reasonable. Firms never publicize their consumer data or their algorithms, so a mechanism that depends on the seller's private information is hard to monitor. Thus, we proceed with this assumption in this section.

\subsubsection{Imperfectly Informed, Cost and Known Utility Flow}
In this subsection, we introduce three additional features each governed by one parameter. By evaluating the parameter at the corner value, each additional feature can be shut down. The main purpose of this subsection is to show the robustness of the relevant mechanisms up to a slight modification.

First, the seller is only partially informed about the match quality. We use $S=H,L$ to denote the seller's type. The high-type seller receives a more positive signal about $\theta$ and believes $\theta=1$ with a higher probability $\mu_H$. In contrast, the low-type seller believes $\theta=1$ with a lower probability $\mu_L$. The seller receives a high signal with probability $p_H$ and a low signal with probability $p_L$. Bayesian consistency requires $\mu_0 = p_H \mu_H + p_L \mu_L$.
 Second, regardless of the match quality, the service now always gives the buyer a constant utility flow of $u I$ if the seller provides $I$ access to the service. Third, regardless of the match quality, the service leads to a flow cost of $c I$. We focus on the analysis of the pooling equilibria and show that the relevant mechanism takes the form of cancellable trials.\footnote{Different from the main model where all equilibrium payoffs can be achieved by ``pooling equilibrium'', in the extension there may be some payoffs achieved only by separating equilibria. However, these payoffs vanish as three parameters converge to the main model and they are not Pareto optimal.}
\begin{definition}
\label{defi:Trialrefund}
    A dynamic mechanism $(I,p)$ is a \textit{cancellable trial} if and only if there exists $v_c,v_0, t_0$ such that 
\begin{align*}
I_U(t) = & \begin{cases} 1 & t \le t_0, \\ 0 & t > t_0, \end{cases} &
 I_{v,t}(t)  = & \begin{cases}
        1 & v \in [v_0,\bar{v}], \\ 
        I_U(t) & v \in [v_c, v_0), \\
        0      & v \in [\ubar{v},v_c].
    \end{cases} 
\end{align*}
\end{definition}
The word ``cancellable'' comes from the fact that $I_{v,t}(t)$ might equal $0$ even before the trial ends at $t_0$. This option gives the seller a chance to immediately stop the service once the buyer finds out that providing the service is not socially efficient during the trial. In contrast to the choice of $v_0$ and $t_0$ which varies across equilibria, the choice of $v_c$ in any ``good'' equilibrium mechanism is rather mechanical: $u + \lambda v_c^* = c$. When $v_c^* < \ubar{v}$, as in our baseline model, introducing early cancellation is not necessary.

Denote $\pi_L^0$ as the low type's maximum profit when separating from the high type. The low type achieves this payoff by using a degenerate cancellable trial, where $v_c= v_c^*$ and $t_0 =T$. Essentially, the low type's optimal separating mechanism is a two-part tariff where the low-type seller charges an upfront payment and a flow price of $c$ for future usage. Denote $\pi_H^0$ as the high type's maximum profit when the buyer believes the type is low. The optimal mechanism to achieve $\pi_H^0$ is a cancellable trial with a similar spirit to the free trial at point $F$ in the baseline model. However, the trial is not necessarily free when the seller only gets an imperfect signal or the constant flow utility $u$ is larger than the cost, as the buyer understands even the low-type seller can create value.

Similar to Proposition \ref{prop:payoff_set}, any $(\pi_L,\pi_H)$ is a pooling equilibrium payoff pair if it Pareto dominates $(\pi_L^0,\pi^0_H)$ and can be achieved by a (pooling) dynamic mechanism that is IC-IR to the buyer. The remaining question is to characterize the payoff boundary of IC-IR dynamic mechanisms.
\[ \partial\Pi_+ = \left.\Big\{(\pi_L, \pi_H) \right| \exists w_L \in [-1,(1-\mu_0)/\mu_0], \ (\pi_L, \pi_H)  \in \argmax_{(\pi_L, \pi_H)  \in \Pi_\ICIR} (w_L \pi_L + \pi_H) \Big\}. \]

\begin{proposition}
    \label{prop:extensionsTRM}
    In the extended model, there exists a cancellable trial with payoffs $(\pi_L, \pi_H)$ for any $(\pi_L, \pi_H) \in \partial\Pi_+$. Further, in any dynamic mechanism with payoffs in $\partial\Pi_+$, the seller must offer a cancellable trial. 
\end{proposition}

\subsubsection{Other Learning Processes}
In this paper, we focus on the good news model. To have a sense of the robustness of the results, this subsection explores the implication of other Poisson learning models such as the bad news model and the mixed new model. 
We show that the relevant mechanisms are still in the realm of cancellable trials as introduced in Definition \ref{defi:Trialrefund} in the previous section.


\paragraph{Bad News Model}
The service delivers a type-independent flow utility of $uI$ to the buyer. However, if the match quality is low $\theta=L$, the buyer receives an instantaneous loss $l$ according to a Poisson process with flow rate $\lambda I$.  If $\theta=H$, no instantaneous losses ever arrive. The buyer believes $\theta=H$ with probability $\mu_0$. To illustrate the new force introduced by the bad news learning, we focus on parameters such that $ \lambda l > u > (1-\mu_0) \lambda l$, which means low match quality hurts the buyer yet in expectation the service still has a positive value. 

We use $I_U(t)$ to denote the access for the uninformed buyer who reports no signal up to time $t$, and denote $I(t)$ to denote the access for the informed buyer who reports to receive a bad signal.
Again, we directly look at the IC-IR mechanism that maximizes the linear-weighted-revenue of sellers:
\begin{gather*}
  max_{M\in \mathcal{M}} \quad  w_L \pi_L(M) + \pi_H(M) ,\\
  s.t. \quad   M \text{ is  IC-IR}.
\end{gather*}
\begin{proposition}
\label{prop:badnews}
In the bad news model, for any $w_L\in[-1,(1-\mu_0)/\mu_0]$, the following (degenerate) cancellable trial is optimal: 
\begin{align*}
I_U(t) = 1,  ~ ~ 
 I(t)  = 0, \quad \forall t.
\end{align*}

\end{proposition}
The optimal mechanism is a cancellable trial in a degenerate form because the trial length is $T$. This is intuitive as there is no good news and the uninformed buyer is the one who has the highest willingness to pay. As we show in the proof, when the buyer reports the arrival of the bad signal, they quit the service and get a maximum refund subject to the IC constraint:
\begin{gather*}
     \int_t^T u-\lambda l(1-\mu_s) \dd s. 
\end{gather*}
$\mu_s$ denotes the posterior of $\theta$ if there is yet no signal arrived by time $s$.
The refund is maximized in the sense that any further increase in refund makes it strictly profitable for an uninformed buyer to falsify a bad signal.


\paragraph{Mixed News Model}
The service delivers an instantaneous lump sum reward at rate $\lambda I$ for both types $\theta=H,L$. Thus, in the absence of reward, the belief remains unchanged. If the match quality is high, the lump sum reward is $\bar{v}$ and if the match quality is low, the lump sum reward is $\ubar{v}$.  To illustrate the new force introduced by the bad signal, we focus on cases where $\bar{v}>0>\ubar{v}$ and normalize the expectation to be 1 so $\mu_0 \bar{v} + (1-\mu_0) \ubar{v}=1$.  This means that low match quality hurts the buyer, yet in expectation, the service still has a positive value.

We use $I_U(t)$ to denote the access for the uninformed buyer who reports no signal up to time $t$, and use $I^{\bar{v}}(t),I^{\ubar{v}}(t)$ to denote the access for the informed buyer who reports receiving a good/bad signal. Again, we directly look at the IC-IR mechanism that maximizes the linear-weighted-revenue of sellers:
\begin{gather*}
  max_{M\in \mathcal{M}} \quad  w_L \pi_L(M) + \pi_H(M) ,\\
  s.t. \quad   M \text{ is  IC-IR}.
\end{gather*}
\begin{proposition}
\label{prop:mixnews}
    In the mixed news model, for any $w_L\in[-1,(1-\mu_0)/\mu_0]$, the cancellable trial
\begin{align*}
I_U(t) = 1_{t\leq t_0(w_L)},  ~ ~ 
 I^{\bar{v}}(t)  = 1, ~ ~   I^{\ubar{v}}(t)  = 0  \quad \forall t,
\end{align*}
 is optimal, where $t_0(w_L)$ is the trial length.
\end{proposition}

Dealing with bad news, or a mixed news model with greater generality creates more challenges in tractability because the uninformed buyer's downward deviation is now binding as the designer tries to maximize refund. This means the designer has to track the buyer's posterior as well as their future option value, which is affected by future allocation and by the future refund recursively. We leave it for future research.

\subsubsection{Infinite Horizon with Discounting}
\label{sec:infinite}
We can extend the main insights to a model where there is an infinite horizon, but both the buyer and seller discount payoffs at rate $r$. 

\begin{proposition}
\label{prop:infinite_horizon}
Suppose that instead of a finite horizon with no discounting, both buyer and seller discount future payoffs at a rate $r$ and experimentation takes place over an infinite horizon. Define $\pi_0$, $v_0$ as in Theorem \ref{thm:trial}.  Then any mechanism maximizing weighted revenue with weight $(w_L, 1)$ is once again a trial mechanism, with length: 
\begin{equation}
    \label{eqn:infinite_horizon_trial}
    t_0 = \begin{cases}
        \frac{1}{\lambda}\ln \left( \frac{\frac{\lambda}{r} + 1 }{1 - \mu_0(1 + w_L)/\pi_0}\right) & \mu_0(1 + w_L) \le \pi_0 \\
        \infty & \mu_0(1 + w_L) \ge \pi_0
    \end{cases}
\end{equation}
Further, there exist prices to support this trial mechanism.
\end{proposition}
That is, there is nothing particularly special about the finite-time horizon with no discounting. Proposition \ref{prop:infinite_horizon} only generalizes the main result, but generalizations of the other results follow straightforwardly.

\subsection{Discussion}
\paragraph{Other Trial Mechanisms}
    Our analysis in Section \ref{sec:mechanisms} is robust to minor adjustments to our baseline model; we can extend Theorem \ref{thm:trial} to show that trial (or trial-like) mechanisms are the optimal IC-IR mechanisms in other variations of the model with infinite horizon, imperfect seller information, seller service cost, buyer flow utility, and other buyer learning processes. An interested reader may find the details in the working paper.

\paragraph{The Value of Data}
To isolate the novel economic channel for how having consumer data can negatively impact seller revenue, this paper focuses on an extreme environment where data has no positive value for implementing social efficiency. There are many simple ways to augment our model to give data positive social value. 
For example, the seller may have multiple services, and data allows the seller to determine which service to offer. In such settings, whether or not collecting data is beneficial depends on the tradeoff between efficiency gains and the possible revenue loss due to signaling forces identified in this paper. Broadly, we have suggested that the question of whether to collect consumer data is nuanced; the question of what noisy consumer data to collect in the presence of imperfect commitment on data usage is an interesting question for future research.

\paragraph{Moral Hazard}
Our model of informed principal assumes the distribution of $\theta$ is exogenously given. However, the welfare judgment on dynamic tiered pricing mechanisms can be ambiguous in a moral hazard problem where the distribution of $\theta$ is endogenous; for example, the seller may need to make a costly private investment, which increases the probability of $\theta = H$. Our dynamic mechanism design analysis could be applied to study how to provide proper incentives by committing to trials/dynamic tiered pricing.

\paragraph{Informed Data Seller}
The informed principal framework emphasizes the limited commitment power of market participants by assuming they can commit to tangible pricing mechanisms but cannot credibly communicate information obtained from unpublished data. Further development of the informed principal framework could deepen understanding of data markets themselves. 
For example, consider a setting where a tech company wants to purchase datasets to train an algorithm or an AI model. The data owner is essentially an informed seller of its data, who faces a similar tradeoff to the seller in our model. Providing cheap access to data can prove its quality, but may reduce the value of the remaining data; thus, another interesting direction for future research could explore how the informed seller and dynamic mechanism forces we consider interact when the service induces richer externalities on future service value.

\bibliographystyle{econometrica} 
\bibliography{bibliography}


\appendix

\newpage

\section{Formal Treatment of Refinements in Section 5}
\label{appendix:Refinment}
First, we formally introduce the notion of the D1 criterion. 
For any off-path mechanism $m=(I,p)$, and any market belief $\mu\in[0,1]$, the continuation game is a game where first the buyer decides whether to accept the mechanism, and then both parties report messages. An equilibrium of the continuation game consists of both parties' strategies and the belief of the buyer, which is obtained by Bayes rule given the belief $\mu$, updating after the proposal of the mechanism. Denote $\Sigma(m,\mu)$ as the set of continuation equilibria after the proposal of $m$ and subsequent buyer posterior belief $\mu$. For any $\sigma \in \Sigma(m,\mu)$ denote $\pi_x(\sigma)$ as the ex-ante payoff of the seller of type $x$. Denote
\begin{align*}
     \Sigma(m,H,\pi_H) & = \{ \sigma  |   \pi_H(\sigma) > \pi_H, ~ \sigma \in \Sigma(m,\mu) \text{ for some } \mu\in[0,1]\},   \\
     \Sigma(m,L,\pi_L) & = \{ \sigma  |   \pi_L(\sigma) > \pi_L, ~ \sigma \in \Sigma(m,\mu) \text{ for some } \mu\in[0,1] \}.
\end{align*}
The set $\Sigma(m,H,\pi_H)$ denotes the continuation equilibria that increase the high-type seller's payoff, relative to payoff $\pi_H$. 
Below, we adapt the \cite{BanksSobel87} version of the D1 criterion to consider continuation equilibria instead of best-responses:
\begin{definition}
    An equilibrium with equilibrium payoff $(\pi_L,\pi_H)$ does not survive D1 Criterion if and only if there exists a deviating mechanism $m$ such that $\Sigma(m,H,\pi_H) \supsetneq \Sigma(m,L,\pi_L)$ and for any $\sigma \in \Sigma(m,1)$, $\pi_H(\sigma) > \pi_H$.
\end{definition}
Recall that in \cite{ck87}, given an off-path message $m$, the receiver should eliminate type $\theta$ by D1 intuitively if another type $\theta'$ benefits ``more often'' (i.e. for more actions) than $\theta$. Here, we instead say that under our D1 notion, when facing the deviating mechanism $m$, the buyer should eliminate type $L$ if the set of \textit{continuation equilibria} where $L$ weakly benefits is strictly contained in the set of continuation equilibria where $H$ benefits.

To prove Proposition \ref{prop:D1}, we proceed in two steps.
First, we consider the equilibrium payoffs that do not survive the D1 Criterion. For any $\pi_L \in [0,\lambda \mu_0T]$, denote $\Pi(\pi_L)$ as the upper bound of the high-type seller's equilibrium payoff when the low-type gets $\pi_L$:
\begin{gather*}
    \Pi(\pi_L) = \max\{ \pi_H |  (\pi_L,\pi_H) \in \Pi_\ICIR \}.
\end{gather*}
Also define
\begin{gather*}
    \pi_L^D = \lambda \mu_0 t_F + \mu_0 (1-e^{-\lambda t_F})\lambda (T - t_F)\int_{v_F}^{\bar{v}} (v-v_F) \dd v.
\end{gather*}
The value $\pi_L^D$ is the maximum price that the buyer is willing to pay ex-ante for a trial mechanism with the same trial length $t_M$ and post-trial price $v_M$ as the Myersonian free trial characterized by point $F$.
\begin{lemma}
\label{lemma:obviouslynotD1}
An equilibrium with payoff $(\pi_L,\pi_H)$  does not survive D1 Criterion if $\pi_L > \pi_L^D$ or $\pi_H < \Pi(\pi_L)$.
\end{lemma}

\begin{proof}[Proof of Lemma \ref{lemma:obviouslynotD1}]

To prove Lemma \ref{lemma:obviouslynotD1}, it is sufficient to focus on a restricted class of deviating mechanisms: the trial mechanisms. Because the seller's report does not matter in a trial mechanism, the continuation equilibria of the mechanism reduces to the rational responses of the buyer. Furthermore, the buyer's rational response upon receiving a Poisson signal in the trial is essentially unique: procuring the post-trial service if and only if the price is less than the value. Thus, the choices for the buyer revolve around 1) whether to engage in the trial and 2) whether to acquire the post-trial service when no signal is received during the trial.

Consequently, we denote the strategy space as a finite set $\Sigma^b_{tr} = {(B,B), (B,N),N}$, where $(B,B)$ implies purchasing the trial and the post-trial service even in the absence of signals, $(B,N)$ implies purchasing the trial but not the post-trial service in the absence of signals, and $N$ implies purchasing nothing. We define
\begin{align*}
    \Sigma_{tr}(m,\mu) &= \{\sigma \in \Sigma^b_{tr} | ~ \sigma \text{ is optimal with belief }\mu \text{ when the trial mechanism m is proposed}  \}, \\
     \Sigma_{tr}(m,H,\pi_H) & = \{ \sigma  |   \pi_H(\sigma) > \pi_H, ~ \sigma \in \Sigma_{tr}(m,\mu) \text{ for some } \mu\in[0,1]\},   \\
     \Sigma_{tr}(m,L,\pi_L) & = \{ \sigma  |   \pi_L(\sigma) > \pi_L, ~ \sigma \in \Sigma_{tr}(m,\mu) \text{ for some } \mu\in[0,1] \}.
\end{align*}

We prove the lemma in two steps. First, we show that if $\pi_L > \pi_L^D$, the equilibrium does not survive D1. Second, we show that if $\pi_H < \Pi(\pi_L)$, the equilibrium also does not survive D1. These two observations together imply the result.

First, consider any equilibrium with equilibrium payoff $(\pi_L,\pi_H)$ such that $\pi_L > \pi_L^D$. Recall there is a trial mechanism that achieves $(\pi_L,\Pi(\pi_L))$, which sells the trial with length $t_0$ at a price $\pi_L$ and sells the post-trial service at the threshold $v_0$.
We know from the comparative statics that $t_0 > t_M$. Now consider the deviation to a trial mechanism $m'$ that is very close to this mechanism. The deviation sells the trial with length $t_0-\varepsilon$ at a price $\pi_L-\varepsilon^2$ and sells the post-trial service at the threshold $v_0$ for some $\varepsilon$ sufficiently small.

Under the deviating trial mechanism, if the buyer's responce is $N$, then seller of both types get $0 < \pi_L \leq \pi_H$, so both types get worse off and $N \not \in (\Sigma_{tr}(m,H,\pi_H) \cup \Sigma_{tr}(m,L,\pi_L)) $. If the buyer's responce is $(B,B)$, then seller of both types get $\pi_L +v_0(T-t_0+\varepsilon) > \pi_H \geq \pi_L$, so both types get better off and  $(B,B) \in (\Sigma_{tr}(m,H,\pi_H) \cap \Sigma_{tr}(m,L,\pi_L)) $. Lastly, if the buyer's response is $(B,N)$, then low-type seller gets $\pi_L-\varepsilon^2 < \pi_L$, while the high-type seller gets $\pi_L -\varepsilon^2 + (1-e^{-\lambda (t_0-\varepsilon)}) (T-t_0+\varepsilon) \mathbb{P}(v> v_0)v_0 .$
In contrast, $\Pi(\pi_L) =\pi_L  + (1-e^{-\lambda t_0}) (T-t_0)  \mathbb{P}(v> v_0) (v_0).$
Because the function $(T-t)(1-e^{-\lambda t})$ is a strictly concave function maximized at $t_M$ and $t_0>t_M$, the high-type seller gets better off. Thus, $(B,N) \in \Sigma_{tr}(m,H,\pi_H) $ but $(B,N) \not \in \Sigma_{tr}(m,L,\pi_L)$. Because for any $\mu=\mu_0$, $(B,N) \in \Sigma_{tr}(m,\mu_0)$, we have proved $\Sigma_{tr}(m,H,\pi_H) \supsetneq \Sigma_{tr}(m,L,\pi_L).$

In addition, note that $\Sigma_{tr}(m,1) \subseteq \{ (B,N),(B,B) \}$ so $\pi_H(\sigma)> \pi_H$ for any $\sigma \in \Sigma_{tr}(m,1)$. These two points imply that $(\pi_L,\pi_H)$ does not survive D1 Criterion if $\pi_L > \pi_L^D$.

Second, for any equilibrium with equilibrium payoff $(\pi_L,\pi_H)$ such that $\pi_H < \Pi(\pi_L)$, recall that there is a trial mechanism that achieves $(\pi_L,\Pi(\pi_L))$, which sells the trial with length $t_M$ at a price $\pi_L$ and sells the post-trial service at a price $v_M(T-t_M)$.
Now consider the deviation to a trial mechanism $m$ that is very close to this mechanism. The deviation sells the trial with length $t_M$ at a price $\pi_L-\varepsilon$ and sells the post-trial service at a price $v_M(T-t_0)$ for some $\varepsilon$ sufficiently small. By the same argument as above, $\Sigma_{tr}(m,H,\pi_H) \supsetneq \Sigma_{tr}(m,L,\pi_L)$ for $\pi_L > 0$. There is one difference at the corner case where $\pi_L=0$: $N \not \in \Sigma_{tr}(m,\mu)$ for any $\mu$. In this case, it is clear that $\Sigma_{tr}(m,1) \subseteq \{ (B,N),(B,B) \}$ so $\pi_H(\sigma)> \pi_H$ for any $\sigma \in \Sigma_{tr}(m,1)$. Thus, in both cases,  $(\pi_L,\pi_H)$ does not survive D1 Criterion when $\pi_H < \Pi(\pi_L)$.
\end{proof}


Our second step is to prove any equilibrium payoff that does not fail the conditions of Lemma \ref{lemma:obviouslynotD1} survives the D1 criterion. In fact, we can prove a stronger result:
\begin{lemma}
\label{lemma:obviouslyD1}
   Any equilibrium with payoff $(\pi_L,\pi_H)$ such that $\pi_L \le \pi^D_L$ and $\pi_H = \Pi(\pi_L)$ survives the D1 criterion. Further, there does not exist a deviating mechanism $m$, a belief $\mu$ and a continuation equilibrium $\sigma \in \Sigma(m,\mu) $ such that $\pi_H(\sigma) > \pi_H$ and $\pi_L(\sigma) \leq \pi_L.$
\end{lemma}

\begin{proof}[Proof of Lemma \ref{lemma:obviouslyD1}]
    We prove this by contradiction. Suppose there exists a $\pi_L \in [0,\pi^D_L]$ such that $(\pi_L,\Pi(\pi_L))$ does not survive D1. Then there exists a deviating mechanism $m$, a belief $\mu$ and a continuation equilibrium $\sigma \in \Sigma(m,\mu)$ such that $\pi_H(m,\mu) > \Pi(\pi_L)$ and $\pi_L(m,\mu) \leq \pi_L$.
    This implies $E(\pi_L)>\Pi(\pi_L)$ where $E(\pi_L)$ is the solution of the following problem
\begin{gather*}
     E(\pi_L)= \max_{q, p_U^\theta (T), \{ \Delta_{v,t}^\theta \} } ~ p^H_U(T) + \mathbb{E} \left[\Delta_{v,\tau}^H \mid \theta = H \right] \notag \\
    \textnormal{s.t. }  \text{(IC-V) + (IC-U)}  \\
         \mu p_U^H(T) + (1-\mu) p_U^L(T)   \le \mu I(U,0) + \mu \\
         \mathbb{E}\left[  v\left( I\left(v,\tau \right) -I(U,\tau) \right)  \left.  - \Delta_{v,\tau} \right\vert \theta = H\right] , \\
         p_U^H(T) \le p_U^L(T)  \le \pi_L.
\end{gather*}
In this problem, we solve for the optimal mechanism that maximizes the high-quality seller subject to the additional constraint that the buyer believes the seller is of high quality with probability $\mu$, and we require the payoff to the lower-quality seller to be bounded by $\pi_L$.
Denote the solution to the above problem as $m^*$.  Clearly, 
\begin{gather}
\label{eq:proof:D1}
     \mathbb{E} \left[\Delta_{v,\tau}^H \mid \theta = H \right]|_{m^*} = E^*(\pi_L)_{m^*} - p_U^H(T)|_{m^*}  > \Pi(\pi_L) - \pi_L.
\end{gather}

On the other hand, there is a trial mechanism  that achieves $(\pi_L,\Pi(\pi_L))$ and solves
\begin{gather*}
    \max_{q, p} ~  - p^L_U(T) +  p^H_U(T) + \mathbb{E} \left[\Delta_{v,\tau}^H \mid \theta = H \right]  \\
    \text{subject to IC, IR} 
\end{gather*}
Moreover, our analysis in Lemma \ref{lem:uninformed_payments_equal} has shown that when $w_L = -1$, the ex-ante IR constraint can be neglected and $p_U^H(T) = p_U^L(T)$. Thus,
\begin{gather*}
    \Pi(\pi_L) - \pi_L =  \max_{q, p_U^\theta (T), \{ \Delta_{v,t}^\theta \} } ~ \mathbb{E} \left[\Delta_{v,\tau}^H \mid \theta = H \right]   \\
   {s.t. } \quad 
        \text{(IC-V) + (IC-U)} .
\end{gather*}
Note that $m^*$ is a feasible solution to this relaxed problem, so $\mathbb{E} \left[\Delta_{v,\tau}^H \mid \theta = H \right]|_{m^*} \leq  \Pi(\pi_L) - \pi_L$.
This contradicts \ref{eq:proof:D1}. Thus, we proved that any equilibrium with payoff $(\pi_L,\Pi(\pi_L))$ for $\pi_L \in [0, \pi_L^D]$ survives the D1 Criterion. 
\end{proof}
We can now finish the proof of Proposition \ref{prop:D1}.
\begin{proof}
Lemmas \ref{lemma:obviouslynotD1} and \ref{lemma:obviouslyD1} jointly imply that the D1-surviving equilibria must have payoff outcomes $(\pi_L,\Pi(\pi_L))$ for $\pi_L \in [0, \pi_L^D]$. 

By Theorem \ref{thm:trial}, any IC-IR mechanism that maximizes $\pi_H-\pi_L$ must be a trial mechanism with trial length $t_M$ and post-trial price $v_M$. Further, in any IC-IR mechanism that maximizes $\pi_H - \pi_L$, $p_U(T)$ can take any value between 0 (by low-type seller IR) and $\pi_L^D$ (by buyer's IR). Outcome-uniqueness thus implies that any D1-surviving equilibrium outcome must mean the seller proposed a trial mechanism with trial length $t_M$ and post-trial price $v_M$.
\end{proof}

\section{Omitted Proofs}
\subsection{Omitted Proofs for Section 3 and 4}

\begin{proof}[Proof of Proposition \ref{prop:benchmark}]
Consider the candidate equilibrium where the seller proposes the mechanism with $I(h_t) = 1$ and $p(h_t) = \lambda \mu_0 T$ for all report histories $h_t$. The buyer's beliefs in this equilibrium are that $\theta = H$ with probability $\mu_0$ on-path prior to any message history (with beliefs updated according to Bayes' rule during the game), and that off-path $\theta = H$ with probability zero (unless a reward arrives in some history).\footnote{We did not specify the buyer's mechanism participation strategy off-path, but we suppose the buyer takes some arbitrary best-response on any off-the-equilibrium history given the belief system. This turns out to be sufficient for our argument.} 

    We first argue that if the seller can only propose mechanisms where $p(h_t)=p(\theta)$, then this is an equilibrium. Since the buyer gets ex-ante a surplus of $\lambda \mu_0 T$ from the service, the buyer's participation condition is satisfied, and the buyer's strategy is a best-response. Additionally, the buyer's belief system is consistent, so it suffices to check that the seller's proposal of such a mechanism is a best-response to the buyer's strategy and belief system. Note that on-path, the seller (of either type) gets $\lambda \mu_0 T$. Suppose the seller has a profitable deviation. Since $p$ is constant over time, the deviating mechanism proposed by the buyer must charge some $p > \lambda \mu_0 T$, else the seller could not benefit. But the buyer's belief system dictates $\mu(H|\cdot)$ for such a mechanism proposal is zero, and hence the buyer is unwilling to participate and buy the service because the buyer believes that $\theta = L$ almost surely. This implies that the seller could only get zero by proposing this mechanism, a contradiction.

    Now, suppose the seller does not observe $\theta$ or can commit to not observe $\theta$. Once again, we can confirm that the buyer's strategy is a best-response and the belief system is consistent with the candidate seller's strategy. We confirm that the seller's strategy is also a best-response. Since the seller does not observe $\theta$ (or commits to not observing $\theta$), the seller's strategy must be independent of $\theta$. Note that the candidate equilibrium would grant the seller a payoff of $\lambda \mu_0 T$, regardless of type. Suppose the seller could do better. That implies some mechanism, that would be accepted, has some price function such that 
    $\mathbb{E}_\theta[\mathbb{E}_{h_T}[p(h_T)]] > \lambda \mu_0 T$,
    where the expectation is taken with respect to the reward arrival process and the buyer's strategy. But the buyer's expected payoff at the terminal time is 
    \begin{align*}
        \mathbb{E}_{\theta,h_T}[N_t v - p(h_T)]]  &=  \mathbb{E}_{\theta,h_T}[N_t v] - \mathbb{E}_{\theta,h_T}[p(h_T)]] < \mu_0 \mathbb{E}[v] \mathbb{E}[N_t \mid \theta = H] - \lambda \mu_0 T \\
        &= \mu_0 \mathbb{E}_{I}\left[ \int_0^T \lambda I_s \ ds \right] - \lambda \mu_0 T 
    \le \lambda \mu_0 T - \lambda \mu_0 T = 0
    \end{align*}  
    where we used the fact that $v$ is drawn independently and with expectation 1, and the fact that $I$ is always bounded above by $1$. But this implies the buyer has a negative ex-ante payoff from participating, contradicting the buyer's initial participation constraint.
\end{proof}
\begin{proof}[Proof of Proposition \ref{prop:free_trial_deviation}]
We first prove the if direction. Suppose \eqref{eqn:fb_condition} holds. The probability a buyer experiences a reward during the trial is $(1 - e^{-\lambda t_M})$, and the probability that the buyer has a value larger than $v_M$ is $1 - F(v_M)$. The expected ex-ante profit of a high-type seller is then $(1-e^{-\lambda t_M})(1 - F(v_M))\lambda v_M(T-t_M)$. 
By \eqref{eqn:fb_condition}, this is larger than $\lambda \mu_0 T$. From Corollary \ref{corr:exante}  we know any equilibrium that maximizes the ex-ante revenue leads to equilibrium payoff $(\lambda \mu_0 T,\lambda \mu_0 T)$ and so the Myersonian free trial is always a profitable deviation.

We next prove the only if direction. Suppose \eqref{eqn:fb_condition} does not hold. It suffices to show that there is an equilibrium with the first-best revenue. Selling full-service access at the beginning at price $\lambda \mu_0 T$  is an IC-IR mechanism with payoffs $(\lambda \mu_0 T, \lambda \mu_0 T)$. Since \eqref{eqn:fb_condition} does not hold, the payoff pair $(\lambda \mu_0 T, \lambda \mu_0 T)$ satisfies the conditions of Proposition \ref{prop:payoff_set} and hence is an equilibrium. Thus, there is an equilibrium that attains the first-best.
\end{proof}

\begin{proof}[Proof of Lemma \ref{lem:uninformed_payments_equal}]

First, we rearrange the (IR-0) constraint. Intuitively, the (IR-0) constraint equivalently states that the expected payment that occurs with no arrivals cannot be more than the expected value of  the service for the uninformed buyer plus the expected additional value a buyer gains after learning $v$:
\[ \mu_0\left( \lambda I(U,0) + \mathbb{E}\left[ \lambda v (I(v,\tau) - I(U,\tau)) - \Delta^H_{v,\tau}(T) \mid \theta = H \right] \right)  \ge \mu_0 p^H_U(T) + (1-\mu_0) p^L_U(T). \]
Recall the objective function is
\begin{align}
    \max_{M \in \mathcal{M}} ~ & \left \{ w_L p^L_U(T) + p^H_U(T) + \mathbb{E} \left[\Delta^H_{v,\tau}(T)\mathbbm{1}[\tau \le T]\mid \theta = H \right] \right \} \notag 
\end{align}

We consider casework on $w_L$. If $w_L = -1$, then $p^L_U(T)$ must be chosen equal to $p^H_U(T)$, else we could decrease $p^L_U(T)$ and do better on the objective without affecting  the other constraints and maintaining $p^L_U(T) \ge p^H_U(T)$. So $p^L_U(T) = p^H_U(T)$ in this case.
    
If $w_L \in \left(-1, (1-\mu_0)/\mu_0\right]$ then note that (IR-0) is a linear constraint on $p^H_U(T)$ and $p^L_U(T)$, and the objective is linear in both. Since $w_L \le (1-\mu_0)/\mu_0$, the objective assigns a higher relative weight to $p^H_U(T)$, and hence if $p^L_U(T) > p^H_U(T)$, the objective would increase by decreasing $p^L_U(T)$ by $\mu_0 \epsilon$ and increasing $p^H_U$ by $(1-\mu_0)\epsilon$, while maintaining (IR-0); thus, at optimum, $p^L_U(T) = p^H_U(T)$. It is then clear to see that it is optimal to increase $p^H_U(T)$ as much as possible since $w_L > -1$, and so (IR-0) must bind.
\end{proof}

\begin{proof}[Proof of Lemma \ref{Lemma_low_bind}]
First note that if  (IC-V) is slack at any point $v$, then one can uniformly increase $\Delta_{v,t}^H(T)$, which does not affect (IC-U), to increase the objective. Thus, in the optimal solution (IC-V) binds at some point $v_1$. 

Suppose by contradiction that there exists $v$ such that $I(v,t) < I(U,t)$. Note that $I(v,t)$ is increasing in $v$ so $u(v,t)$ is convex in $v$. Recall the convex function $u(v,t)$ is supported by the linear function $\lambda v I(U,t)$ and they are equal at $v_1$. This implies $I(v,t)\leq I(U,t)$ for any $v< v_1$ and $I(\ubar{v},t) < I(U,t)$.

Following the standard envelope argument, we can replace the (IC-V) with the requirement that $I(v,t)$ is monotone in $v$ and the envelope representation:
\begin{gather*}
    u(v,t) = u(v_1,t) + \int_{v_1}^v \lambda I(w,t) \dd w,
\end{gather*} 
and a mechanism satisfies (IC-V) and (IC-U) if and only if it satisfies (IC-U), the monotone condition, and the envelope representation.
We construct a new mechanism $I^*(v,t)$ and $\Delta_{v,t}^*(T)$ such that
\begin{align*}
    I^*(v,t) &= I(v,t), \quad \Delta_{v,t}^*(T) = \Delta_{v,t}^H(T) \qquad &\forall v \geq v_1, \\
    I^*(v,t) &= I(U,t), \quad \Delta_{v,t}^*(T) = 0 \qquad &\forall v < v_1.
\end{align*}
The new construction ensures for $v<v_1$
\begin{gather*}
    u^*(v,t) = \lambda v I(U,t) = \lambda v_1 I(U,t) - \lambda (v_1-v) I(U,t) = u^*(v_1,t) + \int_{v_1}^v \lambda I^*(w,t) \dd w
\end{gather*}
Thus, both envelope representation and (IC-U) hold for $v<v_1$. Because for $v>v_1$ nothing changes and the monotone condition is preserved, the new mechanism satisfies all IC constraints. In addition, note that in the original mechanism, for $v<v_1$
\begin{gather*}
    \Delta_{v,t}^H(T) = \lambda v I(v,t) - u(v,t) \leq \lambda v I(U,t) - \lambda v I(U,t) = 0,
\end{gather*}
and the inequality is strict at $v=\ubar{v}$. Thus, the new mechanism is a strict improvement.
\end{proof}

\begin{proof}[Proof of Lemma \ref{lemma:pluginPointWise}]
Recall the objective is
{\small \begin{align}
    & \begin{Bmatrix} (w_L+1)\mu_0\left(  \mathbb{E}\left[ \lambda v I(U,0) + \lambda( v (I(v,\tau) - I(U,\tau)))\mathbbm{1}[\tau \le T]\mid \theta = H \right]  \right)  \\ + (1 - \mu_0(w_L+1))\mathbb{E} \left[\left(\lambda v I(v,\tau) - \int_{\ubar{v}}^v \lambda I(w,\tau) \dd w - \lambda \ubar{v} I(U,\tau) \right)\mathbbm{1}[\tau \le T]\mid \theta = H \right] \end{Bmatrix} \notag 
\notag
\end{align} }

    When $v < v_0$, $I(v,t) = I(U,t)$, so the expectation terms are both  0.
When $v \ge v_0$, the expectation term is 
\begin{gather*} \mu_0(w_L+1)\lambda v \left( I(v,\tau) - I(U,\tau) \right) + (1 - \mu_0(w_L+1))\left(\lambda v I(v,\tau) - \int_{\ubar{v}}^v \lambda I(w,\tau) \dd w - \lambda \ubar{v} I(U,\tau) \right) \\
=  \mu_0(w_L+1)\lambda v \left( T - \tau - I(U,\tau) \right) + (1 - \mu_0(w_L+1))\left(\lambda v_0 (T - \tau) - \lambda v_0 I(U,\tau) \right) \end{gather*}
Hence, taking the expectation over $v$, the term becomes
\begin{gather*} \mu_0(w_L+1)\lambda \int_{v_0}^{\bar{v}}v  f(v)\ \dd v \left( T - \tau - I(U,\tau) \right) + \lambda v_0(1 - F(v_0))(1 - \mu_0(w_L+1))\left( T - \tau -  I(U,\tau) \right) \\
= \lambda \left(\mu_0(w_L+1) \int_{v_0}^{\bar{v}}v  f(v)\ \dd v  + v_0(1 - F(v_0))(1 - \mu_0(w_L+1))\right)(T - \tau - I(U,\tau) \end{gather*}
For notational convenience, denote the expected virtual surplus $\pi_0$ as 
\begin{align*} 
\pi_0 &:= \mu_0(w_L+1) \int_{v_0}^{\bar{v}}v  f(v)\ \dd v  + v_0(1 - F(v_0))(1 - \mu_0(w_L+1)) \\
&= \int_{v_0}^{\bar{v}}v  f(v)\ \dd v - (1 - \mu_0(w_L + 1)) \left(\int_{v_0}^{\bar{v}}v  f(v)\ \dd v - v_0(1 - F(v_0))\right) \\
&= \int_{v_0}^{\bar{v}}v  f(v)\ \dd v - (1 - \mu_0(w_L + 1)) \left(\int_{v_0}^{\bar{v}} (1 - F(v)) \ \dd v \right) \\
&= \int_{v_0}^{\bar{v}} \left(v - (1 - \mu_0(w_L+1)) \frac{1-F(v)}{f(v)} \right) f(v) \dd v 
\end{align*}
where we utilized the integration-by-parts identity in the third step.
Using this definition of $\pi_0$, the objective becomes 
\begin{align}
    \max_{M \in \mathcal{M}} ~ & \left \{ (w_L+1)\mu_0  \lambda I(U,0) + \lambda \pi_0 \mathbb{E}[(T - \tau - I(U,\tau)) \mathbbm{1}[\tau \le T] | \theta = H ] \right \}\notag 
\notag
\end{align}
\end{proof}
\begin{lemma}
The optimal control function for the control problem:
\begin{align*}
    \max_{U \in \mathcal{U}_{ad}} ~ &\left \{ \int_0^T \left(\lambda  \left(T - t - (X(T) - X(t))\right)\exp \left(- \lambda X(t)  \right) + K \right) \lambda U(t)  \dd t  \right \} \\
    \textnormal{subject to } ~ & \dot{X}(t) = U(t) \in [0,1], ~ X(0) = 0, ~ X(T) \textnormal{ free}
\end{align*}
where $K$ is some constant and $U(t)$ is measurable with respect to $t$, is
\[ U(t) = \begin{cases}
    1 & t \le t_0 \\
    0 & t > t_0
\end{cases} \]
where $t_0$ solves $ \lambda e^{-\lambda t_0}( T - t) - (1 - e^{-\lambda t_0}) + K = 0$ if $K < 1 - e^{-\lambda T}$, and $t_0 = T$ otherwise.
\label{lem:ctrl}
\end{lemma}
\begin{proof}[Proof of Lemma \ref{lem:ctrl}]
First, rearrange the objective:
\begin{align*}
    &\int_0^T \left(\lambda  \left(T - t - (X(T) - X(t))\right)e^{- \lambda X(t)}+ K \right) \lambda U(t)  \dd t \\
    =~&\int_0^T \left(\lambda  \left(T - t + X(t)\right)e^{- \lambda X(t)} + K \right) \lambda U(t)  \dd t - \lambda X(T) \left( 1 -  e^{- \lambda X(T)} \right)\\
    =~&\int_0^T \left(\lambda  \left(T - t \right)e^{- \lambda X(t)} - \left( 1 -  e^{- \lambda X(t)}\right) + K \right) \lambda U(t)  \dd t.
\end{align*}  
Now, we invoke the standard Pontryagin maximum principle arguments. For a more specific reference, we apply
Theorem 4.2.4 in \cite{AhmedWang21}  which proves there exists an optimal control, and apply Theorem 5.2.3 in \cite{AhmedWang21} to prove for
any optimal measurable control $U,X$, there exists a costate variable $\rho$ such that 
\begin{align}
    U(t) &\begin{cases}
        = 1 & J(t) > 0 \\
        \in [0,1] & J(t) = 0 \\
        = 0 & J(t) < 0
    \end{cases} \label{eqn:bangbang} \\
    J(t) &= \lambda\left( \lambda  \left(T - t \right)e^{- \lambda X(t)} - \left( 1 -  e^{- \lambda X(t)}\right) + K\right) - \rho(t) \label{eqn:ctrl_index}\\
    \dot{\rho}(t) &= -\left( \lambda  \left(T - t \right) + 1 \right) \lambda e^{- \lambda X(t)} \lambda U(t) \label{eqn:costate_evol} \\
    \rho(T) &= 0 \label{eqn:costate_upperbound} 
\end{align}
where \eqref{eqn:costate_evol} is the costate evolution, and \eqref{eqn:costate_upperbound} is the upper boundary constraint, \eqref{eqn:bangbang} and \eqref{eqn:ctrl_index} follow from Hamiltonian maximization.\footnote{The maximization technically should hold for $a.e.$ t, but for any optimal control and the corresponding costate, we can always find essentially the same control such that the optimality condition holds for any $t$ instead of $a.e.$ t. Thus, without loss, we assume the optimality condition holds for any $t$. }  Note that by differentiating \eqref{eqn:ctrl_index} wherever differentiable, we get 
\begin{align*}
    \dot{J}(t) &= - \lambda\left(\lambda e^{- \lambda X(t)}- \lambda U(t) \lambda  \left(T - t \right)e^{- \lambda X(t)} - \lambda U(t) e^{- \lambda X(t)}\right) - \dot{\rho}(t) = -\lambda^2 e^{-\lambda X(t)},
\end{align*}
since \eqref{eqn:costate_evol} cancels the other terms. Note that this implies that $J$ is always decreasing; since $X$ is continuous, $\dot{J}$ is continuous. 

First, if $K \ge 1 - e^{-\lambda T}$, then $J(T) = \lambda (K - (1 - e^{-\lambda T}) \ge 0$, so $J$ is nonnegative. Hence, the unique solution to the Pontryagin conditions is $U(t) = 1$ everywhere; since Pontryagin is necessary, this must be the only solution to the control problem.

Now, consider the case where $K < 1 - e^{-\lambda T}$. First, we show there exists a unique $t_0$ such that $J(t_0) = 0$. Consider $J(0) = \lambda (\lambda T + K) - \rho(0)$. If this is positive, the fact that $J(T) = \lambda (K - (1 - e^{-\lambda T}) < 0$ and $J$ is strictly decreasing, implies that such a $t_0$ exists and unique by the intermediate value theorem. If $J(0) < 0$, then \eqref{eqn:bangbang} implies that $U(t) = 0$ always, and so $\rho(t) = 0$ by \eqref{eqn:costate_evol}. But this implies that $J(0) = \lambda (\lambda T + K) > 0$, a contradiction.

Given there exists a unique $t_0$ such that $J(t_0) = 0$ and $J$ is decreasing, it follows that $U(t)$ is zero for $t > t_0$. This implies that $\dot{\rho} = 0$ for $t > t_0$, and combined with \eqref{eqn:costate_upperbound}, this implies that $\rho(t_0) = 0$. Additionally, $U(t)$ must be 1 for $t < t_0$, and so $X(t) = t$ for $t < t_0$, and $X(t) = t_0$ for $t \ge t_0$. Hence, $t_0$ must satisfy:
\begin{align*} J(t_0) = \lambda\left( \lambda  \left(T - t_0 \right)e^{- \lambda X(t_0)} - \left( 1 -  e^{- \lambda X(t_0)}\right) + K\right) - \rho(t_0) &= 0  \\
\lambda  \left(T - t_0 \right)e^{- \lambda t_0} - \left( 1 -  e^{- \lambda t_0}\right) + K &= 0
\end{align*}
Hence, the Pontryagin conditions once again admit a unique solution, and so this is the optimal control.
\end{proof}

\begin{proof}[Proof of Lemma \ref{lem:necessity_high}]
    Lemma \ref{lem:ctrl} and the main analysis implies that the necessary and sufficient condition for the optimal solution to the control problem \eqref{prblm:ctrl} derived from Lemma \ref{lemma:pluginPointWise} is that $I^H$ has the trial mechanism form with $v_0$ given by \eqref{eqn:thm_v0} and $t_0$ given by \eqref{eqn:thm_p0}. 
    Given $I^H$ and working backwards through the analysis, the envelope representation of (IC-V) and Lemma \ref{Lemma_low_bind} imply that $I^H$ thus determines the interim utility $u$ uniquely, and so we can back out that $\Delta^H_{v,t}(T) = \lambda v_0 (T - t_0).$ 
    From Lemma \ref{lem:uninformed_payments_equal}, (IR-0) binds when $w_L > -1$; thus, when $w_L > -1$, \[ p^H_U(T) = p^L_U(T) = \lambda \mu_0 t_0 + \lambda (\pi_0 - v_0(1 - F(v_0)))(T - t_0)(1 - e^{-\lambda t_0}).\] 
    When $w_L = -1$, $p^H_U(T)$ cancels, so (IR-0) need not bind, so Lemma \ref{lem:uninformed_payments_equal} implies
    \[ p^H_U(T) = p^L_U(T) \in \left[0,  \lambda \mu_0 t_0 + \lambda (\pi_0 - v_0(1 - F(v_0)))(T - t_0)(1 - e^{-\lambda t_0})\right].\] 
    Thus, the price $p_U(T) := p^H_U(T) = p^L_U(T)$ charged at the terminal time $T$ must satisfy \eqref{eqn:thm_p0}. Now, since the relaxed problem only optimized over prices at the \textit{terminal} time $T$, the trial mechanism in the lemma statement must be optimal, since $I$ maximizes the control problem and the price scheme induces the optimal terminal prices.

    Finally, we argue outcome-uniqueness. The only parts of the mechanism not pinned down by our analysis of the relaxed problem are $I_U^L(t)$, $p_{v,t}^L(t)$, and $p_U(t) = p_U^L(t) = p^H_L(t)$  for $t<T$. $I_U^L(t)$ does not affect any player's payoff and also does not affect the arrival of signals; $p_{v,t}^L(t)$ is completely off-path as low-type seller does not generate signal; $p_U(t)$ for $t<T$ does not affect on-path interim payoffs as what ultimately matters is the final payment.
    
    Thus, interim payoffs of all parties are unique at any time $t$. More specifically, in any optimal mechanism, the interim payoffs are uniquely pinned down as follows:
    \begin{itemize}
        \item The low-type seller always gets $p_U(T)$ in the end.
        \item At any time $t$, with probability $1-e^{\lambda \min\{ t_0,t\}}$, the buyer has received the signal. The high-type seller gets $p_U(T) + \lambda v_0 (T - t_0)$ and the buyer enjoys full service and pays $p_U(T) + \lambda v_0 (T - t_0)$.
        \item At any time $t$,  with probability $e^{\lambda \min\{ t_0,t\}}$, the buyer has not received the signal. In this case, with probability $(1-e^{\lambda \max \{t_0 - t,0 \}  })$ the high-type seller gets $p_U(T) + \lambda v_0 (T - t_0)$ and the buyer enjoys full service and pays $p_U(T) + \lambda v_0 (T - t_0)$. With probability $e^{\lambda \max \{t_0 - t,0 \}  }$ the high-type seller gets $p_U(T)$ and the buyer enjoys service of length $t_0$ and pays $p_U(T)$.
    \end{itemize} 
\end{proof}

\begin{proof}[Proof of Theorem \ref{thm:trial}]
Lemma \ref{lem:necessity_high} showed the solutions to the \textit{relaxed} problem are only attained by trial mechanisms as defined in definition \ref{definition_trail}, with $v_0$ given by \eqref{eqn:thm_v0} and $t_0$ given by \eqref{eqn:thm_t0}. Thus, to finish the argument, it remains to show that for any $w_L$, the trial mechanism maximizes the relaxed problem is feasible in the original problem.

First, note that any trial mechanism is invariant to the seller's report and both sellers get a nonnegative revenue, so this is IC-IR for the seller. 
Additionally, the payment scheme (in absence of any rewards arriving) frontloads payment: $p_U$ is constant and has no dependence on $t$. This means the buyer pays a lump-sum price up front at time 0, and then nothing more if no rewards arrive. If a reward is reported, the payment scheme implies the buyer pays another (fixed) lump-sum price at the end of the trial, $t_0$, if and only if the value $v\geq v_0$ and no more payments thereafter. Hence, the buyer effectively only communicates with the mechanism at the beginning and at the end of the trial, so we only have to check IC and IR at the initial time and at the end of the trial. 

Recall that in Lemma \ref{lem:necessity_high}, $p_U$ was constructed to make (IR-0) bind, and hence IR is satisfied initially. To see that the mechanism is IC-IR at the end of the trial. Note that after receiving a reward of size $v$, the buyer knows the seller is high-type, so the post-trial service is worth $v\lambda (T - t_0)$ to the buyer, and the price for post-trial service is $v_0\lambda (T - t_0)$. Hence it is optimal for the buyer to buy post-trial service if and only if $v\geq v_0$. 

The only remaining constraint to be verified is that the buyer who never receives a reward does not want to purchase the post-trial service. Note that if the buyer does not receive any reward by the time $t_0$, the posterior of $\theta=H$ is
\begin{gather*}
    \mu(t_0) = \frac{\mu_0 e^{-\lambda t_0}}{\mu_0 e^{-\lambda t_0} + 1-\mu_0} = \frac{\mu_0 }{\mu_0 + (1-\mu_0) e^{\lambda t_0}}.
\end{gather*}
Thus, buying the post-trial service is a profitable deviation if and only if $v_0 < \mu(t_0)$. 
We will verify that for all trial mechanisms which are solutions to the relaxed problem, $\mu(t_0) \leq \bar{v}$.

To do so, note that $\mu(t_0)$ is decreasing in $t_0$, and by Proposition \ref{prop:comparative_private}, the trial length $t_0$ increases in $w_L$. So it suffices to show $\mu(t_0) \leq \ubar{v}$ when $w_L=-1$.

Recall that $t_0$ is the unique solution of
\begin{gather*}
     0 = \lambda e^{-\lambda t_0}(T - t_0) - (1 - e^{-\lambda t_0}).
\end{gather*}
Denote $x= e^{-\lambda t_0}$ so $t_0 = -\ln x/\lambda$. We have
\begin{gather*}
    0 = (\lambda  T +1 ) x +  x \ln x  - 1 \geq (\lambda  T +1 ) x + (x -1) -1,
\end{gather*}
where the inequality follows from $1-1/x\leq \ln x$. Consequently, $\frac{1}{x} \geq 1 + \frac{\lambda T}{2}$, and
\begin{align*}
 \implies \mu(t_0) &\leq \frac{\mu_0}{\mu_0 + (1-\mu_0) ( 1+ \frac{\lambda T}{2})} \\
 & \leq  \frac{\mu_0}{\mu_0 + (1-\mu_0) ( 1+ (\mu_0-\ubar{v})/(1-\mu_0)\ubar{v})}\\
 & = \frac{\mu_0 (1-\mu_0)\ubar{v}}{\mu_0 (1-\mu_0)\ubar{v} + (1-\mu_0) (\mu_0-\mu_0 \ubar{v}  )  }  = \ubar{v}.
\end{align*}
The second inequality comes from plugging in the first inequality, and the third inequality comes from the assumption that $T$ is sufficiently large.

To conclude, we have verified that for all trial mechanisms which solve the relaxed problem, $\mu(t_0) \leq \ubar{v}$. Because $v_0 \geq \ubar{v}$, a buyer who never receives a reward during the trial will not buy the post-trial service. 
Since the constructed trial mechanism is IC-IR for the original problem and attains the maximal value of the relaxed problem, it must be optimal.
\end{proof}
\begin{proof}[Proof of Proposition \ref{prop:comparative_private}]
If $v_0 = \ubar{v}$, it is trivially weakly decreasing in $\mu_0$ and $w_L$. Else, $v_0$ satisfies 
\[ v_0 - (1 - \mu_0(w_L+1)_+)\frac{1-F(v_0)}{f(v_0)} = 0\]
Note that regularity implies the left-hand side is increasing in $v_0$ (since regularity implies $(1 - \mu_0(w_L+1)_+)(v - (1 - F(v))/f(v))$ increases in $v$ and $\mu_0(w_L+1)_+ v$ is linearly increasing in $v_0$). Since the left-hand side is also increasing in $\mu_0$ and $w_L$, it follows that $v_0$ is decreasing in $\mu_0$ and $w_L$. For $t_0$, note that we can rewrite $\pi$ as follows:
\begin{gather*}
    \pi (\mu_0, w_L) = \max_{v'} \int_{v'}^{\bar{v}} v f(v) -(1-\mu_0(1 + w_L)) (1-F(v)) \dd v.
\end{gather*}
From the envelope theorem we know
\begin{align*}
    \frac{\partial \pi}{\partial \mu_0} = (1+w_L)\int_{v_0}^{\bar{v}} (1-F(v)) \dd v, && \frac{\partial \pi}{\partial w_L} = \mu_0\int_{v_0}^{\bar{v}} (1-F(v)) \dd v.
\end{align*}
Thus,
\begin{align*}
    \frac{\partial }{\partial \mu_0} \left(\frac{\pi}{\mu_0}\right) =  \frac{  \frac{\partial \pi}{\partial \mu_0} \mu_0 - \pi  }{\mu_0^2} &= \frac{1}{\mu_0^2} \left(  \int_{v_0}^{\bar{v}} \mu_0(1 + w_L)(1-F(v)) -v f(v) + (1-\mu_0(1+w_L)) (1-F(v)) \dd v  \right) \\
    &= \frac{1}{\mu_0^2}  \left(  \int_{v_0}^{\bar{v}} 1-F(v) -v f(v)  \dd v  \right) = -\frac{v_0(1-F(v_0))}{\mu_0^2} < 0 
\end{align*} 
So $\pi/\mu_0$ is decreasing in $\mu_0$, and therefore $\mu_0/\pi$ is increasing in $\mu_0$. Similarly, for $w_L$, 
\begin{align*}
    \frac{\partial }{\partial w_L} \left(\frac{\pi}{1+w_L}\right) &=  \frac{  \frac{\partial \pi}{\partial w_L}(1+w_L) - \pi  }{(1+w_L)^2} \\
    &= \frac{1}{(1+w_L)^2} \left(  \int_{v_0}^{\bar{v}} \mu_0(1 + w_L)(1-F(v)) -v f(v) + (1-\mu_0(1+w_L)) (1-F(v)) \dd v  \right) \\
    &= \frac{1}{(1+w_L)^2}  \left(  \int_{v_0}^{\bar{v}} 1-F(v) -v f(v)  \dd v  \right) = -\frac{v_0(1-F(v_0))}{(1+w_L)^2} < 0 
\end{align*} 
So $\pi/(1+w_L)$ is decreasing in $w_L$, and therefore $(1+w_L)/\pi$ is increasing in $w_L$. 
Note that by rearranging the condition for $t_0$, we get
\[ \lambda e^{-\lambda t_0} (T - t_0)  
 + \frac{\mu_0(1 + w_L)}{\pi_0} = 1 - e^{-\lambda t_0}.  \]
 The right-hand side is increasing in $t_0$, and the left-hand side is decreasing in $t_0$. Since $\mu_0(1+w_L)/\pi$ is increasing in $\mu_0$ and $w_L$, $t_0$ is also increasing in both.
\end{proof}

\begin{proof}[Proof of Proposition \ref{prop:comparative_v}]
To simplify notation, denote $A= 1-\mu_0(w_L+1)$. We can calculate
\begin{align*}
    v_0  &= \max \{  1-\delta, \frac{A}{1+A} (1+\delta) \} \\
    \pi_0 &= \int_{1-\delta}^{1+\delta} \left( v- A (1+\delta-v) \right) f(v) \dd v 
        & = \begin{cases}
              1-A + A (1-\delta) \quad &\text{if } v_0 =1-\delta, \\
              \frac{(1+\delta)^2}{4(1+A)\delta}  \quad &\text{if } v_0 =  \frac{A}{1+A} (1+\delta).
        \end{cases}
\end{align*}
One can verify $\pi_0$ is always decreasing in $\delta$. According to Theorem \ref{thm:trial}, either  $t_0=T$ or  
\begin{gather*}
         \lambda e^{-\lambda t_0}(T - t_0) - (1 - e^{-\lambda t_0}) + \frac{\mu_0(1 + w_L)}{\pi_0} = 0
\end{gather*}
It is straightforward to see that the LHS is decreasing in $t_0$ and decreasing in $\pi_0$.          
\end{proof}

\subsection{Omitted Proofs for Section 5}

\begin{proof}[Proof of Proposition \ref{prop:payoff_set}]
For the if direction, suppose $(\pi_L, \pi_H) \in \Pi_\ICIR$ with $\pi_H \ge \pi_F$. We construct an equilibrium that attains these payoffs. Since $(\pi_L, \pi_H )\in \Pi_\ICIR$, there exists some IC-IR mechanism $M$ with these payoffs. Consider the strategy profile where the seller always proposes $M$ and reports truthfully, the buyer always accepts $M$ and reports truthfully, and the buyer belief system is given by $\mu(H\mid M', \{m_s, I_s, p_s\}_{s<t}, N_t = 0) = 0$ for any mechanism $M' \neq M$.\footnote{As before, we did not specify the buyer's mechanism participation strategy off-path, but we suppose the buyer takes some arbitrary best-response on any off-the-equilibrium history given the belief system.} Since $M$ is IC-IR, it is weakly optimal for all players to report truthfully. Thus, to show this is an equilibrium, it suffices to show that no seller cannot gain from proposing any (possibly indirect mechanism) $M' \neq M$.\footnote{We do not assume $M'$ is a direct mechanism here to justify our discussion that restricting to direct mechanisms are without loss of generality.}

At the node where the deviating mechanism $M'$ is proposed, the buyer believes $\theta=L$ so we are considering a specific belief system where $\mu(H\mid M', \{m_s, I_s, p_s\}_{s<t}, N_t = 0) = 0$ and  $\mu(H\mid M', \{m_s, I_s, p_s\}_{s<t}, N_t > 0) = 1$.
Because the seller's strategy is sequentially optimal, after $M'$ is accepted, the seller's reporting strategy $\hat{\theta}$ is optimal given  the buyer's subsequent reporting strategy which is rational under the belief system. Thus, we can use direct mechanisms to discuss each seller type's maximum feasible payoff at the node where the deviating mechanism $M'$ is proposed. Consider the similar set of necessary conditions of IC-IR constraints as in Theorem \ref{thm:trial}
\begin{align*} \quad & 
    \Delta_{v,t}^H (T) =  \lambda v I(v,t) - \int_{\ubar{v}}^v \lambda I(w,t) \dd w - \lambda \ubar{v} I(U,t) \qquad \forall (v, t) \tag{IC} \\
    &  0  \ge  p^L_U(T) \tag{IR-0} \\
    & p^L_U(T) \ge \max(p^H_U(T), 0)  \tag{IC-S}
\notag
\end{align*}
This immediately implies the low-type gets at most $p_U^L(T) \leq 0 \leq \pi_L$, so low-type seller does not have strictly profitable deviation.


Now, consider the high-type seller.  Denote $M'$ grants ex-ante payoff $\pi_H'$ to the high-type seller. Because in any IC-IR direct mechanism $p_U^H(T) \leq p_U^L(T) \leq 0$, we must have:
\begin{align*}
    \pi_H' \leq \max_{M \in \mathcal{M}} ~ &   \mathbb{E} \left[\Delta^H_{v,\tau}(T)\mathbbm{1}[\tau \le T]\mid \theta = H \right]  \notag \\
    \text{subject to} \quad & 
    \Delta_{v,t}^H (T) =  \lambda v I(v,t) - \int_{\ubar{v}}^v \lambda I(w,t) \dd w - \lambda \ubar{v} I(U,t) \qquad \forall (v, t) \\
    & I(v,t) \geq I(U,t)
\end{align*}
One can check this is exactly the same optimization problem in Theorem \ref{thm:trial} for the case where $w_L=-1$. Thus, $\pi'_H \leq \pi_F$. This means that under the specific belief system, in which the buyer believes $\theta=L$ unless receiving a conclusive signal, the highest payoff of H/L seller is $(0,\pi_F)$ so no sellers want to deviate from equilibrium payoff $(\pi_L, \pi_H) \in \Pi_\ICIR$ with $\pi_H \ge \pi_F$.


For the only-if direction, by the inscrutability principle \citep{myerson83}, any equilibrium can be reconstructed as a direct IC-IR mechanism where both types propose the same mechanism but report their type differently. It then follows that an equilibrium payoff pair must also be the payoff pair of some IC-IR mechanism. To see that the high-type seller must get at least the free-trial revenue, consider the free-trial described earlier in Section \ref{sec:free_trial}. This free-trial mechanism can guarantee high type the payoff of $\pi_F$ regardless of the buyer's belief.
Thus, if some equilibrium had $\pi_H < \pi_F$, then the free-trial is a profitable deviation for the high-type seller, a contradiction. 
\end{proof}

\subsection{Omitted Proofs for Section 6}

\paragraph{Relaxed Design Problem}
To prove Theorem \ref{thm:tieredpricing}, we apply the exact same argument as in the proof of Theorem \ref{thm:trial} to consider the following relaxed problem.
\begin{align*}
    \max_{M \in \mathcal{M}} ~ & \left \{ w_L p^L_U(T) + p^H_U(T) + \mathbb{E} \left[\Delta^H_{v,\tau}(T)\mathbbm{1}[\tau \le T]\mid \theta = H \right] \right \} \notag \\
    \text{subject to} \quad & 
    \Delta_{v,t}^H (T) =  \lambda v q(v,t) - \int_{\ubar{v}}^v \lambda q(w,t) \dd w - \lambda \ubar{v} q(U,t) \qquad \forall (v, t)  \\
    &  q(v,t) \geq q(U,t) \\
    &  \mu_0\left(\lambda  q(U,0) + \mathbb{E}\left[ \lambda v (q(v,\tau) - q(U,\tau)) - \Delta^H_{v,\tau}(T) \mid \theta = H \right] \right) \\
     &\qquad \ge \mu_0 p^H_U(T) + (1-\mu_0) p^L_U(T)  \\
    & p^L_U(T) \ge \max(p^H_U(T), 0) 
\notag
\end{align*}
where $q(w,t)$ and $q(U,t)$ are defined as
\begin{align*}
q(w,t) &:= \int_t^T  q_{w,t}^H(s) I^H_{w,t}(s) \ \dd s, \quad q(U,t) := \int_t^T  q^H_U(s)I^H_{U}(s) \ \dd s.
\end{align*}
Note that we abuse the notation a little bit here by using $q(w,t)$ to represent the cumulative integration of $q I$.
Then, following the approach of Theorem \ref{thm:trial} we reduce the problem to:
\begin{align}
    \max_{M \in \mathcal{M}} ~ & \left \{ (w_L+1)\mu_0  \lambda q(U,0) + \lambda \pi_0 \mathbb{E}[(T - \tau - q(U,\tau)) \mathbbm{1}[\tau \le T] | \theta = H ] \right \}\notag 
\notag
\end{align}
Denote $\mu_L = (w_L+1)\mu_0 / \pi_0$. The above problem is equivalent to
\begin{align}
    \max_{M \in \mathcal{M}} ~ & \left \{ \mu_L  q(U,0) +  \mathbb{E}[(T - \tau - q(U,\tau)) \mathbbm{1}[\tau \le T] | \theta = H ] \right \}\notag 
\notag
\end{align}

Now, we use the fact that $\tau$ given $I_U$ is distributed according to density $\lambda I_U^H(t) e^{-\int_0^t \lambda I_U^H(s) \dd s}$.
Plugging this in, the problem is
\begin{gather*}
   \max_{I,q}  \int_0^T    \mu_L q_U^H(t) dt + \int_0^T  \left (\int_t^T 1 - q_U^H(s) \dd s \right)  \lambda  I_U^H(t) e^{-\int_0^t  \lambda  I_U^H(s) \dd s}  \dd t
\end{gather*}

To simplify the notation, we drop the superscript $H$ and simply use $I_U(t)$ and $q_U(t)$. 
We change variables to reformulate the problem as an optimal control. Define 
\begin{gather*}
    I(t) = \int_{0}^t  I_U(s) \dd s, \qquad Q(t) = \int_{0}^t  q_U(s) \dd s.
\end{gather*}
We rewrite the objective function in its equivalent form:
\begin{align*}
 &\int_0^T  \big( T-t + Q(t)-Q(T) \big)e^{-\lambda I(t)}   \lambda I_{U}(t)  \dd t   + \mu_L Q(t)  \\
     =& \int_0^T \big( T-t + Q(t) \big)e^{-\lambda I(t)}  \lambda I_{U}(t) \dd t    + (e^{-\lambda I(T)}-1+\mu_L) Q(T)
\end{align*}
So the control problem is
\begin{gather*}
    \max_{(I_U( \cdot ),q_U(\cdot )) \in \mathcal{U}_{ad}}  \int_0^T \big( T-t + Q(t) \big)e^{-\lambda I(t)}  \lambda I_{U}(t) \dd t    + (e^{-\lambda I(T)}-1+\mu_L) Q(T)\\
    \dot{I}(t) = I_{U}(t)  \qquad \dot{Q}(t) = q_{U}(t) \\
    Q(0) = I(0) = 0, \qquad  (I_U(t),q_U(t)) \in \mathcal{D}.
\end{gather*}
Define the contingent set $X(I,Q,t)$ as follows:
\begin{align*}
    X(I,Q,t) &= \{ (\zeta,\eta) \in \mathcal{R} \times \mathcal{R}^2  | ~  \zeta \geq (T-t+ Q)e^{-\lambda I} \lambda I_U, ~ \eta = (I_U,q_U)  \text{ for some } (I_U,q_U)\in \mathcal{D}  \}
\end{align*}
Note that the function
$(T-t+ Q)e^{-\lambda I} \lambda I_U$
is linear in $(I_U,q_U)$ and has bounded derivatives with respect to $(I,Q)$. As a correspondence, $X(I,Q,t)$ is closed convex valued, and it is upper and lower hemicontinuous in $(I,Q)$, so $X(I,Q,t)$ has the weak Cesari property. Theorem 4.2.4 in \cite{AhmedWang21} implies there exists an optimal control. 

Further, for any optimal control $(I_U^*(t),q_U^*(t))$ and the corresponding state functions $(I^*(t),Q^*(t))$, Theorem 5.2.3 in \cite{AhmedWang21} ensures there exists absolute continuous functions $\rho_I(t)$ and $\rho_Q(t)$ such that 
\begin{align*}
    \dot{\rho}_Q(t) &= -e^{-\lambda I^*(t)} \lambda I^*_{U}(t)\\
    \dot{\rho}_I(t) &=  \big( T-t + Q(t) \big)e^{-\lambda I^*(t)} \lambda^2 I^*_{U}(t) \\
    (I^*_{U}(t),q^*_{U}(t)) &\in \argmax_{(I,q)\in Conv(\mathcal{D})} ~  \Big( \big( T-t + Q^*(t) \big)\lambda e^{-\lambda I^*(t)} + \rho_I(t) \Big)  I  + \rho_Q(t) q, ~a.e.\\
     \rho_I(T)&= -\lambda Q^*(T) e^{-I^*(T)}, \\
     \rho_Q(T) &= e^{-\lambda I^*(T)}-1+\mu_L.
\end{align*}

Define $A_I(t) := \big( T-t + Q^*(t) \big)\lambda e^{-\lambda I^*(t)} + \rho_I(t) .$
For any optimal control $(I_U^*(t),q_U^*(t))$ and the corresponding  $(I^*(t),Q^*(t))$ and $\rho_I(t)$, $\rho_Q(t)$, we can always find an essentially the same control  $(I_U^{*'}(t),q_U^{*'}(t))$  such that the optimality condition holds for any $t$ instead of $a.e.$ t. Thus, without loss, we assume the optimality condition holds for any $t$.

Now we derive some properties of the optimal solution from this set of necessary conditions.
Note that $\frac{\partial A_I}{\partial t} (t) = (q^*_U(t) - 1) \lambda e^{-\lambda I(t)} \leq 0$,  $
       A_I(T) =  0$, and 
    $\rho_Q (t) = \mu_L - 1 + e^{-\lambda I(t)}. $ Therefore, $A_I(t)$ is weakly decreasing in $t$ and is weakly positive. $\rho_Q$ is decreasing and $\rho_Q(0) = \mu_L \geq 0$.  We discuss two cases.

\paragraph{Case 1} ($\mu_L \geq 1 - e^{-\lambda T}$): In this case we know $\rho_Q(t) \geq \rho_Q(T) \geq 0$. In fact, for any $t<T$ $\rho_Q(t) > 0$. Because $(1,1)\in \mathcal{D}$, for any $t<T$ we must have
\begin{gather*}
    (I^*_U(t),q^*_U(t))= 
    \begin{cases}
        (1,1) \quad &\text{if } A_I(t) > 0, \\
        (I,1) \in Conv(\mathcal{D}) \quad &\text{if } A_I(t) = 0.
    \end{cases}
\end{gather*}
Because $q^*_U(t)=1$ for any $t<T$, the objective value is simply $\mu_L T$.
\paragraph{Case 2}($\mu_L < 1 - e^{-\lambda T}$):
We already know $A_I(t)$ is weakly decreasing and weakly positive. Now we claim
$A_I(0)>0$ for any system associated with optimal control. 

Suppose not, then $A_I(t)=0$ for any $t$, so $q^*_U(t) =1$ for almost all $t$. This leads to a maximized objective value of $\mu_L T$. However, by simply using a trial mechanism the objective can be higher than $\mu_L T$ when $\mu_L < 1 - e^{-\lambda T}$. This is a contradiction. Thus, for any optimal control, we know the associated $A_I(t)$ must have $A_I(0) > 0$.

Denote $t_0 = \min \{ t\in[0,T] |  A_I(t) = 0  \} > 0$. We argue $t_0 = T$. Suppose $t_0<T$; then there exists $\delta>0$ such that $q^*_U(t) <1, A_I(t) > 0$ for a.e. $t\in (t_0-\delta,t_0)$; and $q^*_U(t) =1, A_I(t) = 0$ for a.e. $t\in (t_0,T]$. From  $q^*_U(t) =1, A_I(t) = 0$ for a.e. $t\in (t_0,T]$ we know $\rho_Q(t) \geq 0$ for a.e. $t\in (t_0,T]$. From monotonicity of $\rho_Q(t)$ we know $\rho_Q(t) \geq 0$ for any $t\in (t_0-\delta,T]$. On the other hand, from $q^*_U(t) <1, A_I(t) > 0$ for a.e. $t\in (t_0-\delta,t_0)$ we know $\rho_Q(t) \leq 0$ for a.e. $t\in (t_0-\delta,t_0]$ otherwise $q^*_U(t)$ would be 1. These two inequalities imply $\rho_Q(t) = 0$ for any $t\in (t_0-\delta,t_0]$, which further implies $I^*_U(t) =1$ for a.e. $t\in (t_0-\delta,t_0]$. However, this means $\rho_Q(t)$ is strictly decreasing in  $t\in (t_0-\delta,t_0]$, a contradiction.

We have proved $A_I(t)>0$ for any $t<T$. Now denote $t_1 = \max\{ t\in[0,T]| ~q^*_U(s) = 1 ,~a.e.~ s \in [0,t] \} < T$. First, we first prove $\rho_Q(t)<0$ for any $t>t_1$. Suppose not; then combining this with the definition of $t_1$ we know there exists $\delta>0$ such that 
$q^*_U(t) < 1, ~\rho_Q(t) \geq 0 ,~a.e.~ t \in [t_1,t_1+\delta]$. However, this implies $I^*_U(t) = 1~a.e.~ t \in [t_1,t_1+\delta]$ and so $\rho_Q(t)$ is strictly decreasing in $[t_1,t_1+\delta]$, a contradiction. Second, we prove $q^*_U(t)<1$ $ a.e. ~\forall t>t_1$. Suppose not; then there exists $t_2$ and $\delta$ such that $q^*_U(t)=1$ $ a.e. ~ t \in [t_2-\delta,t_2]$. Because $A_I(t)>0$ we know $I^*_U(t)=1$ $ a.e. ~ t \in [t_2-\delta,t_2]$.
Thus we know for a $t_3 \in [t_1,t_1+\delta]$, the following holds:
\begin{gather*}
    (q^*_U(t_3)-1) \rho_Q(t_3) +  (I^*_U(t_3)-1) A_I(t_3)  \geq 0 \\
    (q^*_U(t_2)-1)  \rho_Q(t_2) +  (I^*_U(t_2)-1) A_I(t_2)  \leq 0
\end{gather*}
However, both $\rho_Q(t)$ and $A_I(U)$ are decreasing and $q^*_U$ is strictly decreasing over $[t_2-\delta,t_2]$. This leads to a contradiction.

In summary, we have proved the following facts. 
\begin{itemize}\itemsep0em
    \item When $t\leq t_1$, $I^*_U(t) = q^*_U(t)= 1$ a.e..
    \item When $t > t_1$, $q^*_U(t) <1$ a.e. so $A_I(t)$ is strictly decreasing in $t$.
    \item When $t > t_1$, $A_I(t) >0$ and $\rho_Q(t)<0$.
\end{itemize}

Note that the indifference curve of the linear optimization problem
\begin{gather*}
    \max_{(I,q)\in D} ~  \Big( \big( T-t + Q^*(t) \big)\lambda e^{-\lambda I^*(t)} + \rho_I(t) \Big)  I  + \rho_Q(t) q
\end{gather*}
is $A_I(t)/\rho_Q(t).$  Thus, when $t>t_1$, we know that $A_I(t)/\rho_Q(t)$ is negative and continuous,the absolute value $|A_I(t)/\rho_Q(t)|$ is strictly decreasing, and $A_I(T)/\rho_Q(T)=0$. 
Thus, we know as $t$ ranges from $t_1$ to $T$, the optimal control $(I_U^*(t),q_U^*(t))$ moves along the boundary:
\begin{gather*}
    \partial^- \conv(\tilde{\mathcal{D}}) = \{  (I,q ) \in \conv(\tilde{\mathcal{D}})| I \geq I',~ \forall (I',q)\in \conv(\tilde{\mathcal{D}})     \}.
\end{gather*}
Because the slope of the indifference curve  $|A_I(t)/\rho_Q(t)|$  is strictly decreasing, for almost all $t$, the control is at the extrme point:
\begin{gather*}
    (I^*_U(t),q^*_U(t)) \in \partial^- \conv(\tilde{\mathcal{D}}) \cap \mathcal{D} ~ a.e. ~ t\in [0,T].
\end{gather*}
Thus, any optimal control of the relaxed control problem (with the control set $ \conv(\tilde{\mathcal{D}})$) is also feasible with the control set $\mathcal{D}$.

\paragraph{Incentive Compatibility}
Now, we verify that the optimal mechanisms of the relaxed design problem are indeed incentive-compatible. First, we argue a buyer who receives $v$ at time $t$ finds it optimal to be truthful. Because $v_0$ is the threshold type and is indifferent between being truthful and staying silent forever, we have the following binding constraint for any $t$:
\begin{align*}
    \lambda v_0 (T-t) - \Delta_{v_0,t}^{*H} (T) = \lambda v_0 q^*(T,t)
\end{align*}
This implies that for any $t'>t$,
\begin{align}
    &\lambda v_0 (t'-t-(q^*(U,t)-q^*(U,t'))) - \Delta_{v_0,t}^{*H}(T) + \Delta_{v_0,t'}^{*H}(T) = 0, \notag \\
    \implies &\lambda v (t'-t-(q^*(U,t)-q^*(U,t'))) - \Delta_{v_0,t}^{*H}(T) + \Delta_{v_0,t'}^{*H}(T) \geq 0, \quad \forall v\geq v_0.
    \label{eq_tired_inequa1}
\end{align}

Now consider the potential deviation where the buyer chooses to report $v'$ at some future time $t'\geq t$. If $v'< v_0$, this deviation is equivalent to stay silent forever is not profitable as required in the relaxed problem. Thus, we just consider the deviation where $v'\geq v_0$. The payoff will be
\begin{align*}
&\lambda v (q^*(U,t) - q^*(U,t')) +  \lambda v (T-t') - \Delta^H_{v_0,t'}(T) - p_U^*(T) 
\end{align*}
For a buyer with $v \geq v_0$, 
the payoff difference between this misreport and truth-telling is
\begin{align*}
   &\lambda v (q^*(U,t) - q^*(U,t')) +  \lambda v (T-t') - \Delta^{*H}_{v_0,t'}(T) - p_U^*(T)  -\lambda v(T-t) + \Delta^{*H}_{v_0,t}(T) + p_U^*(T)\\
= & \lambda v (q^*(U,t) - q^*(U,t')) +  \lambda v (t-t') - \Delta^{*H}_{v_0,t'}(T)  + \Delta^{*H}_{v_0,t}(T)  \leq 0,
\end{align*}
where the last inequality comes from (\ref{eq_tired_inequa1}). For a buyer with $v< v_0$, the payoff difference between this misreport and truth-telling is 
\begin{align*}
&\lambda v (q^*(U,t) - q^*(U,t')) + \lambda v q^*(v',t') - \Delta^{*H}_{v',t'}(T) - p^{*H}_U(T) - \lambda v q^*(U,t)  + p^{*}_U(T)  \\
= &  \lambda v q^*(v',t') - \lambda v q^*(U,t') - \Delta^{*H}_{v',t'}(T) \leq 0,
\end{align*}
where the last inequality comes from the the requirement of the relaxed design problem.

To conclude, we have verified that for a buyer who receives the first reward of size $v$ at time $t$, there is no profitable deviation. The uninformed buyer at $t$ cannot profitably deviate by reporting $v\geq v_0$ (to upgrade to premium service), since their perceived expected value of the service is at most $\mu_0\mathbb{E}[v] \le \ubar{v} \le v_0$. Thus, since the solution to the relaxed problem is incentive-compatible with the original constraints, it is optimal. 
Note that deviating to report $v< v_0$ is weakly worse than being truthful because staying silent induces the same service provided by the seller, but does not forfeit the option to upgrade. 

At time $t$, the uninformed buyer has a posterior $\mu_t$ of $\theta=H$, where 
$$\mu_t = \frac{\mu_0 e^{-\lambda I^*(t)}}{\mu_0 e^{-\lambda I^*(t)} + 1-\mu_0}.$$
If $\mu_t \leq \ubar{v}$, then $\mu_t \leq v_0$ and deviating to report $v\geq v_0$ is worse than staying silent forever. This is because type $v_0$ is indifferent between reporting $v\geq v_0$ and staying silent forever.
\begin{align*}
    \lambda v_0 (T-t) - \Delta_{v_0,t}^{*H} (T) = \lambda v_0 q^*(T,t) \implies      \lambda  (T-t) - \Delta_{v_0,t}^{*H} (T) \leq  \lambda \mu_t q^*(T,t)
\end{align*}
Staying silent forever is clearly weakly worse than being truthful because of the future option values. Thus, it is not profitable to deviate to report $v\geq v_0$ at time $t$ where $\mu_t \leq \ubar{v}$. Note that this argument completes the proof when $\mu_0 \leq \ubar{v}$.

To discuss cases where $\mu_t > \ubar{v}$,
we proceed to establish a lower bound on the buyer's payoff when being truthful. To do so,
we first establish a lower bound on the total amount of learning $I^*(T)$ in the optimal solution of the relaxed problem. Denote $(\ubar{q},\ubar{I}) = \min \{ (q,I) | (q,I) \in \mathcal{D}_*, ~ I>0   \}$. $(\ubar{q},\ubar{I})$ has the worst service value and the slowest learning rate among all possible (non-zero) service $(I,q)$ in the extreme point set $\mathcal{D}_*$. Our assumption is that $\ubar{I}>0$, which is true, for example, if $\mathcal{D}_*$ is finite.  We discuss two cases. First, if $\ubar{q}=0$, then for any $t$, $I^*_U(t) \geq \ubar{I}$, so we obtain a simple lower bound $I^*(T) \geq \ubar{I} T$. Second, if $\ubar{q}>0$, then $I^*_U(T)=0$. We denote  $\bar{t} = \min \{  t | I^*(t)=I^*(T) \}$ and we know $\bar{t} \leq I^*(T)/\ubar{I}$. According to the Pontryagin equations, 
\begin{align*}
    \rho_Q^*(\bar{t}) &= \mu_L -1 + e^{-\lambda I^*(T)}\\
    A_I^*(\bar{t}) & = \lambda e^{-\lambda I^*(T)} (T-\bar{t}).
\end{align*}
Because at $\bar{t}$, $I^*_U(t)$ switches from $(\ubar{q},\ubar{I})$ to $(0,0)$, we have
\begin{align*}
    0 &= \ubar{q}( \mu_L -1 + e^{-\lambda I^*(T)}) + \ubar{I}  \lambda e^{-\lambda I^*(T)} (T-\bar{t}) \\
    0 &\geq \ubar{q}( \mu_L -1 + e^{-\lambda I^*(T)}) + \ubar{I}  \lambda e^{-\lambda I^*(T)} (T-I^*(T)/\ubar{I})
\end{align*}
Denote $x= e^{-\lambda I^*(T)}$ so $I^*(T) = -\ln x/\lambda$. Because $1-1/x\leq \ln x$, we have
\begin{gather*}
    0 \geq   \ubar{q}( \mu_L -1 + x) + \ubar{I}  \lambda x T + x \ln x \geq \ubar{q}( \mu_L -1 + x) + \ubar{I}  \lambda x T + x -1 \\
    \implies \frac{1}{x} \geq 1 + \frac{\lambda \ubar{I}}{1+ \ubar{q}}T  \quad \implies  I^*(T) \geq \frac{1}{\lambda} \ln ( 1 + \frac{\lambda \ubar{I}}{1+ \ubar{q}}T )
\end{gather*}
This is a lower bound on $I^*(T)$. Now recall that the uninformed buyer's posterior of $\theta=H$ at time $\bar{t}$ is $\mu_{\bar{t}}$, and 
$$\mu_{\bar{t}} = \frac{\mu_0 e^{-\lambda I^*(\bar{t})}}{\mu_0 e^{-\lambda I^*(\bar{t})} + 1-\mu_0} = \frac{\mu_0 }{\mu_0  + \frac{1-\mu_0}{x}}.$$
Because we have either the above lower bound or $I^*(T) \geq \ubar{I} T$, we can take a large $T_0$ such that for any $T\geq T_0$, $\mu_{\bar{t}}\leq \ubar{v}$. In particular, we can take
\begin{gather*}
    T_0 =   \frac{1+\ubar{q}}{\lambda \ubar{I}}  \frac{\mu_0 -\ubar{v}}{(1-\mu_0) \ubar{v}} \geq  \frac{1}{\lambda \ubar{I}} \ln \frac{\mu_0(1-\ubar{v})}{(1-\mu_0)\ubar{v}}  
\end{gather*}

For any $T\geq T_0$, we can find a time $t^* \leq \bar{t}$ such that $\mu_{t^*}=\ubar{v}$. One lower bound on the uninformed buyer's payoff when being truthful is the buyer's payoff under the strategy, where he keeps silence until time $t^*$ and then reports $v_0$ at $t^*$ (upgrading to premium service at time $t^*$) if and only if he receives a reward before and $v\geq v_0$. We know this strategy is weakly worse than being truthful because all deviations in this strategy either happens when the buyer receives a reward before time $t^*$, or when the buyer hasn't received a reward after $t^*$. Note that $\mu_{t^*}= \ubar{v}$  so we have proven previously that in both cases being truthful is the optimal strategy.

For an uninformed buyer at time $t$,
with  probability $ p_t^{t^*}= 1 -  e^{-\lambda (I^*(t^*) - I^*(t)) }$, the uninformed buyer will  receive a reward before time $t^*$ conditioned on $\theta=H$. The buyer's payoff under this particular deviation is
\begin{align*}
    \lambda \mu_t (q(U,t)-q(U,t^*)) +  \mu_t p_t^{t^*} \lambda  ( (T-t^*   )\int_{v_0}^{\bar{v}} (v-v_0) f(v) \dd v + q(U,t^*)  )+ \lambda ( 1-  \mu_t p_t^{t^*}) \mu_{t^*}   q(U,t^*)  -p_U^*(T)
\end{align*}
In contrast, if the uninformed buyer deviates to $v\geq v_0$, the payoff is
\begin{gather*}
    \lambda (\mu_t-v_0)(T-t) + \lambda \mu_t q(U,t)-p_U^*(T)
\end{gather*}
Using Bayesian consistency: $ \mu_t p_t^{t^*}+( 1-  \mu_t p_t^{t^*}) \mu_{t^*} = \mu_t $,  the difference in payoff is
\begin{align*}
      \lambda  \mu_t p_t^{t^*}  (\int_{v_0}^{\bar{v}}   (v-v_0) \dd v) (T-t^*     ) - \lambda (\mu_t-v_0) (t^*-t)
\end{align*}
Denote $s_{v_0} =  \int_{v_0}^{\bar{v}} (v-v_0) f(v) \dd v $. 
Denote $y^*=I^*(t^*) = \frac{1}{\lambda} \ln \frac{\mu_0(1-\ubar{v})}{\ubar{v}(1-\mu_0)},$ which only depends on primitives of the model, and $y = I^*(t)$. The payoff difference is 
\begin{align*}
    \lambda \mu_t s_{v_0} (T-t^*) (1-e^{-\lambda(y^*-y)}) -  \lambda (\mu_t-v_0) (t^*-t)
\end{align*}
If $\mu_t \leq v_0$, then this difference is positive, which implies being truthful is optimal. Otherwise,  $\mu_0 \geq \mu_t > v_0$. Note that $t^*<\bar{t}$ so $t^*-t\leq (y^*-y)/\ubar{I}$ and $t^*\leq y^*/\ubar{I}$. The payoff difference must be larger than
\begin{align*}
    &\lambda \mu_t s_{v_0} (T-\frac{y^*}{\ubar{I}}) (1-e^{-\lambda(y^*-y)}) - \frac{\lambda}{\ubar{I}} (\mu_t-v_0) (y^*-y)  \\
    =& \frac{\lambda}{\ubar{I}} (\mu_t-v_0) (y^*-y)  \big[    \frac{\mu_t}{\mu_t-v_0}\frac{ s_{v_0} (\ubar{I}T- y^*) (1-e^{-\lambda(y^*-y)})}{ (y^*-y)}   - 1  \big] \\
   \geq & \frac{\lambda}{\ubar{I}} (\mu_t-v_0) (y^*-y)  \big[     \frac{\mu_0}{\mu_0-\ubar{v}} s_{\mu_0} (\ubar{I}T- y^*) \frac{1-e^{-\lambda y^*}}{y^*}  - 1  \big]
\end{align*}
As long as the expression is positive, the payoff difference is positive. Consequently, being truthful is optimal.

Recall that  $y^* = \frac{1}{\lambda} \ln \frac{\mu_0(1-\ubar{v})}{\ubar{v}(1-\mu_0)}$. To make sure the expression is positive, we need
\begin{align*}
 &T \geq \frac{y^*}{\ubar{I}} (1 + \frac{\mu_0 - \ubar{v}}{\mu_0 s_{\mu_0} (1-e^{-\lambda y^*})}) \\
\Leftrightarrow &   T \geq \frac{y^*}{\ubar{I}} (1 + \frac{\mu_0 - \ubar{v}}{\mu_0 s_{\mu_0} } \frac{\mu_0(1-\ubar{v})}{\mu_0 -\ubar{v}} ) = \frac{y^*}{\ubar{I}} (1 + \frac{1-\ubar{v}} 
 {s_{\mu_0}} ) = y^* \frac{1+ s_{\mu_0} - \ubar{v}}{\ubar{I} s_{\mu_0}} \\
 \Leftarrow  & T \geq \frac{1+ 2 s_{\mu_0} - \ubar{v}}{\lambda \ubar{I} s_{\mu_0}}  \frac{\mu_0 -\ubar{v}}{(1-\mu_0)\ubar{v}}
\end{align*}

To conclude, we have shown that being truthful is optimal for uninformed buyer, as long as $ T \geq T_1$, where $T_1$ depends on primitives:
\begin{align*}
   T_1 &=  \frac{1+ 2 s_{\mu_0} - \ubar{v}}{\lambda \ubar{I} s_{\mu_0}}  \frac{\mu_0 -\ubar{v}}{(1-\mu_0)\ubar{v}} \geq \frac{1+\ubar{q}}{\lambda \ubar{I}}  \frac{\mu_0 -\ubar{v}}{(1-\mu_0) \ubar{v}}
\end{align*}
This completes the proof of incentive-compatability.

\subsection{Omitted Proofs for Section 7}

\begin{proof}[Proof of Proposition \ref{prop:extensionsTRM}]
    We consider a similar relaxed problem as follows
\begin{gather*}
     \max_{I, p, \Delta  } ~ (1+w_L) (p_U(T)-cI(U,0)) + w_L \mathbb{E} \left[\Delta_{v,\tau} + cI(U,\tau) - cI(v,\tau) \mid S = L \right] \\ +  \mathbb{E} \left[\Delta_{v,\tau} + cI(U,\tau) - cI(v,\tau) \mid S = H \right]  \\
    \textnormal{s.t. } \quad \text{(IC-V) + (IC-U)} \\
          p_U(T)   \le \mu_L (u+\lambda) I(U,0) + \mu_L \mathbb{E}\left[ (u+\lambda v) \left( I\left(v,\tau \right) -I(U,\tau) \right)  \left.  - \Delta_{v,\tau} \right\vert \theta = H\right].
\end{gather*}
Note that the objective function can be rewritten as
\begin{gather*}
   (1+w_L) (p_U(T)-cI(U,0)) + (w_L \mu_L + \mu_G) \mathbb{E} \left[\Delta_{v,\tau} + cI(U,\tau) - cI(v,\tau) \mid \theta = H \right] 
\end{gather*}
It is without loss of generality to assume the ex-ante IR constraint binds. Plugging it in, the objective becomes:
\begin{align*}
       &(1+w_L) \big(   \mu_0 (u+\lambda) I(U,0) + \mu_0 \mathbb{E}\left[ (u+ \lambda v) \left( I\left(v,\tau \right) -I(U,\tau) \right)  \left.  - \Delta_{v,\tau} \right\vert \theta = H\right] -cI(U,0) \big ) \\
       & ~ + (w_L \mu_L + \mu_G ) \mathbb{E} \left[\Delta_{v,\tau} + cI(U,\tau) - cI(v,\tau) \mid \theta = H \right]   \\
       =& (1+w_L) \big(\mu_0 (u+ \lambda) I(U,0) + \mu_0 \mathbb{E}\left[ (u+\lambda v) \left( I\left(v,\tau \right) -I(U,\tau) \right)  \left.  \right\vert \theta = H\right] -cI(U,0) \big) \\
         & ~ +  \mathbb{E} \left[  ( \mu_G-\mu_0 - w_L(\mu_0 - \mu_L) ) \Delta_{v,\tau} + (w_L \mu_L  + \mu_G) (I(U,\tau) - I(v,\tau))c | \theta = H \right]
\end{align*}
Because we focus on $w_L \leq p_B/p_G$ we have
\begin{gather*}
   \mu_G-\mu_0 - w_L(\mu_0 - \mu_L) \geq 0
\end{gather*}
Thus, higher $\Delta_{v,t}$ is always preferred. Next, from the envelop representation of (IC-V) we know for any $v_1$
\begin{align*}
         \Delta_{v,t}  &=   (u+\lambda v) I \left(v,t \right)  - \lambda \int_{v_1}^v I(w,t) \dd w -   u(v_1,t)  \quad \forall t,v  
\end{align*}
Plug it in we rewrite the objective as
\begin{align*}
       & (1+w_L) \big(\mu_0 (u+ \lambda) I(U,0) + \mu_0 \mathbb{E}\left[ (u+\lambda v) \left( I\left(v,\tau \right) -I(U,\tau) \right)  \mid \theta = H\right] -cI(U,0) \big) \\
             +&  \mathbb{E} \left[   (w_L \mu_L  + \mu_G) (I(U,\tau) - I(v,\tau))c \mid \theta = H \right] +   ( \mu_G-\mu_0 - w_L(\mu_0 - \mu_L) ) \\
         &  \mathbb{E} \left[  (u+\lambda v) I \left(v,\tau \right)  - \lambda \int_{v_1}^v I(w,\tau ) \dd w -  u(v_1,\tau)  \mid \theta = H \right] 
\end{align*}
Rearranging terms we get
\begin{align*}
    & \mathbb{E}  \left[  (      \mu_G + w_L \mu_L  ) (u + \lambda v-  c)  I(v,\tau) 
    \mid \theta =1   \right]  +  (1+w_L) (\mu_0 (u+ \lambda)-c) I(U,0)\\
   - &  \mathbb{E}  \left[ \big(  (w_L\mu_L + \mu_G) c - \mu_0(1+ w_L) (u+\lambda v)  \big) I(U,\tau) 
    \mid \theta =1   \right] \\
   - &    ( \mu_G-\mu_0 - w_L(\mu_0 - \mu_L) ) \mathbb{E} \left[ u(v_1,\tau) +  \lambda \int_{v_1}^v I(w,\tau ) \dd w   \mid \theta = H \right] 
\end{align*}
Denote $v_c$ such that $u+ \lambda v_c =c$. By the same argument as Lemma \ref{Lemma_low_bind},
\begin{align*}
         I(v,t)&\geq I(U,t), \quad \forall v \geq v_c,\\
         I(v,t) & = 0, \qquad \quad ~ \forall v < v_c.
\end{align*}
Note that here the inequality holds only for $v\geq v_c$ because the coefficient $(u+ \lambda v -c)$ of $I(v,\tau)$ can be negative: for lower $v$ even without the concern of rent reduction it is suboptimal to provide service. 
Now we fix the design of $I(U,t)$ and do point-wise maximization with respect to $I(v,t)$. To do so, denote
\begin{gather*}
    b_1 =   \mu_G-\mu_0 - w_L(\mu_0 - \mu_L), \quad b_2  =  \mu_G + w_L \mu_L  
\end{gather*}
Now note that
\begin{align*}
    &\mathbb{E}  \left[  (u + \lambda v-  c)  I(v,\tau) 
    -\frac{\lambda b_1}{b_2}  \int_{\ubar{v}}^v I(w,\tau ) \dd w 
    \mid \theta =1   \right] \\
    =& \mathbb{E}  \left[  (u + \lambda v-  c)  I(v,\tau) 
    -\frac{\lambda b_1}{b_2}   \frac{1-F(v)}{f(v)} I(v,\tau)
    \mid \theta =1   \right]
\end{align*}
Define the virtual value as
\begin{gather*}
    \phi(v) = u + \lambda v-  c - \frac{\lambda b_1}{b_2} \frac{1-F(v)}{f(v)}
\end{gather*}
Denote $v_0$ such that $\phi(v_0)=0$. Point-wise maximization implies that
\begin{gather*}
    I(v,t) = 
    \begin{cases}
        T-t   \quad &\text{if } v \in [v_0,\bar{v}] \\
        I(U,t) \quad &\text{if } v \in [v_c,v_0)>0, \\
        0    \quad &\text{if } v \in [\ubar{v},v_c)<0.
    \end{cases}
\end{gather*}
Denote the virtual surplus $\pi_v$ and $\pi_w$ as
\begin{align*}
    \pi_v &= \int_{v_0}^{\ubar{v}}  \phi(v) f(v) \dd v,\\
     \pi_w &= \int_{\ubar{v}}^{v_c} (u+ \lambda v - c) f(v) \dd v.
\end{align*}
Note that
\begin{align*}   
\pi_v + \int_{\ubar{v}}^{v_0} \phi (v) f(v) \dd v &= \int_{\ubar{v}}^{\bar{v}} (u+\lambda v -c) f(v) -  \frac{\lambda b_1}{b_2}  (1 - F(v)) \dd v \\
&= u+\lambda -c  - \frac{\lambda b_1}{b_2}  \left(\int_{\ubar{v}}^{\bar{v}}(1 - F(v)) \dd v \right) \\
&= u+\lambda -c  - \frac{\lambda b_1}{b_2}  \left(1 - \ubar{v} \right)
\end{align*}
Plugging in the point-wise maximization, the objective function becomes
\begin{align*}
    & \mathbb{E}  \left[  b_2 \pi_v (T-\tau) 
    +  b_2 \Big( -\pi_w  + ( u+\lambda -c  - \frac{\lambda b_1}{b_2}  \left(1 - \ubar{v} \right) -\pi_v ) -  (u -  c) \Big)  I(U,\tau) 
    \mid \theta =1   \right] \\
    & - \mathbb{E}  \left[ \big(  \lambda  b_1 \ubar{v} + \lambda \mu_0 (1+w_L)   \big) I(U,\tau) 
    \mid \theta =1   \right] +  (1+w_L) (\mu_0 (u+ \lambda)-c) I(U,0)\\
    =& \mathbb{E}  \left[  b_2 \pi_v (T-\tau) 
    -  b_2  (\pi_v+\pi_w)   I(U,\tau) 
    \mid \theta =1   \right]  +  (1+w_L) (\mu_0 (u+ \lambda)-c) I(U,0)
\end{align*}
Denote 
\begin{align*}
    b_3 &= \frac{ (1+w_L) (\mu_0 (u+ \lambda)-c) }{b_2 (\pi_v + \pi_w)} = \frac{ (1+w_L) (\mu_0 (u+ \lambda)-c) }{  (\pi_v + \pi_w) (\mu_G + w_L \mu_L)} \\
    b_4 &= \frac{ \pi_v} {\pi_v +\pi_w}.
\end{align*}

The remaining work is essentially adapting the proof of Lemma \ref{lem:ctrl}. We perform the main calculation; the remaining details follow by similar arguments to Lemma \ref{lem:ctrl}. When $b_4<0$, then the optimal solution is clearly to set $I_U(t)=1$ always (a degenerate trial length). Otherwise, denote $U(t) = I_U(t)$ and $X(t) = \int_0^t I_U(s) \dd s$. We have the following optimal control problem:
\begin{align*}
    \max_{U \in \mathcal{U}_{ad}} ~ &\left\{ \int_0^T (b_4 \left(T - t\right) - (X(T) - X(t))) e^{- \lambda X(t) }   \lambda U(t)  \dd t + b_3 X(T)  \right \} \\
    \textnormal{subject to } ~ & \dot{X}(t) = U(t) \in [0,1], ~ X(0) = 0, ~ X(T) \textnormal{ free}
\end{align*}
First, rearrange the objective:
\begin{align*}
    &\int_0^T (b_4 \left(T - t\right) - (X(T) - X(t))) e^{- \lambda X(t) }   \lambda U(t)  \dd t + b_3 X(T)  \\
    =~&\int_0^T (b_4 \left(T - t\right)  + X(t)) e^{- \lambda X(t) }   \lambda U(t)  \dd t -  X(T) \left( 1 -  e^{- \lambda X(T)} \right)+ b_3 X(T)
\end{align*}  
For
any optimal measurable control $U,X$, there exists a costate variable $\rho$ such that 
\begin{align*}
    U(t) &\begin{cases}
        = 1 & J(t) > 0 \\
        \in [0,1] & J(t) = 0 \\
        = 0 & J(t) < 0
    \end{cases} \\
    J(t) &=   \lambda (b_4 \left(T - t\right)  + X(t)) e^{- \lambda X(t) }  - \rho(t) \\
    \dot{\rho}(t) &= -\left( \lambda  \left( b_4(T - t)+X(t) \right) - 1 \right) \lambda e^{- \lambda X(t)}  U(t)\\
    \rho(T) &= 1-b_3 + (\lambda X(T)-1) e^{-\lambda X(T)}
\end{align*}
Taking derivatives we get 
\begin{align*}
    \dot{J}(t)    &= -\lambda b_4 e^{-\lambda X(t)}
\end{align*}
since \eqref{eqn:costate_evol} cancels the other terms. Note that this implies that $J$ is always strictly decreasing; since $X$ is continuous, $\dot{J}$ is continuous. Therefore, the optimal control $U(t)$ features bang-bang control. Namely, there exists a time $t_0$ such that $U(t) =1$ for $t< t_0$ and $U(t) =0$ for $t>t_0$. Given this optimal mechanism for the relaxed problem, it is straightforward to verify the omitted IC constraints are satisfied, so it is indeed the optimal IC-IR mechanism.

\end{proof}

\begin{proof}[Proof of Proposition \ref{prop:badnews}]

For any IC-IR mechanism $(I,p)$ denote $\mu_t$ as the buyer's of $\theta=H$ at time $t$ if no bad news ever arrived. By definition
\begin{align*}
    \mu_t = \frac{\mu_0}{\mu_0 + (1-\mu_0)e^{-\lambda (I(U,T)-I(U,t))}}, \qquad \text{where }
     I(U,t)= \int_t^T  I_U(s) \dd s.
\end{align*}
Denote $I(t) = \int_t^T I_t(s) \dd s$.
The relaxed design problem with linear Pareto weight is
\begin{gather*}
     \max_{I, p_U(T), \{ \Delta_t \} } ~ (1+w_L) p_U(T)  + w_L \mathbb{E} [ \Delta_\tau |\theta=0 ] \\
    \textnormal{s.t. }    \\
            (u-\lambda l) I(t) - \Delta_t \geq 0~ \forall t, \\
       u I(U,t) + (1-\mu_t) \mathbb{E} [ -\lambda l I(U,t) + (u-\lambda l) (I(\tau)-I(U,\tau) ) - \Delta_{\tau}  \mid \theta=0,\tau>t ] \\ 
       \geq    (u - \lambda(1-\mu_t)l ) I(t) - \Delta_t,  ~\forall t, \\
         p_U(T) \le  u I(U,0) + (1-\mu_0) \mathbb{E} [ -\lambda l I(U,0) + (u-\lambda l)(I(\tau)-I(U,\tau) )- \Delta_{\tau}  \mid \theta=0 ].
\end{gather*}
The first constraint requires the agent who receives the signal at time $t$ does not want to quit. The second constraint requires that the uninformed buyer does not want to pretend to receive a signal at time $t$. The third constraint is the ex-ante IR constraint.

Because $w_L \geq -1$, the ex-ante IR constraint should bind so
\begin{gather*}
     \max_{I, \{ \Delta_t \} } ~  u(1+w_L) I(U,0) + (1-\mu_0)(1+w_L) \mathbb{E} [ -\lambda l I(U,0) + (u-\lambda l)(I(\tau)-I(U,\tau) ) \mid \theta=0 ]  \\
     - (1-\mu_0(1+w_L)) \mathbb{E} [ \Delta_{\tau}  \mid \theta=0 ] \\
    \textnormal{s.t. }    \\
     (u-\lambda l) I(t) - \Delta_t \geq 0~ \forall t, \\
       u I(U,t) + (1-\mu_t) \mathbb{E} [ -\lambda l I(U,t) + (u-\lambda l) (I(\tau)-I(U,\tau) ) - \Delta_{\tau}  \mid \theta=0,\tau>t ] \\ 
       \geq    (u - \lambda(1-\mu_t)l ) I(t) - \Delta_t,  ~\forall t.
\end{gather*}
Because $ \lambda l > u > (1-\mu_0) \lambda l$, and  $\mu_t \geq \mu_0$ for any $t$, it is clearly optimal to set $I(t)=0$ for any $t$, as it increases the objective and relaxes the constraint. Now the design problem can be simplified as 
\begin{gather*}
     \max_{I, \{ \Delta_t \} } ~  u(1+w_L) I(U,0) + (1-\mu_0)(1+w_L) \mathbb{E} [ -\lambda l I(U,0) - (u-\lambda l)I(U,\tau)  \mid \theta=0 ]  \\
     - (1-\mu_0(1+w_L)) \mathbb{E} [ \Delta_{\tau}  \mid \theta=0 ] \\
    \textnormal{s.t. }   \quad
      \Delta_t \leq 0 ~ \forall t,\\
       u I(U,t) + (1-\mu_t) \mathbb{E} [ -\lambda l I(U,t) - (u-\lambda l) I(U,\tau)  - \Delta_{\tau}  \mid \theta=0,\tau>t ]
       \geq    - \Delta_t,  ~\forall t.
\end{gather*}

We discuss two cases. First, if $w_L \geq (1-\mu_0)/\mu_0$,  the coefficient of $\Delta_{\tau}$ is negative. Clearly, it is optimal to set $\Delta_t = 0$ and set $I(U,t)=T-t$. The optimal mechanism is just to sell everything ex-ante.

 Second, if $w_L < (1-\mu_0)/\mu_0$,the coefficient of $\Delta_{\tau}$ is positive. Note that from the two constraint, we get $\Delta_T=0$. We keep $\Delta_T=0$ and neglect all constraints such that $\Delta_t\leq 0$. For any fixed measurable control $I_U(t)$, the fixed function $I(U,t)$ is absolute continuous in $t$. Note $\Delta_T=0$ and denote 
\begin{gather*}
    \Gamma(-\Delta_t) =  u I(U,t) + (1-\mu_t) \mathbb{E} [ -\lambda l I(U,t) - (u-\lambda l) I(U,\tau)  - \Delta_{\tau}  \mid \theta=0,\tau>t ].
\end{gather*}
$\Gamma$ is clearly increasing in $-\Delta_t$,
so according to Grönwall's inequality, the point-wise minimum feasible $\Delta_t$ is the one such that equality always holds:
\begin{gather*}
       \mathbb{E} \left[ \int_t^{\tau} (u-\lambda l(1-\theta) ) I_U(t) \dd s - \Delta_{\tau}  \mid \tau>t \right]
       =    - \Delta_t,  ~\forall t.
\end{gather*}
The above integration equation has a unique solution. One can easily verify that the solution is 
\begin{gather*}
    -\Delta_t = \int_t^T (u-\lambda l(1-\mu_s)) I_U(s) \dd s. 
\end{gather*}
To see this, plugging in the candidate solution to the equation we get
\begin{align*}
     -\Delta_t - \mathbb{E} [ - \Delta_{\tau}  \mid \tau>t ]
       &=     \mathbb{E} [ \int_t^{\tau} (u-\lambda l(1-\mu_s) ) I_U(t) \dd s \mid \tau>t ] \\
       &=   \mathbb{E} [ \int_t^{T} 1_{s\leq \tau } (u-\lambda l(1-\mu_s) ) I_U(t) \dd s \mid \tau>t ]\\
       &=   \mathbb{E} [ \int_t^{T} 1_{s\leq \tau } (u-\lambda l(1-\theta) ) I_U(t) \dd s \mid \tau>t ] \\
       & =       \mathbb{E} [ \int_t^{\tau} (u-\lambda l(1-\theta) ) I_U(t) \dd s \mid \tau>t ]
\end{align*}
Thus, we prove the candidate solution is indeed the unique solution. Plugging this into the objective, we have an unconstrained maximization problem
\begin{gather*}
     \max_{I } ~   (1+w_L) \mathbb{E} [\int_0^{\tau} (u-\lambda l(1-\mu_s))I_U(s) \dd s ]  +  (1-\mu_0(1+w_L))  \mathbb{E} [ \int_{\tau}^T (u-\lambda l(1-\mu_s))I_U(s) \dd s  ].
\end{gather*}
Clearly, it is optimal to set $I_U(t) =1$ for any $t$. Once again, it is straightforward to verify this satisfies all the omitted IC constraints so it is indeed the optimal IC-IR mechanism.
\end{proof}

\begin{proof}[Proof of Proposition \ref{prop:mixnews}]
    
We consider a relaxed design problem as follows
\begin{gather*}
     \max_{I, p_U(T), \{ \Delta_t \} } ~ (1+w_L) p_U(T) + \mathbb{E}   [  \Delta^{\bar{v}}_\tau + w_L \Delta^{\ubar{v}}_\tau ]  \\
    \textnormal{s.t. }    \\
    \lambda \bar{v} I(\bar{v},t) - \Delta_t^{\bar{v}} \geq \lambda \bar{v} I(U,t),\quad \forall t,  \\
        \lambda \bar{v} I(\bar{v},t) - \Delta_t^{\bar{v}} - \lambda \bar{v} I(U,t)  \geq \lambda \bar{v} I(\bar{v},s) - \Delta_s^{\bar{v}} - \lambda \bar{v} I(U,s) ,\quad \forall t<s,  \\
 \lambda  I(U,t) +  \mathbb{E} [ \lambda (\mu_0 \bar{v} I(\bar{v},\tau) + (1-\mu_0) \ubar{v} I(\ubar{v},\tau)  - I(U,\tau) )  -\mu_0 \Delta^{\bar{v}}_{\tau} - (1-\mu_0)\Delta^{\ubar{v}}_{\tau}  \mid \tau > t  ] \\
 \geq \lambda I(\ubar{v},t) - \Delta^{\ubar{v}}_t,\quad \forall t, 
 \\
         p_U(T) \le  \lambda  I(U,0) + \mathbb{E} [ \lambda (\mu_0 \bar{v} I(\bar{v},\tau) + (1-\mu_0) \ubar{v} I(\ubar{v},\tau)  - I(U,\tau) )  -\mu_0 \Delta^{\bar{v}}_{\tau} - (1-\mu_0)\Delta^{\ubar{v}}_{\tau}   ].
\end{gather*}
The second constraint is the new constraint, which ensures that the buyer who receives the high signal at time $t$ does not want to delay the report to a later time $s> t$. The third constraint requires that the uninformed buyer does not want to deviate to report the arrival of a low signal. We can plug in the second line into the third line and get a relaxed expression
\begin{gather*}
  \mathbb{E} [ \lambda (\mu_0 \bar{v} (I(\bar{v},t)  - I(U,t) ) + (1-\mu_0) \ubar{v} (I(\ubar{v},\tau)  - I(U,\tau) ))  -\mu_0 \Delta^{\bar{v}}_{t} - (1-\mu_0)\Delta^{\ubar{v}}_{\tau}  \mid \tau > t  ]     \\
+ \lambda  I(U,t)  \geq \lambda I(\ubar{v},t) - \Delta^{\ubar{v}}_t,\quad \forall t, 
\end{gather*}

We use this new relaxed constraint to replace the second and third constraints.
It is also clear the ex-ante IR constraint must bind, so the objective becomes:
\begin{gather*}
     \lambda(1+w_L)\Big(    I(U,0) + \mathbb{E} [ \mu_0 \bar{v} I(\bar{v},\tau)  + (1-\mu_0) \ubar{v} I(\ubar{v},\tau)  -  I(U,\tau)  ]  \Big) \\ + \mathbb{E}   [ (1-\mu_0 (1+w_L)) (\Delta^{\bar{v}}_\tau - \Delta^{\ubar{v}}_\tau) ] 
\end{gather*}
As we focus on the case where $w_L \leq (1-\mu_0)/\mu_0$, the coefficient for $\Delta_t^{\bar{v}}$ is positive while coefficient for  $\Delta_t^{\ubar{v}}$ is negative. Because increasing $I(\bar{v},t)$ strictly increases the objective and relaxes all constraints (recall the second constraint is replaced by the new constraint), it is optimal to set $I(\bar{v},t)=T-t$. Also, because decreasing $I(\ubar{v},t)$ strictly increases the objective and relaxes all constraints, it is optimal to set $I(\ubar{v},t)=0$. 

Thus, the design problem reduces to
\begin{gather*}
     \max_{I, \{ \Delta_t \} } ~  \lambda(1+w_L)\Big(    I(U,0) + \mathbb{E} [ \mu_0 \bar{v} (T-\tau)   -  I(U,\tau)  ]  \Big)  + \mathbb{E}   [ (1-\mu_0 (1+w_L)) (\Delta^{\bar{v}}_\tau - \Delta^{\ubar{v}}_\tau) ]  \\
    \textnormal{s.t. }    \\
    \lambda \bar{v} (T-t) - \Delta_t^{\bar{v}} \geq \lambda \bar{v} I(U,t),\quad \forall t,  \\
  \mu_0 (\lambda\bar{v} (T-t  - I(U,t) ) - \Delta_t^{\bar{v}}) -  \mathbb{E} [   \lambda (1-\mu_0) \ubar{v} I(U,\tau)   + (1-\mu_0)\Delta^{\ubar{v}}_{\tau}  \mid \tau > t  ]     \\
+ \lambda  I(U,t)  \geq  - \Delta^{\ubar{v}}_t,\quad \forall t, 
\end{gather*}
 Note $\Delta_T^{\ubar{v}}=0$ and denote 
\begin{gather*}
    \Gamma(-\Delta_t^{\ubar{v}}) =    \mu_0 (\lambda \bar{v} (T-t  - I(U,t) ) - \Delta_t^{\bar{v}}) - (1-\mu_0)  \mathbb{E} [   \lambda \ubar{v} I(U,\tau)   + \Delta^{\ubar{v}}_{\tau}  \mid \tau > t  ]   
+ \lambda  I(U,t)
\end{gather*}
$\Gamma$ is clearly increasing in $-\Delta_t^{\ubar{v}}$,
so according to Grönwall's inequality, the point-wise minimum feasible $\Delta_t^{\ubar{v}}$ is the one such that equality always holds:
\begin{align*}
-\Delta_t^{\ubar{v}} &=  \lambda  I(U,t) +   \mu_0 (\lambda \bar{v} (T-t  - I(U,t) ) - \Delta_t^{\bar{v}}) - (1-\mu_0) \int_t^T  ( \lambda  \ubar{v} I(U,\tau)   +   \Delta^{\ubar{v}}_{\tau} ) f(\tau|t) \dd \tau
\end{align*}
where $f(s|t) = \lambda I_U(s)e^{-\lambda (I(U,t) - I(U,s) )}$ is the conditional density of the arrival time. Taking derivatives on both sides gives
\begin{align*}
    -\dot{\Delta}_t^{\ubar{v}} &=  -\lambda  I_U(t) +  \mu_0 ( \lambda \bar{v} (I_U(t) -1 ) - \dot{\Delta}_t^{\bar{v}}) + \lambda I_U(t) (1-\mu_0)   (\lambda  \ubar{v} I(U,t)   +   \Delta^{\ubar{v}}_{t} )    \\ & \quad -  \lambda I_U(t)  \big(\Delta_t^{\ubar{v}}  +  \lambda  I(U,t) +   \mu_0 (\lambda \bar{v} (T-t  - I(U,t) ) - \Delta_t^{\bar{v}}) \big)   \\
    \dot{\Delta}_t^{\ubar{v}} & = \lambda I_U(t) \big( 1+  \lambda \mu_0 \bar{v}(T-t)  - \mu_0  \Delta_t^{\bar{v}}    + \mu_0 \Delta_t^{\ubar{v}}   \big) -  \mu_0 (\lambda \bar{v} (I_U(t) -1 )  - \dot{\Delta}_t^{\bar{v}}) \\
    \Delta_t^{\ubar{v}} & = - \int_t^T \big( \lambda I_U(s) \big( 1+  \lambda \mu_0 \bar{v}(T-s)      \big) - \lambda \mu_0 \bar{v} (I_U(s) -1 ) \big) e^{-\mu_0 \lambda (I(U,t)-I(U,s))} \dd s \\
    & ~ + \int_t^T  \mu_0  \big(\lambda I_U(s)    \Delta_s^{\bar{v}}  - \dot{\Delta}_s^{\bar{v}}) \big) e^{-\mu_0 \lambda (I(U,t)-I(U,s))} \dd s \\
    & = \mu_0 \Delta_t^{\bar{v}} - \int_t^T \big( \lambda I_U(s) \big( 1+  \lambda \mu_0 \bar{v}(T-s)      \big) - \lambda \mu_0 \bar{v} (I_U(s) -1 )  \big) e^{-\mu_0 \lambda (I(U,t)-I(U,s))} \dd s
\end{align*}
Plugging this equation back to the objective, the only term with parameter $\Delta^{\bar{v}}_t$ is
\begin{gather*}
  (1-\mu(1+w_L))  (1-\mu_0) \mathbb{E} \left[    \Delta^{\ubar{v}}_\tau \right].
\end{gather*}
Thus, the objective is strictly increasing in $\Delta_t^{\ubar{v}}$ so the only remaining constraint is also binding:
\begin{gather*}
    \lambda \bar{v} (T-t) - \Delta_t^{\bar{v}} = \lambda \bar{v} I(U,t),\quad \forall t.
\end{gather*}
With this equation, we know 
\begin{align*}
-\Delta_t^{\ubar{v}} &=  \lambda  I(U,t)  - (1-\mu_0)  \mathbb{E} [   \lambda \ubar{v} I(U,\tau)   + \Delta^{\ubar{v}}_{\tau}  \mid \tau > t  ]   \\
 -\dot{\Delta}_t^{\ubar{v}}       &=  -\lambda  I_U(t)  + \lambda I_U(t) (1-\mu_0)   (\lambda  \ubar{v} I(U,t)   +   \Delta^{\ubar{v}}_{t} )   -  \lambda I_U(t)  (\Delta_t^{\ubar{v}}  +  \lambda  I(U,t) )   \\
  \dot{\Delta}_t^{\ubar{v}}       &=  \lambda  I_U(t) ( 1 +  \mu_0 \lambda  \bar{v} I(U,t)      +  \mu_0 \Delta_t^{\ubar{v}}   ) \\
  \Delta_t^{\ubar{v}}       &=- \int_t^T \lambda  I_U(s) ( 1 +  \mu_0 \lambda  \bar{v} I(U,s)   ) e^{-\mu_0 \lambda (I(U,t)-I(U,s))} \dd s.
\end{align*}
Substituting this expression into the objective, we get
\begin{gather*}
     \max_{I  } ~  \lambda(1+w_L)\Big(    I(U,0) + \mathbb{E} [ \mu_0 \bar{v} (T-\tau)   -  I(U,\tau)  ]  \Big)  + (1-\mu_0 (1+w_L)) \mathbb{E}   [   \lambda \bar{v}(T-\tau-I(U,\tau)) ] \\
      +  (1-\mu_0 (1+w_L)) \mathbb{E}   [  \int_\tau^T \lambda  I_U(s) ( 1 +  \mu_0 \lambda  \bar{v} I(U,s)   ) e^{-\mu_0 \lambda (I(U,\tau)-I(U,s))} \dd s   ] \\
   \Leftrightarrow \max_{I  } ~  \lambda(1+w_L) I(U,0)  + \lambda  \mathbb{E}   [    \bar{v}(T-\tau- (1+ (1-\mu_0) (1+w_L))I(U,\tau)) ]   \\
    +  (1-\mu_0 (1+w_L)) \mathbb{E}   [  ( 1 +  \mu_0 \lambda  \bar{v} I(U,\tau)   ) \frac{e^{ \lambda (1-\mu_0) (I(U,0)-I(U,\tau))} -1}{1-\mu_0}   ] 
\end{gather*}
The second line can be expressed as
\begin{gather*}
 (1-\mu_0 (1+w_L)) \int_0^T   ( 1 +  \mu_0 \lambda  \bar{v} I(U,\tau)   ) \frac{e^{ \lambda (1-\mu_0) (I(U,0)-I(U,\tau))} -1}{1-\mu_0} \lambda I_U(\tau) e^{-\lambda (I(U,0) -I(U,\tau  ) }\dd \tau \\
 =  (1-\mu_0 (1+w_L)) \int_{0}^{I(U,0)}   ( 1 +  \mu_0 \lambda  \bar{v} I(U,\tau)   ) \frac{e^{ \lambda (1-\mu_0) (I(U,0)-I(U,\tau))} -1}{1-\mu_0} \lambda  e^{-\lambda (I(U,0) -I(U,\tau  ) }\dd I(U,\tau)
\end{gather*}
Thus, the second line is just a function of $I(U,0)$ and does not depend on any other $I(U,t)$ for $t>0$. Similarly, the only term in the first line that does depend on $I(U,t)$ is
\begin{gather*}
    \lambda \bar{v} \int_0^T (T-\tau) \lambda I_U(\tau) e^{-\lambda (I(U,0)- I(U,\tau)  )} \dd \tau.
\end{gather*}
For any fixed $I(U,0)$ it is optimal to minimizes $I(U,t)$ so it is optimal to frontload access $I_U(t)$ as much as possible so that a trial $I_U(t) = 1_{t\leq t_0}$ is optimal. Once again, it is straightforward to verify this satisfies all the omitted IC constraints, so it is indeed the optimal IC-IR mechanism.

\end{proof}

\begin{proof}[Proof of Proposition \ref{prop:infinite_horizon}]
Once again, first consider the relaxed problem. Since there is discounting now, we let $p_U, p_{v,t}$ functions denote the time-zero \textit{discounted} cumulative flow payments. We consider the same problem relaxation, and the same logic as Theorem \ref{thm:trial} implies that $p^L_U(\infty) = p^H_U(\infty)$, so the problem is 
{\small \begin{gather}
     \max_{q, p} ~ w_L p_U(\infty) + \mathbb{E} \left[p_\tau(v,\infty)\mathbbm{1}[\tau \le T] + p_U(\infty)\mathbbm{1}[\tau > T] \mid \theta = H \right] \label{obj:R_inf} \\
    \text{s.t. } \int_t^{\infty} e^{-r(x-t)} \lambda v I_{v,t}(x) \dd x - e^{rt}(p_{v,t}(\infty)-p_U(t)) \geq \int_t^{\infty} e^{-r(x-t)} \lambda v I_U(x) \dd x - e^{rt}(p_{U}(\infty)-p_U(t))  \quad \forall t, \label{cst:ICR_inf} \\
    \int_t^{\infty} e^{-r(x-t)} \lambda v I_{v,t}(x) \dd x - e^{rt}(p_{v,t}(\infty)-p_U(t)) \geq \int_t^{\infty} e^{-r(x-t)} \lambda v I_{v',t}(x) \dd x - e^{rt}(p_{v',t}(\infty)-p_U(t)) \label{cst:ICV_inf}\\
         \mu_0\int_0^{\infty} e^{-rs} \lambda I_U(s) \dd s + \mu_0 \mathbb{E}\left[\left. \int_{\tau}^{\infty}e^{-rs}  \lambda v \left( I_{\tau}(s) - I_U(s)\right)\dd s  - \Delta_{v,\tau} \right\vert \theta = H\right] - p_U(0) \geq 0 . \label{cst:IRR_inf}
\end{gather}}


\normalsize

The only difference is we now have discounting here, and the horizon is taken to be infinite. Once again, define $\Delta_{v,t} = p_{v,t}(\infty) - p_U(\infty)$, and set $\Delta_{\cdot,\infty} = 0$, so that $\Delta_t$ is well-defined for any realization $t$ of $\tau$. We change the problem into these price variables as before. The objective can be rewritten:
\begin{align*}
    &w_L p_U(\infty) + \mathbb{E} \left[p_\tau(\infty)\mathbbm{1}[\tau \le \infty] + p_U(\infty)\mathbbm{1}[\tau > \infty] \mid \theta = H \right] \\
    &= (w_L + 1)p_U(\infty) + \mathbb{E} \left[(p_\tau(\infty)-p_U(\infty))\mathbbm{1}[\tau \le \infty] \mid \theta = H \right] \\
    &= (w_L + 1)p_U(\infty) + \mathbb{E} \left[\Delta_{v,\tau} \mid \theta = H \right]
\end{align*}Define the interim utility as before:
\[ u(v,t) = \int_t^{\infty} e^{-rx} \lambda v I_{v,t}(x) \dd x -  p_{v,t}(\infty), \]
and the discounted cumulative intensities as before:
\[ I(v,t) = \int_{t}^\infty e^{-rx} I_{v,t}(x) \dd x, \]
\[ I(U,t) = \int_t^\infty e^{-rx}I_U(x) \dd x. \]
We can apply the envelope theorem to \eqref{cst:ICV_inf}, and we get 
\[ u(v,t) = u(\ubar{v},t) + \int_{\ubar{v}}^v \lambda I(w,t) \ dw \]
\[ \ge \int_{\ubar{v}}^v \lambda I(w,t) \ \dd w + \lambda \ubar{v} I(U,t) - p_U(\infty)  \]
Again, a similar argument to Lemma \ref{Lemma_low_bind} implies $I(w,t) \ge I(U,t)$, so rewriting the above implies that 
\[ \Delta_{v,t} \le \lambda v I(v,t) -  \int_{\ubar{v}}^v \lambda I(w,t) \dd w - \lambda \ubar{v}I(U,t). \]

Once again, it is straightforward to see that the IR constraint must bind; otherwise, we could increase $p_U(\infty)$ and do better on the objective. So we can plug $p_U(\infty)$ into the objective, and the objective becomes 
\begin{gather*}
    \lambda \mu_0(w_L + 1) I(U,0) + \mu_0(w_L + 1) \mathbb{E}\left[\left. \lambda v \left( I(v,\tau) - I(U,\tau)\right) - \Delta_{v, \tau }\right\vert \theta = H\right] + \mathbb{E} \left[\Delta_\tau \mid \theta = H \right] \\
    = \lambda \mu_0(w_L + 1) I(U,0) + \mu_0(w_L + 1) \mathbb{E}\left[\left. \lambda v \left( I(v,\tau) - I(U,\tau)\right) \right\vert \theta = H\right] \\+ (1 - \mu_0(w_L + 1)) \mathbb{E} \left[\Delta_{v,\tau} \mid \theta = H \right]
\end{gather*}
Note that only the last term depends on $\Delta_{v,t}$, and it is strictly increasing, so it is optimal to raise $\Delta_{v,t}$ to make the combined IC bind. So we can again reduce the optimization problem to
\begin{gather*}
    \max_{I} ~ \begin{Bmatrix} \lambda \mu_0 (w_L + 1) I(U,0) + \mu_0 (w_L + 1)\mathbb{E}\left[\left. \lambda v \left( I(v,\tau) - I(U,\tau)\right) \right\vert \theta = H\right] \\ + (1 - \mu_0(w_L + 1)) \mathbb{E} \left[\lambda v I(v,t) -  \int_{\ubar{v}}^v \lambda \max(I(U,t),I(w,t)) \dd w - \lambda \ubar{v}I(U,t)\mid \theta = H \right] \end{Bmatrix}.
\end{gather*}
Or equivalently, 
{\small \begin{align}
    \max_{M \in \mathcal{M}} ~ &  \left\{ \mu_0(w_L + 1)  \lambda I(U,0)   +  \mathbb{E}\left[ \left. \begin{pmatrix} \lambda v I(v,\tau) - (1 - \mu_0(w_L + 1))\frac{1-F(v)}{f(v)}  \lambda r(v,\tau) \\ - \lambda\mu_0(w_L + 1)v I(U,\tau) -(1 - \mu_0(w_L + 1))\lambda\ubar{v} I(U,\tau) \end{pmatrix}\mathbbm{1}[\tau \le T] \right| \theta = H \right]\right\}. \notag 
\notag
\end{align}}
As before, pointwise maximization of $I(v,\tau)$ dictates that 
\[ I_{v,t}(s) = \begin{cases}
    1 & v \ge v_0,\\
    I_{U}(s) & v < v_0,
\end{cases} \]
where $v_0$ solves \eqref{eqn:thm_v0} as before.
Plugging in, we find that the expectation term becomes 
\[ \mathbb{E}\left[\lambda \left(\mu_0(w_L+1) \int_{v_0}^{\bar{v}}v  f(v)\ \dd v  + v_0(1 - F(v_0))(1 - \mu_0(w_L+1))\right)(T - \tau - I(U,\tau) \right ]. \]
As before, define 
\[ \pi_0 = \mu_0(w_L+1) \int_{v_0}^{\bar{v}}v  f(v)\ \dd v  + v_0(1 - F(v_0))(1 - \mu_0(w_L+1)). \]
We have that $\pi_0$ is still the expected virtual surplus. So the problem becomes 
\begin{align}
    \max_{M \in \mathcal{M}} ~ &  \left\{ \mu_0(w_L + 1)  \lambda I(U,0)   +  \mathbb{E}\left[\lambda \pi_0 (T - \tau - I(U,\tau) \right ]\right\} \notag 
\notag
\end{align}
Equivalently, 
\begin{gather*}
    \max_{I_U} ~ \left \{ \frac{\mu_0(w_L + 1)}{\pi_0}  \int_0^{\infty} e^{-rs}  \lambda I_U(s)   \dd s  +   \mathbb{E}  \left[ \left. \int_{\tau}^{\infty} e^{-rs}    \lambda (1 - I_U(s))  \dd s \right\vert \theta = H \right] \right \}
\end{gather*}
Now, we use the fact that $\tau$ given $I_U$ is distributed according to density $\lambda I_U(t) e^{-\int_0^t \lambda I_U(s) \dd s}$.
Plugging this in, the problem is
\begin{gather*}
   \max_I  \int_0^\infty  e^{-rt} \frac{\mu_0(w_L + 1)}{\pi_0}\lambda I_U(t) \dd t + \int_0^\infty  \left (\int_t^\infty e^{-rs}  \lambda (1 - I_U(s)) \dd s \right) \lambda I_U(t) e^{-\int_0^t \lambda I_U(s) \dd s}  \dd t
\end{gather*}
We change the variable to reformulate the problem as an optimal control. Define the cumulative access
\begin{gather*}
    I(t) = \int_0^{t} \lambda I_U(s) \dd s.
\end{gather*}
Note that $\dot{I}(t) = \lambda I_U(t)$. We rewrite the objective function in its equivalent form:
\begin{align*}
& \int_0^\infty  e^{-rt} \lambda \frac{\mu_0(w_L + 1)}{\pi_0} I_U(t) \dd t + \int_0^\infty  \left (\int_t^\infty e^{-rs}  \lambda (1 - I_U(s)) \dd s \right) \lambda I_U(t) e^{-\int_0^t \lambda I_U(s) \dd s}  \dd t \\ 
&= \int_0^\infty  e^{-rt} \frac{\mu_0(w_L + 1)}{\pi_0} \dot{I}(t)\dd t + \int_0^\infty  \left (\int_t^\infty e^{-rs}  \lambda (1 - I_U(s)) \dd s \right) \dot{I}(t) e^{-I(t)}  \dd t \\
&= \int_0^\infty  e^{-rt} \frac{\mu_0(w_L + 1)}{\pi_0} \dot{I}(t) \dd t + \int_0^\infty  \left (\frac{\lambda}{r}e^{-rt} - \lambda \int_t^\infty e^{-rs} I_U(s) \dd s \right) \dot{I}(t) e^{-I(t)}  \dd t \\
&= \int_0^\infty  e^{-rt} \dot{I}(t) \left(\frac{\mu_0(w_L + 1)}{\pi_0} + \frac{\lambda}{r} e^{-I(t)} \right)\dd t - \lambda \int_0^\infty  \left ( \int_t^\infty e^{-rs} I_U(s) \dd s \right) \dot{I}(t) e^{-I(t)}  \dd t \\
&= \int_0^\infty  e^{-rt} \dot{I}(t) \left(\frac{\mu_0(w_L + 1)}{\pi_0} + \frac{\lambda}{r} e^{-I(t)} \right)\dd t - \lambda \int_0^\infty  \left ( \int_0^s e^{-rs} I_U(s)  \dot{I}(t) e^{-I(t)} \dd t \right)  \dd s 
\end{align*}
where the last step changed the order of integration. 

Simplfying further, the objective becomes
\begin{align*}
&= \int_0^\infty  e^{-rt} \dot{I}(t) \left(\frac{\mu_0(w_L + 1)}{\pi_0} + \frac{\lambda}{r} e^{-I(t)} \right) \dd t+\int_0^\infty \lambda  e^{-rs} I_U(s) \left( e^{-I(s)} - 1 \right)   \dd s \\
&= \int_0^\infty  e^{-rt} \dot{I}(t) \left(\frac{\mu_0(w_L + 1)}{\pi_0} + \frac{\lambda}{r} e^{-I(t)} \right) \dd t+\int_0^\infty e^{-rt} \dot{I}(t) \left( e^{-I(t)} - 1 \right)   \dd t \\
&= \int_0^\infty  e^{-rt} \dot{I}(t) \left(\frac{\mu_0(w_L + 1)}{\pi_0} - 1 + \left( \frac{\lambda}{r} + 1\right) e^{-I(t)} \right) \dd t .
\end{align*}
All together, we have the control problem
\begin{gather}
       \max_{I_U(t)}   \int_0^\infty e^{-rt} \left( \left( \frac{\lambda}{r} + 1  \right)e^{-I(t)} + \frac{\mu_0(w_L + 1)}{\pi_0}-1 \right)\lambda I_U(t) \  \dd t \label{obj:ctrl_inf} \\
       \textnormal{subject to} \quad \dot{I}(t) = \lambda I_U(t) \notag \\
       I_U(t) \in [0,1], \quad I(0) = 0, \quad I(T) \textnormal{ free}. \notag 
\end{gather}
Again for convenience, define $u = \lambda I_U$, as the control variable. According to the Pontryagin maximum principle, for any optimal solution $(I(t),u(t))$, there exists a piece-wise differentiable costate $\rho(t)$ such that the following holds:
\begin{align}
     \lim_{t \to \infty} \dot{\rho}(t) &= 0,  \label{eqn:control_transversality} \\
    \dot{\rho} (t) &= -\left( \frac{\lambda}{r} +1\right)e^{-I(t)} u(t), \label{eqn:control_evolution}\\
     \dot{I}(t) &= u(t)  \begin{cases}  
        =\lambda \quad &\text{ if } J(t) > 0\\ 
        \in [0,\lambda ] \quad &\text{ if } J(t) = 0 \\
        =0 \quad &\text{ if } J(t) < 0    
    \end{cases}  \label{eqn:control_optimality}\\
     J(t) &= \left( \frac{\lambda}{r} + 1  \right)e^{-I(t)} + \frac{\mu_0(w_L + 1)}{\pi_0}-1 - \rho(t) \label{eqn:control_helper}.
\end{align}
The difference is now we have a transversality condition, instead of a free end-time fixed point. By differentiating equation \eqref{eqn:control_helper} (at wherever $J$ is differentiable) and plugging in equation \eqref{eqn:control_evolution}, we get
\begin{align*}
    \dot{J}(t) &=  -\lambda e^{-I(t)} - u(t) \left( \frac{\lambda}{r} + 1  \right)e^{-I(t)}  - \dot{\rho}(t),\\
        \dot{J}(t) &= -\lambda e^{-I(t)}.
\end{align*}
Because $\rho$ is piece-wise differentiable and so is $J$, for any $t$ the right (left) derivative of $J$ always exists and equals the right (left) limit of derivatives. Formally:
\begin{gather*}
    \dot{J}^+ (t) = \lim_{s\searrow t} \dot{J}(s), \quad \dot{J}^- (t) = \lim_{s\nearrow t} \dot{J}(s).
\end{gather*}
From the continuity of $I(t)$, we know $J$ is always differentiable and $ \dot{J}(t) = -\lambda e^{-I(t)} < 0$. 
If $\frac{\mu_0(w_L + 1)}{\pi_0} \ge 1$, it follows that in the limit $t \to \infty$, the limit of $J$ is positive, and hence $J$ is positive everywhere, so the mechanism sells the entire service ex ante. 

Suppose $\frac{\mu_0(w_L + 1)}{\pi_0} < 1$. We now argue there exists $t_0$ such that $J(t_0) =0$. Suppose not; then either $J(t)$ is either always positive or always negative. If $J(t)$ is always positive, then according to equation \eqref{eqn:control_optimality}, $u(t) =\lambda $ and $I(t) = \lambda t$. This implies
\begin{align*}
    \lim_{t \to \infty} J(\infty) = \frac{\mu_0(w_L + 1)}{\pi_0}-1< 0.
\end{align*}
This is a contradiction. Suppose $J(t)$ is always negative, then according to equation \eqref{eqn:control_optimality}, $u(t) =I(t) =0$. This implies
\begin{align*}
    \rho(t) &= \dot{\rho}(t) = 0, \\
    J(0) &= \frac{\lambda}{r} + \frac{\mu_0(w_L + 1)}{\pi_0} > 0 .
\end{align*}
This is a contradiction.

Now because $J(t_0) =0$ and $\dot{J}(t) < 0$, we know 
\begin{align*}
    J(t) \begin{cases}
      > 0, \quad \text{if } t< t_0,\\
      < 0, \quad \text{if } t> t_0.
    \end{cases}
    \implies ~
     u(t)= \begin{cases}
       \lambda, \quad \text{if } t< t_0,\\
       0, \quad \text{if } t> t_0.
    \end{cases}   
\end{align*}
This implies $I(t) = \lambda \min \{ t,t_0\}.$ We also know $\rho(t)=0$ for any $t>t_0$ and so $\rho(t_0)=0$ from the continuity. Then we can check equation \eqref{eqn:control_helper} at $t_0$:
\begin{align*}
    J(t_0) = \left( \frac{\lambda}{r} + 1  \right)e^{-\lambda t_0} + \frac{\mu_0(w_L + 1)}{\pi_0} -1 = 0.
\end{align*}
This equation has a unique solution. In fact, we can analytically solve for $t_0$ in this case:
\begin{gather*}
    \left( \frac{\lambda}{r} + 1  \right)e^{-\lambda t_0} + \frac{\mu_0(w_L + 1)}{\pi_0}-1 = 0 \\
    \implies t_0 = \frac{1}{\lambda}\ln \left( \frac{\frac{\lambda}{r} + 1 }{1 - \frac{\mu_0(w_L + 1)}{\pi_0}}\right)
\end{gather*}

To finish characterizing the optimal solution to our original design problem, we now construct the appropriate payment scheme. Consider
\[ p_U(t) = \mu_0 \frac{\lambda}{r}(1 - e^{-rt_0}) + \frac{\lambda}{r} (\pi_0 - v_0(1 - F(v_0)))e^{-rt_0}(1 - e^{-\lambda t_0}) \]
\[ p_{v,t}(s) = p_U(t) + \frac{\lambda}{r}v_0 e^{-rt_0} \]
Once again, the payment scheme frontloads payment, so the IR needs to be checked only at the initial time and the time of report. By the same arguments as before, IR and IC holds. Since the $I_U(t)$, $I_{v,t}(s)$ constructed maximizes the value of the relaxed problem, and we have just constructed a mechanism that is feasible in the original problem with all of the IC and IR constraints, it follows that this mechanism must be optimal.  \end{proof}

\section{Additional Signaling Refinements}
The reason why we use invoke $D_1$ in the main model is due to the rich response action space for the buyer.
In this section, we discuss what would intuitive criterion, a weaker refinement concept, predict. 
\begin{definition}
    An equilibrium with equilibrium payoff $(\pi_L,\pi_H)$ does not survive the Intuitive Criterion if and only if there exists a deviating mechanism $(I,p)$ such that
    \begin{itemize}
        \item For any $\mu \in [0,1]$ and any $\sigma \in \Sigma(m,\mu)$, $\pi_L(\sigma) < \pi_L$.
        \item  For any $\sigma \in \Sigma(m,1)$, $\pi_H(\sigma) > \pi_H$.
    \end{itemize}
\end{definition}
An equilibrium does not survive the Intuitive Criterion if 
in all continuation games consistent with some initial belief of the buyer, the low-type seller gets strictly worse off, and the high-type is strictly better off in all continuation games in which the buyer believes $\theta=H$. 

\begin{proposition}
\label{prop:intuitive}
The set of equilibrium payoffs that survive the Intuitive Criterion is identical to the set of equilibrium payoffs that survive the $D_1$ Criterion, if $v_H\geq 1$, or equivalently, if $f(1) - (1-\mu_0) (1-F(1)) \leq 0$.
\end{proposition}
\begin{proof}[Proof of Proposition \ref{prop:intuitive}]
Let the payoff at point $H$ be $(\pi^H_L,\pi^H_H)$. First, we prove that for any equilibrium with payoff $(\pi_L,\pi_H)\in \Pi$ such that $\pi_L > \pi^H_L$, the equilibrium does not survive Intuitive Criterion. To see this, recall the trial mechanism corresponding to point $H$ leads to a payoff pair $(\pi_L^H, \pi_H^H)$.  Now consider the deviation to a trial mechanism $m''$ that is very close to point $H$ trial mechanism. It sells the trial with length $t_H$ at a price $\pi_L^H$ and sells the remaining service at a price $(v_H+ \varepsilon)(T-t_H)$. Under the deviation mechanism $m''$, because the post-trial rate $v_H + \varepsilon > 1$ the buyer who never receives the signal would not purchase the post-trial service regardless of the belief of $\theta$. Thus, the low-type seller gets at most $\pi_L^H < \pi_L$. In contrast, if the buyer believes $\theta =1$ then for sufficiently small $\varepsilon$, the buyer would be willing to purchase the trial and purchase the post-trial service if $v>v_H+\varepsilon$. Thus, the high type gets a payoff that is close to $\pi^H_H> \Pi(\pi_L)\geq \pi_H$ (this comes from the uniqueness of the payoff). In conclusion, the payoff $(\pi_L,\pi_H)$ does not survive Intuitive Criterion.

Second, we prove that for any equilibrium with payoff $(\pi_L,\pi_H)\in \Pi$ such that $\pi_L \in (\pi^D_L,\pi^H_L]$, the equilibrium does not survive Intuitive Criterion. To see this, recall there is a trial mechanism that achieves $(\pi_L,\Pi(\pi_L))$, which sells the trial with length $t_0$ at a price $\pi_L$ and sells the post-trial service at a price $v_0(T-t_0)$. We know from the comparative analysis that $t_0 > t_M$ and $v_0 \geq v_H$. Now consider the deviation to a trial mechanism $m'$ that is very close to this mechanism. The deviation sells the trial with length $t_0-\varepsilon$ at a price $\pi_L-\varepsilon^2$ and sells the post-trial service at a price $(v_0+ \varepsilon^2)(T-t_0)$ for some $\varepsilon$ sufficiently small. Under the deviation mechanism $m'$, because the post-trial rate $v_0 + \varepsilon^2 > 1$ the buyer who never receives the signal would not purchase the post-trial service regardless of the belief of $\theta$. Thus, the low-type seller gets at most $\pi_L-\varepsilon^2 < \pi_L$. However, if the buyer believes $\theta=1$, the high type gets
\begin{gather*}
    \pi_L -\varepsilon^2 + (1-e^{-\lambda (t_0-\varepsilon)}) (T-t_0+\varepsilon) \mathbb{P}(v> v_0+ \varepsilon^2) (v_0 + \varepsilon^2).
\end{gather*}
In contrast,
\begin{gather*}
    \Pi(\pi_L) =\pi_L  + (1-e^{-\lambda t_0}) (T-t_0)  \mathbb{P}(v> v_0) (v_0).
\end{gather*}
Because the function $(T-t)(1-\exp(-\lambda t))$ is a strictly concave function maximized at $t_M$, the deviation leads to an increase of the high type's payoff.

Now for $\pi_L \in [0,\pi_L^D]$, one can easily show any equilibrium payoff $(\pi_L,\pi_H)$ such that $\pi_L< \Pi(\pi_L)$ does not survive Intuitive Criterion by using a deviation that modifies the trial mechanism that achieves $(\pi_L,\Pi(\pi_L))$ using the same argument as Lemma \ref{lemma:obviouslynotD1}.

To complete the proof, we finally argue any equilibrium payoff $(\pi_L,\pi_H) \in  \Pi_{D1}$ survives Intuitive Criterion. We use a proof by contradiction. Suppose not; then there exists a deviation mechanism $m=(I,p)$ such that  
    \begin{itemize}
        \item For any $\mu \in [0,1]$ and any $\sigma \in \Sigma(m,\mu)$, $\pi_L(\sigma) < \pi_L$.
        \item  For any $\sigma \in \Sigma(m,1)$, $\pi_H(\sigma) > \pi_H$.
    \end{itemize}
In particular, there exists $\sigma \in \Sigma(m,1)$ such that $\pi_L(\sigma) < \pi_L$ and $\pi_H(\sigma) > \pi_H$. This means $ E^*(\pi_L) > \Pi(\pi_L)$, where $E^*(\pi_L)$ is defined as the solution of the following optimization problem:
\begin{gather*}
     E^*(\pi_L)= \max_{q, p_U^\theta (T), \{ \Delta_{v,t}^\theta \} } ~ p^H_U(T) + \mathbb{E} \left[\Delta_{v,\tau}^H \mid \theta = H \right] \notag \\
    \textnormal{s.t. } \quad  \text{(IC-V) + (IC-U)}  \\
         p_U^H(T)   \le  I(U,0) +  \mathbb{E}\left[  v\left( I\left(v,\tau \right) -I(U,\tau) \right)  \left.  - \Delta_{v,\tau} \right\vert \theta = H\right] , \\
         p_U^H(T) \le p_U^L(T)  \le \pi_L.
\end{gather*}
In this problem, we solve for the relaxed IC-IR mechanism that maximizes the high-quality seller subject to the additional constraint that the buyer believes the seller is of high quality, and we require the payoff to the lower-quality seller to be bounded by $\pi_L$.

Denote the solution to the above problem as $m^*$.  Clearly, 
\begin{gather}
\label{eq:proof:intuitive}
     \mathbb{E} \left[\Delta_{v,\tau}^H \mid \theta = H \right]|_{m^*} = E^*(\pi_L)_{m^*} - p_U^H(T)|_{m^*}  > \Pi(\pi_L) - \pi_L.
\end{gather}

On the other hand, from our previous analysis we know there is a trial mechanism  that achieves $(\pi_L,\Pi(\pi_L))$ and solves
\begin{gather*}
    \max_{q, p} ~  - p^L_U(T) +  p^H_U(T) + \mathbb{E} \left[\Delta_{v,\tau}^H \mid \theta = H \right]  \\
    \text{subject to IC, IR} 
\end{gather*}
Moreover, Lemma \ref{lem:uninformed_payments_equal} showed that when $w_L = -1$, the ex-ante IR constraint can be neglected and $p_U^H(T) = p_U^L(T)$. Thus,
\begin{gather*}
    \Pi(\pi_L) - \pi_L =  \max_{q, p_U^\theta (T), \{ \Delta_{v,t}^\theta \} } ~ \mathbb{E} \left[\Delta_{v,\tau}^H \mid \theta = H \right]   \\
   {s.t. } \quad 
         \text{(IC-V) + (IC-U)} .
\end{gather*}
Note that $m^*$ is a feasible solution to this relaxed problem so 
\begin{gather*}
    \mathbb{E} \left[\Delta_{v,\tau}^H \mid \theta = H \right]|_{m^*} \leq   \Pi(\pi_L) - \pi_L.
\end{gather*}
This contradicts \eqref{eq:proof:intuitive}. Thus, we proved that any equilibrium payoff $(\pi_L,\pi_H) \in  \Pi_{D1}$ survives Intuitive Criterion. 
\end{proof}

Because the intuitive criterion is weaker than D1, it is not surprising that we need additional restrictions on parameters. We view the similarity of the characterizations in Propositions \ref{prop:D1} and \ref{prop:intuitive} as a robustness result.

\section{Multiple Types}
\label{sec:extensions}
To discuss the robustness of our results, we consider a multiple-seller-type extension.
The seller is only partially informed about the match quality. We use $s=1,2,...,S$ to denote the seller's signal. The seller receiving signal $s$ believes $\theta=H$ with probability $\mu_s$ and $0 = \mu_1<\mu_2<...<\mu_S\leq 1$. The seller receives the signal $s$ with probability $p_s$. Denote $\mu_0 = \sum_s p_s \mu_s$ as before.

To facilitate analysis,we focus on the baseline model where the seller only chooses access $I\in[0,1]$, though the results generalize straightforwardly with general $\mathcal{D}$. However, we do impose one restriction that may be with loss of generality. In particular, instead of assuming the seller can propose a mechanism that elicits information from both the seller and the buyer, we assume the proposed mechanism only elicits information from the buyer. 

In the baseline model, it turns out that imposing this restriction to mechanisms that only elicit information from the buyer is without loss, because the boundary outcomes are achieved by complete pooling mechanisms where the allocation does not depend on the seller's type. 
From the application perspective, we view the restriction that mechanisms can only elicit information from the buyer as reasonable. Firms never publicize their consumer data or their algorithms, so a mechanism that depends on the seller's private information is hard to monitor.  the baseline specification.

In a given equilibrium, denote $m_s$ as the on-path mechanism proposed by the seller of type $s$, and denote $\nu_s$ as the buyer's belief about $\theta$ facing the mechanism $m_s$. If types $s$ and $s'$ pool in the same mechanism then $m_s=m_{s'}$ and $\nu_s = \nu_{s'}$. Denote the sellers' equilibrium payoff conditioned on $\theta=H$ as $\pi_H^s$ and the equilibrium payoff conditioned on $\theta=L$ as $\pi_L^s$. 

First, we aim to show that all reasonable equilibrium payoffs can be achieved by equilibria where sellers propose trial mechanisms. The definition of reasonable equilibrium is analogous to the main text: namely, the seller's equilibrium payoff $\mu_s \pi_H^s + (1-\mu_s) \pi_L^s$ is increasing in $s$. 
\begin{lemma}
\label{extend_reasonable}
In any equilibrium, $\pi^s_H$ is increasing in $s$ while $\pi^s_L$ is decreasing in $s$. Consequently, in reasonable equilibria, $\pi_H^s\geq \pi_L^s$ for all $s$.
\end{lemma}
\begin{proof}[Proof of Lemma \ref{extend_reasonable}]
    The incentive of the seller requires
\begin{align*}
        \pi_H^s \mu_s + \pi_L^s (1-\mu_s) &\geq  \pi_H^{s'} \mu_s + \pi_L^{s'} (1-\mu_s), \quad \forall s, s',\\
         \pi_H^{s'} \mu_{s'} + \pi_L^{s'} (1-\mu_{s'}) &\geq  \pi_H^{s} \mu_{s'} + \pi_L^{s} (1-\mu_{s'}), \quad \forall s, s'.       
\end{align*}
Combing the two constraints we get
\begin{align*}
   \pi_H^s \mu_s + \pi_L^s (1-\mu_s) +  \pi_H^{s'} \mu_{s'} + \pi_L^{s'} (1-\mu_{s'}) \geq   \pi_H^{s'} \mu_s + \pi_L^{s'} (1-\mu_s) + \pi_H^{s} \mu_{s'} + \pi_L^{s} (1-\mu_{s'}) \\
   (\pi_H^{s'}-\pi_H^{s} - \pi_L^{s'} +\pi_L^{s}) ( \mu_{s'}-\mu_{s})
\end{align*}
This implies $\pi_H^s-\pi_L^s$ is increasing in $s$. Note that $\pi_H^s, \pi_L^s$  cannot be strictly larger than $\pi_H^{s'}, \pi_L^{s'}$ simultaneously, so we know   $\pi_H^s$ is increasing in $s$ while $\pi_L^s$ is decreasing in $s$.

If in some equilibrium and for some $s$, $\pi_H^s < \pi_L^s$, by monotonicity, we know $\pi_H^1 < \pi_L^1$. If the lowest type seller (type 1) is not pooling with any other types, then $\pi^1_H=\pi^1_L=0$ which is a contradiction.
If type 1 is pooling with any other types, $s'$, then because $\pi_H^1 < \pi_L^1$, 
\begin{gather*}
    \pi_H^1 \mu_1 + \pi_L^1 (1-\mu_1) > \pi_H^1 \mu_{s'} + \pi_L^1 (1-\mu_{s'})=   \pi_H^{s'} \mu_{s'} + \pi_L^{s'} (1-\mu_{s'}) ,
\end{gather*}
which implies the equilibrium is not a reasonable equilibrium.
\end{proof}

Now we can replicate the first main result in the main model.
\begin{proposition} 
\label{multi_trial}
All reasonable equilibria payoff vectors can be achieved by sellers proposing (potentially different) trial mechanisms.
\end{proposition}

\begin{proof}[Proof of Proposition \ref{multi_trial}]
    For any seller type $s$, the proposed mechanism $m_s$ leads to a buyer's belief of $\nu_s$, and the seller's conditional payoff pair $(\pi_H^s,\pi_L^s)$ (the seller's payoff conditioned on $\theta=H$ and $\theta=L$ respectively). We proceed to prove that any pair of reasonable conditional payoff $(\pi_H,\pi_L)$, which, can be achieved by any IC-IR dynamic mechanism with the initial belief $\nu_s$, can be achieved by trial mechanisms. To this end, we need to study the boundary of the payoff set and so to study the following dynamic mechanism design with weight for weight $w_L \le (1-\mu_0)/\mu_0$:
\begin{align*}
    \max_{M \in \mathcal{M}} ~ & \left \{ w_L p_U(T) + \mathbb{E} \left[p_{v,\tau}(T)\mathbbm{1}[\tau \le T] + p_U(T)\mathbbm{1}[\tau > T] \mid \theta = H \right] \right \}  \\
    \text{subject to} \quad & 
    \lambda v I(v,t) - p_{v,t}(T) \geq \lambda v I(U,t)  - p_U(T) \qquad \forall (v, t) \tag{IC-U}\\
    & \lambda v I(v,t) - p_{v,t}(T)  \geq \lambda v I(v',t) - p_{v',t}(T) \qquad \forall (v, v', t) \tag{IC-V}\\
    &  \nu_s \mathbb{E}\left[N_T v - p_{v,\tau}(T) \mid \theta = H \right] - (1-\nu_s) p_U(T) \ge 0 \tag{IR-0} 
\notag
\end{align*}
This is a similar dynamic mechanism design problem as we studied in the main model, except that the mechanism does not depend on the seller's report and the prior of the buyer is $\nu_s$ instead of $\mu_0$ (which affects the IR condition). Now by plugging in the binding IR constraint, we get    
\begin{align*}
    \max ~ &  (w_L+1)\mu_0\left( \mathbb{E}\left[ \lambda v I(v,t)  - \Delta_{v,t}(T) \mid \theta = H \right]  \right)   + \mathbb{E} \left[\Delta_{v,t}(T) \mid \theta = H \right]  \\
   \text{subject to} \quad & 
    \lambda v I(v,t) - \Delta_{v,t}(T) \geq \lambda v I(U,t)  \qquad \forall  v \tag{IC-U}\\
    & \lambda v I(v,t) - \Delta_{v,t}(T)  \geq \lambda v I(v',t) - \Delta_{v',t}(T) \qquad \forall  (v, v') \tag{IC-V}
\end{align*}
This is exactly what we study in the proof of Theorem \ref{thm:trial}. Thus, we know all reasonable payoffs achieved by IC-IR mechanisms can be indeed achieved by trial mechanisms. 

To prove Proposition \ref{multi_trial}, we start from any equilibrium payoff vector $(\pi_1,\pi_2,...,\pi_s)$. Suppose this payoff vector is achieved in an equilibrium where type $s$ proposes $m_s$, which leads to buyer belief $\nu_s$ and seller conditional payoff pair $(\pi_H^s,\pi_L^s)$. We know $\pi_s = \mu_s \pi_H^s + (1-\mu_s) \pi_L^s$. To achieve   $(\pi_1,\pi_2,...,\pi_s)$ by trial mechanisms, we identity the trial mechanism $t_s$ that achieves $(\pi_H^s,\pi_L^s)$ with the initial belief $\nu_s$, and construct a similar equilibrium by changing the seller's on-path proposal from $m_s$ to $t_s$.
\end{proof}
Note that this result is slightly different the main model, where all reasonable equilibrium payoffs can be achieved by sellers of both types proposing the same trial mechanism. In the extension, it is possible that sellers of different types propose different trial mechanisms in separating equilibria. However, if we restrict attention to equilibria surviving the D1 refinement, we can recover the result of Proposition \ref{prop:D1}:
\begin{proposition}
\label{extent:D1}
Any equilibrium payoff vector that survives the D1 criterion can be achieved by sellers of all types pooling in the same trial mechanisms with trial length $t_M$ and post-trial rate $v_M$. All D1-surviving payoff vectors are Pareto dominated.
\end{proposition}

\begin{proof}[Proof of Proposition \ref{extent:D1}]
    Consider an D1-surviving-equilibrium payoff vector $(\pi_1,\pi_2,...,\pi_S)$, which is achieved in an equilibrium where type $s$ proposes $m_s$, leading to a buyer's belief of $\nu_s$, and the seller's conditional payoff pair $(\pi_H^s,\pi_L^s)$. Note that $\pi^1_L\geq 0$ so we discuss two cases.
    
    \textbf{First}, suppose $\pi^1_L=0$; then the incentive compatibility requires $\pi^s_L \leq \pi^1_L = 0$ for any $s$. We proceed to prove that $\pi^s_L=0$ and $\pi^s_H=\pi_F$ for all $s$, where $\pi_F$ is defined in Section 4.1:
\begin{gather*}
    \pi_F := (1-e^{-\lambda t_M})(1 - F(v_M))\lambda v_M(T-t_M).
\end{gather*}
To prove this, we first focus on the highest type $S$ and prove that $\pi^S = 0$ while $\pi^S_H=\pi_F$. In equilibrium, the conditional payoff $(\pi_H^S,\pi_L^S)$ must belong to the set of payoffs of IC-IR mechanisms with the initial belief of $\nu_S$. Thus, by looking at the direction of $\pi_H-\pi_L$ we know
\begin{align*}
    \pi_H^S - \pi_L^S \leq  \max_{M \in \mathcal{M}} ~ & \left \{ (-1) p_U(T) + \mathbb{E} \left[p_{v,\tau}(T)\mathbbm{1}[\tau \le T] + p_U(T)\mathbbm{1}[\tau > T] \mid \theta = H \right] \right \}  \\
    \text{subject to} \quad & 
    \lambda v I(v,t) - p_{v,t}(T) \geq \lambda v I(U,t)  - p_U(T) \qquad \forall (v, t) \tag{IC-U}\\
    & \lambda v I(v,t) - p_{v,t}(T)  \geq \lambda v I(v',t) - p_{v',t}(T) \qquad \forall (v, v', t) \tag{IC-V}\\
    &  \nu_S \mathbb{E}\left[N_T v - p_{v,\tau}(T) \mid \theta = H \right] - (1-\nu_S) p_U(T) \ge 0 \tag{IR-0} 
\notag
\end{align*}
Consequently, we know
\begin{align*}
    & \pi_H^S - \pi_L^S \leq \max ~   \mathbb{E} \left[\Delta_{v,t}(T) \mid \theta = H \right]  \\
   \text{subject to} \quad & 
    \lambda v I(v,t) - \Delta_{v,t}(T) \geq \lambda v I(U,t)  \qquad \forall  v \tag{IC-U}\\
    & \lambda v I(v,t) - \Delta_{v,t}(T)  \geq \lambda v I(v',t) - \Delta_{v',t}(T) \qquad \forall  (v, v') \tag{IC-V}
\end{align*}
According to the proof of Theorem \ref{thm:trial}, we know $\pi_H^S - \pi_L^S \leq \pi_F$.
Therefore, $\pi_H^S \leq \pi_L^S + \pi_F \leq \pi_F$ and if $\pi_H^S=\pi_F$ we must have $\pi_L^S=0$. 

Suppose by contradiction that $\pi_H^S<\pi_F$ so $\pi_L^S<0$. Consider the deviation of type $S$ such that $S$ proposes the Myersonian free trial mechanism with $\epsilon$ negative payment to the buyer, as in Section 4.1. As we have discussed in Section 4.1, $S$ can guarantee that the buyer accepts the mechanism and purchases the post-trial once they receive the positive shock, so the $S$ can guarantee an expected payoff of $\mu_S \pi_F$ regardless of off-path beliefs. This is a profitable deviation which leads to a contradiction.

Now suppose that $\pi_H^S = \pi_F$ and $\pi_L^S=0$. From incentive compatibility, we know $0=\pi_L^1 \geq \pi_L^s \geq \pi_L^S=0$. Since this implies $\pi_L^s=\pi_L^S$, then from incentive compatibility, $\pi^s_H=\pi^S_H=\pi_F$ for any $s>1$. This payoff vector can be achieved by sellers of all types pooling in the Myersonian free trial mechanism.

\textbf{Second,} suppose $\pi^1_L>0$; then type $1$ must be pooling with some other types. Denote $s'$ as the highest type that pools with type 1 and $\pi^1_L=\pi^{s'}_L$, $\pi^1_H=\pi^{s'}_H$. Incentive compatibility implies that  $\pi^1_L=\pi^{s}_L$ and $\pi^1_H=\pi^{s}_H$ for all types $1\leq s \leq s'$. 
The conditional payoff pair $(\pi^1_H,\pi^1_L)$ can be achieved by a trial mechanism $(v_1,t_1,\pi^1_L)$ with $p_1$ as the ex-ante fee, $t_1$ as the trial length, and $v_1$ as the post-trial threshold. This trial mechanism should respect the ex-ante IR with the prior belief of $\nu_1$. 
Recall that trial mechanisms with $(v_M,t_M)$ are all trial mechanisms that maximize $\pi_H-\pi_L$ (and its maximum value is $\pi_F$).

We prove the statement by contradiction. Suppose $\pi^1_H-\pi^1_L<\pi_F$. Without loss of generality, we can assume the post-trial and trial length $(v_1,t_1)$, joint with an ex-ante fee that makes the IR condition (with respect to the initial belief $\nu_1$) binding, achieve the boundary of the feasible payoff set of IC-IR mechanisms. Now consider a trial mechanism $(v_1,t_1-\varepsilon,\pi^1_L)$, where $\varepsilon$ is a sufficiently small positive number. Under this trial mechanism, if the buyer always buys the trial and then buys the post-trial if and only if they learn $\theta=H$ and $v>v_1$, then the resulting conditional payoff pair $(\pi_H,\pi_L^1)$ satisfies $\pi_H>\pi^1_H$ and $\pi_H-\pi^1_L > \pi^1_H- \pi^1_L$. Now take a positive $\delta'$ such that
\begin{gather*}
    \mu_{s'} \pi_H + (1-\mu_{s'}) (\pi^1_L -\delta') = \mu_{s'} \pi^1_H  + (1-\mu_{s'}) \pi^1_L. 
\end{gather*}
Then for any $1\leq s < s'$, 
\begin{gather*}
      \mu_{s} \pi_H + (1-\mu_{s}) (\pi^1_L -\delta') < \mu_{s} \pi^1_H  + (1-\mu_{s}) \pi^1_L. 
\end{gather*}
Thus, we can find a $\delta$ that is slightly smaller than $\delta'$ such that
\begin{gather*}
     \mu_{s'} \pi_H + (1-\mu_{s'}) (\pi^1_L -\delta) > \mu_{s'} \pi^1_H  + (1-\mu_{s'}) \pi^1_L, \\
      \mu_{s} \pi_H + (1-\mu_{s}) (\pi^1_L -\delta) < \mu_{s} \pi^1_H  + (1-\mu_{s}) \pi^1_L.     
\end{gather*}

As we have discussed in the proof of Lemma \ref{lemma:obviouslynotD1}, the buyer has three rational responses to the deviating mechanisms: ${(B,B), (B,N),N}$, where $(B,B)$ implies purchasing the trial and the post-trial service even in the absence of signals, $(B,N)$ implies purchasing the trial but not the post-trial service in the absence of signals, and $N$ implies purchasing nothing. Under $N$, all types of sellers are worse off (compared to equilibrium payoffs); under $(B,B)$, all types $s=1,2,...,s'$ are better off; under $(B,N)$, all types $s=1,2,...,s'-1$ are worse off while type $s'$ is better off. Consequently, $D_1$ criterion implies that, if the seller proposes the off-path trial mechanism $(v_1,t_1-\varepsilon,\pi^1_L-\delta)$,
the buyer's off-path belief would assign no probability to types $s=1,2,...,s'-1$ so the off-path belief would be higher than $\mu_{s'}>\nu_{s'}=\nu_1$. This implies the buyer is willing to accept the off-path trial mechanism and so it is profitable for the type $s'$ to deviate to $(v_1,t_1-\varepsilon,\pi^1_L-\delta)$. This contradiction means $\pi^1_H-\pi^1_L=\pi_F$.

In the last step, we claim that for any other $s$, $\pi^s_H=\pi^1_H$ and $\pi^s_L=\pi^1_L$. If not, then from the monotonicity implied by the incentive compatibility we know $\pi^s_H>\pi^1_H$ and $\pi^s_L\leq \pi^1_L$. However, this means $\pi^s_H-\pi^s_L >\pi^1_H-\pi^1_L=\pi_F$, which is not possible. This completes the discussion of the second case.

Finally, we argue why all such equilibrium payoffs are Pareto dominated. Because all types pool in the same trial mechanism. The ex-ante fee of the trial mechanism must respect the IR condition with the initial belief $\mu_0$. Then, following the analysis in the main model, all such trial mechanisms are dominated by the trial mechanism associated with point $H$.
\end{proof}

\end{document}